\documentclass[abstract,headings=normal,DIV=13]{scrartcl}

\usepackage[automark]{scrpage2}
\usepackage[english]{babel}
\usepackage[T1]{fontenc}
\usepackage[applemac]{inputenc}
\usepackage{graphicx}
\usepackage{tikz}
\usetikzlibrary{shapes,decorations,shadows}
\usepackage{amsmath}
\usepackage{amssymb}
\usepackage{amsthm}
\usepackage{subfig}
\usepackage{accents}
\usepackage{url}
\usepackage[colorlinks=true,linkcolor=blue,citecolor=red,urlcolor=black]{hyperref}

\usetikzlibrary{arrows}

\addtokomafont{caption}{\small}

\DeclareMathOperator{\Ad}{Ad}

\DeclareMathOperator{\tr}{tr}

\makeatletter
\renewcommand*\env@cases[1][1.2]{%
  \let\@ifnextchar\new@ifnextchar
  \left\lbrace
  \def\arraystretch{#1}%
  \array{@{}l@{\quad}l@{}}%
}
\makeatother

\newcommand{\itbf}{\itshape\bfseries}

\newcommand{\g}{\mbox{\boldmath$g$}}

\newcommand{\G}{\mbox{\boldmath$G$}}

\newcommand{\sh}{\mathrm{sinh}}

\newcommand{\cth}{\mathrm{coth}}

\newcommand{\rC}{{\mathrm C}}

\newcommand{\rH}{{\mathrm H}}

\newcommand{\rL}{{\mathrm L}}

\newcommand{\ga}{{\mathfrak a}}
\newcommand{\gb}{{\mathfrak b}}
\newcommand{\gc}{{\mathfrak c}}
\newcommand{\gd}{{\mathfrak d}}

\newcommand{\mbA}{\mbox{\boldmath$A$}}
\newcommand{\mbB}{\mbox{\boldmath$B$}}
\newcommand{\mbC}{\mbox{\boldmath$C$}}

\newcommand{\cD}{{\cal D}}

\newcommand{\cK}{{\cal K}}
\newcommand{\cL}{{\cal L}}

\newcommand{\cP}{{\cal P}}

\newcommand{\cT}{{\cal T}}

\newcommand{\cX}{{\cal X}}

\newcommand{\wa}{\widetilde{a}}
\newcommand{\wb}{\widetilde{b}}

\newcommand{\wip}{\widetilde{p}}

\newcommand{\wx}{\widetilde{x}}

\newcommand{\whx}{\widehat{x}}
\newcommand{\whp}{\widehat{p}}
\newcommand{\htilde}[1]{\widehat{\widetilde{#1}}}
\newcommand{\wL}{\widetilde{L}}

\newcommand{\wT}{\widetilde{T}}
\newcommand{\wU}{\widetilde{U}}
\newcommand{\wV}{\widetilde{V}}

\newcommand{\bbR}{{\mathbb R}}
\newcommand{\bbZ}{{\mathbb Z}}

\newtheorem{theorem}{Theorem}[section]
\newtheorem{proposition}[theorem]{Proposition}
\newtheorem{corollary}[theorem]{Corollary}
\newtheorem{lemma}[theorem]{Lemma}
\newtheorem{definition}[theorem]{Definition}

\newcommand{\eto}[1]{e^{\displaystyle #1}}
\newcommand{\nm}{\!-\!}
\newcommand{\np}{\!+\!}
\newcommand{\ueto}[1]{e^{\raisebox{2pt}{$\displaystyle #1$}}}

\newbox\meibox
\def\placeunder#1#2#3#4{\setbox\meibox%
\vbox{\hbox{\hskip#4$\hphantom{#2}$}\hbox{$\hphantom{#1}$}}%
\vtop{\baselineskip=0pt\lineskiplimit=\baselineskip%
\lineskip=#3\hbox to \wd\meibox{\hfil\hskip#4$#2$\hfil}%
\hbox to \wd\meibox{\hfil$#1$\hfil}}}
\def\undertilde#1{\mathchoice{%
\placeunder{\vbox to 1.4pt{\hbox{$\displaystyle\widetilde{\,\,\,
}$}\vss}}{\displaystyle#1}{1.5pt}{1.5pt}}%
{\placeunder{\vbox to 1.4pt{\hbox{$\textstyle\widetilde{\,\,
}$}\vss}}{\textstyle#1}{1.5pt}{1.5pt}}%
{\placeunder{\vbox to 1.4pt{\hbox{$\scriptstyle\tilde{
}$}\vss}}{\scriptstyle#1}{1pt}{1pt}}%
{\placeunder{\vbox to 1.4pt{\hbox{$\scriptscriptstyle\tilde{
}$}\vss}}{\scriptscriptstyle#1}{1pt}{1pt}}%
}

\graphicspath{{./graphics/}}
\DeclareGraphicsExtensions{.pdf}

\begin{document}
%\allowdisplaybreaks

\title{Discrete time Toda systems}
\author{Yuri~B.~Suris}
\publishers{\vspace{0.5cm}{\small Institut f\"ur Mathematik, MA 7-1, Technische Universit\"at Berlin,\\
Stra{\ss}e des 17. Juni 136, 10623 Berlin, Germany\\
E-mail:  \url{suris@math.tu-berlin.de}}}
\maketitle
\begin{abstract}
In this paper, we discuss several concepts of the modern theory of discrete integrable systems, including:
\begin{itemize}
\item Time discretization based on the notion of B\"acklund transformation;
\item Symplectic realizations of multi-Hamiltonian structures;
\item Interrelations between discrete 1D systems and lattice 2D systems;
\item Multi-dimensional consistency as integrability of discrete systems;
\item Interrelations between integrable systems of quad-equations and integrable systems of Laplace type;
\item Pluri-Lagrangian structure as integrability of discrete variational systems. 
\end{itemize}
All these concepts are illustrated by the discrete time Toda lattices and their relativistic analogs. 
\end{abstract}

\setcounter{equation}{0}
\section{Introduction}

The one-dimensional lattice with exponential interaction of nearest neighbors, discovered by M. Toda,
\begin{equation}\label{TL New introd}
\ddot{x}_k=\eto{x_{k+1}\nm x_k}-\eto{x_k\nm x_{k-1}},
\end{equation}
and its relativistic generalization, discovered by S. Ruijsenaars,
\begin{equation}\label{RTL+ New introd}
\ddot{x}_k=
(1+\alpha\dot{x}_{k+1})(1+\alpha\dot{x}_k)\,\frac{\eto{x_{k+1}\nm x_k}} {\,\raisebox{-1mm}{$1+\alpha^2\eto{x_{k+1}\nm x_k}$}}-
(1+\alpha\dot{x}_k)(1+\alpha\dot{x}_{k-1})\,\frac{\eto{x_k\nm x_{k-1}}} {\,\raisebox{-1mm}{$1+\alpha^2\eto{x_k\nm x_{k-1}}$}},
\end{equation}
belong to the most celebrated integrable models. They enjoy a great amount of generalizations and applications in various branches of mathematics and physics. This paper reviews a variety of generalized Toda lattices and relativistic Toda lattices, along with their integrable discretizations. This gives us an opportunity to touch upon some of the most important recent developments of the theory of discrete integrable systems, including the multi-dimensional consistency and pluri-Lagrangian structure. The paper is organized as follows. 

In Section \ref{Sect TL} we quickly review the main integrability attributes of the Toda lattice in the Flaschka-Manakov variables. It is one of the basic systems amenable to the Adler-Kostant-Symes scheme, which is presented in Section \ref{sect: AKS}. A recipe for integrable discretization of the systems within the AKS scheme is formulated in Section \ref{Sect recipe}. It is applied to the Toda lattice in the Flaschka-Manakov variables in Section \ref{Sect discretization TL}. Then in Section \ref{Sect Toda linear Newtonian} these results are applied to a symplectic realization of the linear Poisson brackets for the Flaschka-Manakov variables, which leads to the most classical exponential Toda lattice \eqref{TL New introd} and its time discretization. A variety of relatives of this system, which appear through symplectic realization of different Poisson brackets for the Flaschka-Manakov variables, together with their time discretizations, are treated in Section \ref{sect Toda New}. After that, a similar work is done for the relativistic Toda lattice: in Section \ref{Sect RTL tri-Ham}, the main integrability attributes in the Flaschka-Manakov variables are reviewed, the discretization in these variables is performed in Section \ref{Sect RTL param discretization}, and various symplectic realizations are presented in Sections \ref{Sect additive exp rel Toda}, \ref{Sect Newtonian rel Toda}. An interesting phenomenon is investigated in Section \ref{sect explicit}: it turns out that explicit discretiaztions of the Toda lattice belong to the relativistic Toda hierarchy, the discrete time step playing the role of the inverse speed of light. In Section \ref{sect: discr Toda}, we address an important conceptual twist: time discretizations of 1D evolutionary equations are re-interpreted as lattice 2D systems. In the second half of the paper, we deal with recent conceptual breakthroughs in the theory of discrete integrable systems. In Section \ref{Sect: dToda}, we discuss the relation of discrete Laplace type equations to quad-equations, and the notion of multi-dimensional consistency of quad-equations as their integrability. This development allows one to derive, in an algorithmic way, zero curvature representations for discrete relativistic Toda type systems, as demonstrated in Section \ref{Sect main}. Then, we turn to the pluri-Lagrangian theory, which describes integrability features of variational systems. The 1D pluri-Lagrangian theory is formulated in Section \ref{sect: discr 1d results} and is illustrated by the discrete time exponential Toda lattice in Section \ref{sect: BT Toda}. The 2D pluri-Lagrangian theory is formulated in Section \ref{sect: discr 2d results} and is illustrated by the discrete time relativistic Toda lattice in Section \ref{sect: from pluri to rel}.

Bibliographic references are given at the end of each section; they are kept to a necessary minimum and therefore are by no means exhaustive. We hope, however, that they will enable an interested reader to get oriented in the relevant literature.

\setcounter{equation}{0}
\section[Toda lattice]{Toda lattice in Flaschka variables: \\
\quad equations of motion, Lax representation and tri-Hamiltonian structure}
\label{Sect TL}

Equations of motion of the Toda lattice (TL) in Flaschka-Manakov variables:
\begin{equation}\label{TL}
\dot{b}_k=a_k-a_{k-1},\quad \dot{a}_k=a_k(b_{k+1}-b_k),
\qquad 1\le k\le N,
\end{equation}
with one of two types of boundary conditions: open-end
($a_0=a_N=0$), or periodic (all subscripts are taken (mod $N$), so
that $a_0\equiv a_N$, $b_{N+1}\equiv b_1$). 

The Lax representation of the TL which we mainly use in this paper is:
\begin{equation}\label{TL Lax}
\dot{T}=[\,T,A_+]=-[\,T,A_-]
\end{equation}
with
\begin{equation}\label{TL T}
T(a,b,\lambda)  =\lambda^{-1}\sum_{k=1}^{N} a_k E_{k,k+1}+\sum_{k=1}^N b_k E_{kk}
+\lambda\sum_{k=1}^{N} E_{k+1,k},
\end{equation}
\begin{equation}\label{TL B}
A_+(a,b,\lambda) = \sum_{k=1}^N b_k E_{kk}+\lambda\sum_{k=1}^{N} E_{k+1,k},
\qquad
A_-(a,b,\lambda)  =  \lambda^{-1}\sum_{k=1}^{N} a_k E_{k,k+1}.
\end{equation}
Here and below $E_{jk}$ stands for the matrix whose only nonzero entry is on 
the intersection of the $j$th row and the $k$th column and is equal to 1. 
Naturally, we set in the periodic case $E_{N+1,N}=E_{1,N}$, 
$E_{N,N+1}=E_{N,1}$. In the open--end case we set $E_{N+1,N}=E_{N,N+1}=0$ 
and always put $\lambda=1$.

In the notation which will be explained in the next section, we have:
$A_+=\pi_+(T)$, $A_-=\pi_-(T)$, so that \eqref{TL Lax} takes the form
\begin{equation}\label{Tdot spec}
\dot{T}=\big[\,T,\pi_+(T)\,\big]=-\big[\,T,\pi_-(T)\,\big],
\end{equation}
which is a (prototypical) example of the systems eligible to the AKS (Adler-Kostant-Symes) scheme. Spectral invariants of the matrix $T$ are integrals of motion of TL; there are $N$ functionally independent ones.

The phase space of TL can be equipped with three local Poisson brackets which are preserved by this flow which is therefore Hamiltonian with respect to any of these Poisson structures. These brackets are compatible in the sense that their linear combinations are invariant Poisson brackets, as well. 
\paragraph{Linear bracket:}
\begin{equation}\label{TL l br}
\{b_k,a_k\}_1=-a_k, \quad \{a_k,b_{k+1}\}_1=-a_k.
\end{equation}
The Hamilton function of TL in this bracket is 
\begin{equation}\label{TL H2}
\rH_2(a,b)=\frac{1}{2}{\tr}(T^2)=\frac{1}{2}\sum_{k=1}^N b_k^2+\sum_{k=1}^{N}a_k.
\end{equation}

\paragraph{Quadratic bracket:}
\begin{equation}\label{TL q br}
\begin{array}{cclcccl}
\{b_k,a_k\}_2 & = & -b_k a_k, & \; &
\{a_k,b_{k+1}\}_2 & = & -a_k b_{k+1}, \\ 
\{b_k,b_{k+1}\}_2 & = & -a_k, & \; &
\{a_k,a_{k+1}\}_2 & = & -a_k a_{k+1}.
\end{array}
\end{equation}
The Hamilton function of TL in this bracket is 
\begin{equation}\label{TL H1}
\rH_1(a,b)={\tr}(T)=\sum_{k=1}^N b_k.
\end{equation}

\paragraph{Cubic bracket:}
\begin{equation}\label{TL c br}
\begin{array}{cclcccl}
\{b_k,a_k\}_3     & = & -a_k (b_k^2+a_k),    & \; &
\{a_k,b_{k+1}\}_3 & = & -a_k (b_{k+1}^2+a_k), \\ 
\{b_k,b_{k+1}\}_3 & = & -a_k (b_k+b_{k+1}), & \; &
\{a_k,a_{k+1}\}_3 & = & -2a_k a_{k+1}b_{k+1},   \\ 
\{b_k,a_{k+1}\}_3 & = & -a_k a_{k+1},         & \; &
\{a_k,b_{k+2}\}_3 & = & -a_k a_{k+1}.
\end{array}\end{equation}
The Hamilton function of TL in this bracket is 
\begin{equation}\label{TL H0}
\rH_0(a,b)={\tr}(\log T)=\log(\det T).
\end{equation}

All spectral invariants of the Lax matrix $T$ (including $\rH_0$, $\rH_1$, $\rH_2$) are in involution with respect to any of these three brackets.

\paragraph{Bibliographical remarks.}

An integrable lattice with an exponential interaction of nearest neighbors
was discovered by Toda \cite{To67}. Since then it became one of the most popular
and important integrable models in general. The best general reference
remains Toda's monograph \cite{To89}, where also the story of the discovery of
this system is given from the first hand. The change of variables
$(x,p)\mapsto(a,b)$ was introduced in \cite{Fl74}, \cite{Ma74}, along with 
the Lax representation. Concerning the tri-Hamiltonian structure of the Toda lattice: the linear bracket was known from the very
beginning, since it is a simple consequence of the
Flaschka-Manakov change of variables; the quadratic and the cubic ones appeared in
\cite{A79} and in \cite{Ku85}, respectively.

\setcounter{equation}{0}
\section{Adler-Kostant-Symes scheme}
\label{sect: AKS}

Let $\g$ be a Lie algebra of some associative algebra, equipped with a non-degenerate
bi-invariant scalar product, which allows us to identify $\g^*$ with $\g$. Let $\g$, as a linear space, be a
direct sum of its two subspaces $\g_{\pm}$ which are also Lie subalgebras:
\begin{equation}\label{splitting}
\g=\g_+\oplus \g_-, \quad [\g_+,\g_+]\subset \g_+, \quad [\g_-,\g_-]\subset \g_-.
\end{equation}
Let $\pi_+$, $\pi_-$ denote the projections from $\g$ to the
corresponding subspaces, so that for any $T\in \g$ we have:
\begin{equation}\label{pi+pi-}
T=\pi_+(T)+\pi_-(T)\;,\quad \pi_{\pm}(T)\in \g_{\pm}.
\end{equation}
Then the AKS scheme deals with explicit solutions and Hamiltonian structures of the flows 
\begin{equation}\label{Tdot}
\dot{T}=\big[\,T,\pi_+(f(T))\,\big]=-\big[\,T,\pi_-(f(T))\,\big],
\end{equation}
with $\Ad$-covariant functions $f:\g\to \g$.

This general setting is sufficient to ensure that equations
(\ref{Tdot}) possess several remarkable properties. First of
all, different flows of the type (\ref{Tdot}) commute. Second,
they admit an explicit solution in terms of a factorization
problem in a Lie group. Both these properties have a purely
kinematical nature and do not depend on the Hamiltonian theory
(though the latter provides a deeper insight into the situation).

Let $\G$ be a Lie group with the Lie algebra $\g$, and let $\G_+$
and $\G_-$ be its two subgroups having $\g_+$ and $\g_-$,
respectively, as Lie algebras. Then in a certain neighborhood
$V$ of the group unit $I$  following factorization is
uniquely defined, so that for any $g\in V\subset \G$ we have:
 \index{factorization problem in a Lie group}
\begin{equation}\label{Pi+Pi-}
g=\Pi_+(g)\,\Pi_-(g),\quad \Pi_{\pm}(g)\in \G_{\pm}.
\end{equation}
In what follows we suppose, for the sake of notational simplicity,
that $\G$ is a matrix group, and write the adjoint action of the
group on $\g$ as a conjugation by the corresponding matrices. Correspondingly, we shall call ${\mathrm
Ad}$--covariant functions $f:\,\g\to \g$ also ``conjugation
covariant''. This notation has an additional advantage of being
applicable also to functions $F:\,\g\to \G$. Namely, we shall
call such a function conjugation covariant, if $F({\mathrm Ad}\,
g\cdot T)=gF(T)g^{-1}$.

\begin{theorem} \label{split kinematics}
Let $f:\g\to \g$ be a conjugation covariant function. Then
the solution of the differential equation \eqref{Tdot}
with the initial condition $T(0)=T_0$ is given, at least for sufficiently small $t$, by
\begin{equation}\label{solution cont}
T(t) = \Pi_+^{-1}\Big(\eto{tf(T_0)}\Big)\cdot T_0\cdot \Pi_+\Big(\eto{tf(T_0)}\Big)
= \Pi_-\Big(\eto{tf(T_0)}\Big)\cdot T_0\cdot \Pi_-^{-1}\Big(\eto{tf(T_0)}\Big).
\end{equation}
\end{theorem}

\begin{definition}\label{def Backlund}
Consider the hierarchy of flows \eqref{Tdot} on $\g$.
For an arbitrary conjugation covariant function $F:\g\to \G$ define the
{\itbf B\"acklund transformation} $\mathrm{BT}_F:\g\to \g$ of this
hierarchy as
\begin{equation}\label{BT}
\wT=\mathrm{BT}_F(T)=\Pi_+^{-1}\left(F(T)\right)\cdot T\cdot\Pi_+\left(F(T)\right)
= \Pi_-\left(F(T)\right)\cdot T\cdot\Pi_-^{-1}\left(F(T)\right)\,.
\end{equation}
\end{definition}

One of the most important properties of B\"acklund
transformations, implying also other ones, is contained in the
following theorem.

\begin{theorem}\label{superposition of BTs}
For two arbitrary conjugation covariant functions $F_1,F_2:\g\to\G$,
\begin{equation}\label{BT superp}
\mathrm{BT}_{F_2}\circ\mathrm{BT}_{F_1}=\mathrm{BT}_{F_2F_1},
\end{equation}
so that the B\"acklund transformations $\mathrm{BT}_{F_1}$, $\mathrm{BT}_{F_2}$
commute.
\end{theorem}

\begin{corollary} $\mathrm a)$ Any two flows of the type \eqref{Tdot} commute.

$\mathrm b)$ An arbitrary B\"acklund
transformation commutes with an arbitrary flow of the hierarchy
\eqref{Tdot}. In other words, any B\"acklund transformation
maps solutions of \eqref{Tdot} onto solutions.

\end{corollary}

\noindent Indeed, according to Theorem \ref{split kinematics} any
flow governed by a differential equation (\ref{Tdot}) consists of
B\"acklund transformations $\mathrm{BT}_F$ with $F(T)=\eto{tf(T)}$.

As another important consequence of Theorem \ref{superposition of BTs} we have the following statement.

\begin{theorem} \label{discr split kinematics}
Let $F:\g\to \G$ be a conjugation covariant function. Consider
the formula \eqref{BT} for the B\"acklund transformation
${\mathrm BT}_F$ as the difference equation
\begin{equation}\label{Tn+1}
\wT= \Pi_+^{-1}\left(F(T)\right)\cdot T\cdot\Pi_+\left(F(T)\right)
= \Pi_-\left(F(T)\right)\cdot T\cdot\Pi_-^{-1}\left(F(T)\right)
\end{equation}
for $T=T(n)$, $\wT=T(n+1)$, with the initial condition $T(0)=T_0$.
Then the solution of this difference equation is given by
\begin{equation}\label{solution discr}
T(n) = \Pi_+^{-1}\left(F^n(T_0)\right)\cdot T_0\cdot\Pi_+\left(F^n(T_0)\right)
= \Pi_-\left(F^n(T_0)\right)\cdot T_0\cdot\Pi_-^{-1}\left(F^n(T_0)\right).
\end{equation}
\end{theorem}
\begin{proof} By induction from Theorem \ref{superposition of BTs},
\[
\mathrm{BT}_{F_n}\circ\ldots\circ\mathrm{BT}_{F_2}\circ\mathrm{BT}_{F_1}=
\mathrm{BT}_{F_n\ldots F_2F_1}.
\]
In particular, for $F_1=F_2=\ldots=F_n=F$ we obtain:
\[
(\mathrm{BT}_{F})^n=\mathrm{BT}_{F^n},
\]
which is the statement of the Theorem. 
\end{proof}

This Theorem gives a discrete-time counterpart of Theorem
\ref{split kinematics}. 
Comparing the formulas (\ref{solution discr}), (\ref{solution
cont}), we see that the map (\ref{Tn+1}) is interpolated by the
flow (\ref{Tdot}) with the time step $h$, if
\[
\eto{hf(T)}=F(T)\quad\Leftrightarrow\quad f(T)=h^{-1}\log(F(T)).
\]

For the flow TL, the main ingredients of the AKS construction are as follows.

\paragraph{Open-end case.}
For the {\em open-end case} we set $\g=\mathrm{gl}(N)$, the algebra of
$N\times N$ matrices with the usual matrix product, the Lie bracket
$[L,M]=LM-ML$, and the non-degenerate bi-invariant scalar product
$\langle L,M\rangle={\tr}(LM)$ which allows to identify $\g^*$ with $\g$.  We have a splitting (\ref{splitting}), where 
$\g_+$ consists of lower triangular matrices, while $\g_-$ consists of strictly upper triangular matrices.
The Lie group $\G$ corresponding to the Lie algebra $\g$ is $\mathrm{GL}(N)$,
the group of $N\times N$ nondegenerate matrices. The subgroups $\G_+$, $\G_-$
corresponding to the Lie algebras $\g_+$, $\g_-$ consist of non-degenerate
lower triangular matrices and of upper triangular matrices with unit diagonal,
respectively. The $\Pi_+\Pi_-$ factorization is well known in the linear
algebra under the name of the {\em LU factorization}.

\paragraph{Periodic case.}
In the {\em periodic case} we set $\g$ to be a certain {\em twisted 
loop algebra} over $\mathrm{gl}(N)$:
\[
\g=\left\{T(\lambda)\in \mathrm {gl}(N)[\lambda,\lambda^{-1}]:
\Omega T(\lambda)\Omega^{-1}=T(\omega\lambda)\right\},
\]
where $\Omega=\mathrm{diag}(1,\omega,\ldots,\omega^{N-1})$,
$\omega=\exp(2\pi i/N)$. The nondegenerate bi--invariant scalar product is chosen as
$\langle L(\lambda),\,M(\lambda)\rangle={\tr}(L(\lambda)M(\lambda))_0$,
the subscript 0 denoting the free term of the formal Laurent series. Again, we have a splitting (\ref{splitting}), where
\begin{align*}
\g_+ &=\left\{T(\lambda)\in \mathrm{gl}(N)[\lambda]:
\Omega T(\lambda)\Omega^{-1}=L(\omega\lambda)\right\},\\
\g_- &=\left\{T(\lambda)\in \mathrm{gl}(N)[\lambda^{-1}]:
\Omega T(\lambda)\Omega^{-1}=T(\omega\lambda)\;{\mathrm and}\; T(\infty)=0\right\}.
\end{align*}

The group $\G$ corresponding to the Lie algebra $\g$ is a {\em twisted loop
group}, 
\[
\G=\left\{g:  {\mathbb C}P^1\backslash\{0,\infty\}\to  \mathrm{GL}(N): g\; \mathrm{regular}, \;  
\Omega g(\lambda)\Omega^{-1}=g(\omega\lambda)\right\}.
\]
Its subgroups
$\G_+$ and $\G_-$ corresponding to the Lie algebras $\g_+$ and $\g_-$,
are singled out by the following conditions:
\begin{align*}
\G_+ &=\left\{g:  {\mathbb C}P^1\backslash\{\infty\}\to  \mathrm{GL}(N): g\; \mathrm{regular}, \;  
\Omega g(\lambda)\Omega^{-1}=g(\omega\lambda)\right\},\\
\G_- &=\left\{g:  {\mathbb C}P^1\backslash\{0\}\to  \mathrm{GL}(N): g\; \mathrm{regular}, \;  
\Omega g(\lambda)\Omega^{-1}=g(\omega\lambda)\;\mathrm{and}\; g(\infty)=I\right\}.
\end{align*}
We call the corresponding $\Pi_+\Pi_-$ factorization the {\em generalized
LU factorization}. It is uniquely defined in a certain neighborhood of
the unit element of $\G$. As opposed to the open-end case, finding the
generalized $LU$ factorization is a problem of the Riemann-Hilbert type
which is solved in terms of algebraic geometry rather than in terms of linear
algebra. 
\smallskip

In both cases, open-end and periodic, one has $f(T)=T$.

\paragraph{Bibliographical remarks.}  The foundational references for the AKS-scheme are \cite{A79},
\cite{Ko79}, \cite{Sy80, Sy82}.

\setcounter{equation}{0}
\section{Recipe for integrable discretization}
\label{Sect recipe}

The results of the previous Section suggest the following prescription.
\medskip

\noindent
\textbf{Recipe.} {\em For an integrable system allowing a Lax representation of
the form \eqref{Tdot},
an integrable discretization is given by the difference equation \eqref{BT}, with some conjugation covariant function
$F:\g\to \G$ such that $F(T)=I+hf(T)+o(h)$ (B\"acklund
transformation close to identity)}.
\medskip

Of course, this prescription is only practical  if the
corresponding factors $\Pi_{\pm}(F(T))$ admit more or less explicit expressions,
allowing to write down the corresponding difference equations in a
more or less closed form. The choice of $F(T)$ is a transcendent
problem. Miraculously, the simplest possible choice
$F(T)=I+hf(T)$ works perfectly well for a vast set of examples
(when it makes sense, i.e., when $I+hf(T)\in \G$), including the Toda lattice.

Let us stress once more the advantages of this approach to the
problem of integrable discretization.
\begin{itemize}
\item The discretizations obtained in this way share the Lax matrix
and the integrals of motion with their underlying continuous time
systems.
\item Suppose that the hierarchy of continuous time systems
(\ref{Tdot}) is
Hamiltonian with respect to some Poisson bracket. Then our
discretizations have the Poisson property with respect to the same
bracket. 
\item The initial value problem for our discrete time equations
can be solved in terms of the same factorization in a Lie group as
the initial value problem for the continuous time system.
\item Interpolating Hamiltonians are granted by-products of this approach.
\end{itemize}

\paragraph{Bibliographical remarcs.} This recipe for integrable discretization was clearly formulated for the first time in 
\cite{Su95, Su96}, and was put at the basis of a monographic study \cite{Su03}. However, a viewpoint according to which 
B\"acklund transformations lie at the basis of discretization was already pushed forward in \cite{Le81}. This is also well established in
discrete differential geometry, cf. \cite{Bo99}, \cite{DSM00}, \cite{BS08}.

\setcounter{equation}{0}
\section{Discretization of the Toda lattice in the Flaschka-Manakov variables}
\label{Sect discretization TL}

We now turn to the problem of finding an integrable time discretization
for the flow TL. To this purpose we apply the recipe of Section \ref{Sect recipe}
with
\[
F(T)=I+hT,
\]
i.e., we take as a discretization of the flow TL the map described by the
discrete time Lax equation
\[
\wT=\Pi_+^{-1}(I+hT)\cdot T\cdot\Pi_+(I+hT)=
\Pi_-(I+hT)\cdot T\cdot\Pi_-^{-1}(I+hT).
\]
Thus, the main problem is to determine the factors
\begin{equation}\label{dTL factors}
\mbA_+=\Pi_+(I+hT),\quad \mbA_-=\Pi_-(I+hT).
\end{equation}
\begin{lemma}\label{discr TL factor}
For the matrix $T=T(a,b,\lambda)$ of the form \eqref{TL T}, the factors \eqref{dTL factors}
are of the form
\begin{align}
\mbA_+ & =  \sum_{k=1}^N\beta_kE_{kk}+h\lambda\sum_{k=1}^{N}E_{k+1,k},
\label{dTL Pi+}\\
\mbA_- & =  I+h\lambda^{-1}\sum_{k=1}^{N}\alpha_kE_{k,k+1}.
\label{dTL Pi-}
\end{align}
Here, the coefficients $\beta_k=\beta_k(a,b)$ are defined by the relations 
\begin{equation}\label{dTL beta}
\beta_k=1+hb_k-\frac{h^2a_{k-1}}{\beta_{k-1}},
\end{equation} 
and
\begin{equation}\label{dTL alpha}
\alpha_k=\frac{a_k}{\beta_k}\;.
\end{equation}
\end{lemma}
\noindent
Indeed, it is easy to realize that the factors $\mbA_{\pm}$ must be of the form \eqref{dTL Pi+}, \eqref{dTL Pi-}. The factorization
$\mbA_+\mbA_-=I+hT$ is equivalent to the system
\begin{equation}\label{dTL aux}
\beta_k+h^2\alpha_{k-1}=1+hb_k,\quad \beta_k\alpha_k=a_k,
\end{equation}
which, in turn, is equivalent to \eqref{dTL beta}, \eqref{dTL alpha}. In the open-end case, due to $a_0=0$, relation (\ref{dTL beta}) is uniquely solvable, and leads to
explicit expressions in terms of finite continued fractions:
\[
\beta_1=1+hb_1;\quad
\beta_2=1+hb_2-\frac{h^2a_1}{1+hb_1};\quad\cdots\quad;
\]
\[
\beta_N=1+hb_N-\frac{h^2a_{N-1}}{1+hb_{N-1}-
\displaystyle\frac{h^2a_{N-2}}{1+hb_{N-2}-\;
\raisebox{-3mm}{$\ddots$}
\raisebox{-4.5mm}{$\;-\displaystyle\frac{h^2a_1}{1+hb_1}$}}}.
\]
In the periodic case $\beta_k$ may be expressed as analogous infinite
$N$-periodic continued fractions and are, therefore, {\em double-valued}
functions of $(a,b)$. However, in the limit $h\to 0$ one branch of $\beta_k$
can be singled out by the asymptotics 
\begin{equation}\label{dTL beta as}
\beta_k(a,b)=1+hb_k+O(h^2).
\end{equation}

\begin{theorem}\label{discrete TL}
The discrete time Lax equation
\begin{equation}\label{dTL Lax}
\wT=\mbA_+^{-1}T\mbA_+=\mbA_- T\mbA_-^{-1} \quad with \quad
\mbA_+=\Pi_+(I+hT),\;\;\mbA_-=\Pi_-(I+hT)
\end{equation}
is equivalent to the map $(a,b)\mapsto(\wa,\wb)$ described by the following
equations:
\begin{equation}\label{dTL}
\wb_k=b_k+h\left(\frac{a_k}{\beta_k}-\frac{a_{k-1}}{\beta_{k-1}}\right),
\quad    \wa_k=a_k\;\frac{\beta_{k+1}}{\beta_k},
\end{equation}
where the functions $\beta_k$ are defined by
the recurrent relation \eqref{dTL beta}.
\end{theorem}
The map (\ref{dTL}) will be of a fundamental interest to us in this paper.
We denote it by dTL$(h)$; it is a genuine map in the open-end case, while in the periodic case it is a double-valued map (a correspondence). The following statements
automatically follow from our construction.
\begin{itemize}
\item The map dTL$(h)$ is Poisson with respect to any invariant Poisson bracket of the flow TL.
\item The map dTL$(h)$ commutes with all flows of the TL hierarchy \eqref{Tdot}. 
\item The maps dTL$(h)$ with different $h$ commute among themselves. (However, the notion of commutativity of double-valued maps is non-trivial, see Section \ref{sect: BT Toda}.)
\item The map dTL$(h)$ is interpolated by the flow \eqref{Tdot} with $f(T)=h^{-1}\log(I+ hT)$.
\end{itemize}
For most of the time, the parameter $h$ will be suppressed from the notation.

\paragraph{Bibliographical remarks.} B\"acklund-Darboux transformation for the Toda lattice in Flaschka-Manakov variables, which essentially coincides with the map given in Theorem \ref{discrete TL}, was given for the first time in \cite{MS79}. However, as observed in \cite{Su95}, it is not different from  the $qd$ algorithm well known in the numerical analysis for a long time \cite{Ru57}. Moreover, the latter reference contains also  the equations of motion of the Toda lattice (under the name of a ``continuous analogue of the $qd$ algorithm'')!

\setcounter{equation}{0}
\section[Exponential Toda lattice]
{Symplectic realization of the linear bracket: exponential  Toda lattice}
\label{Sect Toda linear Newtonian}

A great variety of canonical Hamiltonian systems and their equivalent Lagrangian systems arise from TL upon parametrizing various invariant Poisson brackets via the canonical symplectic brackets. A map from ${\mathbb R}^{2N}(x,p)$ equipped with the canonical symplectic structure to a Poisson manifold $(\cP,\{\cdot,\cdot\})$ is called a {\em symplectic realization} of the bracket $\{\cdot,\cdot\}$, if it is Poisson. In particular, the most classical (exponential) form of TL (the original discovery by Toda) appears this way via the Flaschka-Manakov map, 
\begin{equation}\label{TL l par}
a_k=\eto{x_{k+1}\nm x_k},\quad b_k=p_k.
\end{equation}
We consider this map for two types of the boundary 
conditions. For the periodic case we assume that $x_0=x_N$, $x_{N+1}=x_1$,
while for the open-end case we set $x_0=\infty$, $x_{N+1}=-\infty$, which
corresponds to $a_N=0$. 
\begin{proposition} \label{Toda linear parametrization}
The map \eqref{TL l par} is a symplectic realization of the linear Poisson bracket $\{\cdot,\cdot\}_1$, see \eqref{TL l br}, on the phase space of $\mathrm{TL}$. 
\end{proposition}

\begin{theorem} \label{Toda in linear parametrization}
Pull-back of the flow $\mathrm{TL}$ under parametrization 
\eqref{TL l par} is a canonical Hamiltonian system with the Hamilton and the Lagrange functions
\begin{align}
\rH(x,p)= & \ \frac{1}{2}\sum_{k=1}^N p_k^2+\sum_{k=1}^N \eto{x_{k+1}\nm x_k}, 
\label{TL l H in xp}\\
\rL(x,\dot{x})= & \ \frac{1}{2}\sum_{k=1}^N \dot{x}_k^2-\sum_{k=1}^N \eto{x_{k+1}\nm x_k}.
\label{TL l Lagr}
\end{align}
The corresponding Newtonian equations of motion:
\begin{equation}\label{TL l New}
\ddot{x}_k=\eto{x_{k+1}\nm x_k}-\eto{x_k\nm x_{k-1}}.
\end{equation} 
\end{theorem}
\begin{proof}
Hamilton function \eqref{TL l H in xp} is obtained by substituting \eqref{TL l par} into \eqref{TL H2}. Lagrange function \eqref{TL l Lagr} is obtained from $\rH(x,p)$ via the Legendre transformation.
\end{proof}

For the map dTL, the machinery of Hamiltonian flows is no more available, but it
may be successfully replaced by a direct analysis of equations of motion.

\begin{theorem} \label{discr Toda in linear parametrization}
Pull-back of the map $\mathrm{dTL}$ under parametrization 
\eqref{TL l par} is a symplectic map $(x,p)\mapsto(\wx,\wip)$ with the following equations of motion:
\begin{equation}\label{dTL l}
\left\{\begin{array}{l}
p_k = \dfrac{1}{h}\Big(\eto{\wx_k\nm x_k}-1\Big)+h\eto{x_k\nm\wx_{k-1}}, \vspace{2mm}\\
\wip_k  =  \dfrac{1}{h}\Big(\eto{\wx_k\nm x_k}-1\Big)+h\eto{x_{k+1}\nm\wx_k}.
\end{array}\right.
\end{equation}
The corresponding Newtonian equations of motion: 
\begin{equation}\label{dTL l New}
\ueto{\wx_k\nm x_k}-\ueto{x_k\nm \undertilde{x}_k}=
h^2\Big(\ueto{\undertilde{x}_{k+1}-x_k}-\ueto{x_k\nm \wx_{k-1}}
\Big).
\end{equation}
\end{theorem}
\begin{proof} Under parametrization \eqref{TL l par} equations of motion
(\ref{dTL}) of dTL together with recurrent relation (\ref{dTL beta}) for the
auxiliary quantities $\beta_k$ take the following form:
\begin{align}
&\eto{\wx_{k+1}\nm \wx_k} = \eto{x_{k+1}\nm x_k}\;\frac{\beta_{k+1}}{\beta_k},
\label{dTL l proof aux1}\\ 
&\wip_{k} = p_k+\frac{h\eto{x_{k+1}\nm x_k}}{\beta_k}-
\frac{h\eto{x_k\nm x_{k-1}}}{\beta_{k-1}},
\label{dTL l proof aux2}\\ 
&\beta_k = 1+hp_k-\frac{h^2\eto{x_k\nm x_{k-1}}}{\beta_{k-1}}.
\label{dTL l proof aux3}
\end{align}
Equation (\ref{dTL l proof aux1}) implies that the quantity
$\beta_k/\eto{\wx_k\nm x_k}$
is constant, i.e., does not depend on $k$. Choosing this constant is equivalent
to choosing one representative among all possible pull--backs of the map
dTL. We set this constant equal to 1, so that
\begin{equation}\label{dTL l beta}
\beta_k=\eto{\wx_k\nm x_k}.
\end{equation}
Substituting this into (\ref{dTL l proof aux3}), we obtain the first
equation of motion in (\ref{dTL l}). Finally, the second equation of motion in
(\ref{dTL l}) follows from (\ref{dTL l proof aux2}) by using the previously
obtained expressions. Newtonian equations of motion (\ref{dTL l New}) follow 
directly from comparing both equation in (\ref{dTL l}). 
\end{proof}

We will say that the Hamiltonian system with the Hamilton function \eqref{TL l H in xp} and the map \eqref{dTL l} are symplectic realizations of the flow TL and of the map dTL, respectively. The latter map can be named {\em discrete exponential Toda lattice}. 

\paragraph{Bibliographical remarks.} The discrete time Toda lattice 
(\ref{dTL l New}), in the Lagrangian form (\ref{dTL l}), 
was found in \cite{WT75, TW75} as a B\"acklund transformation for the Toda lattice. The Lagrangian function
appeared as a generating function of the B\"acklund transformation, which
demonstrated also the symplectic nature of this transformation. Notice that
in the infinite lattice situation this B\"acklund transformation is, 
generically, no longer isospectral; rather, it adds one soliton to the 
solution.

\setcounter{equation}{0}
\section{The variety of symplectic realizations of TL and dTL}
\label{sect Toda New}

In this section, we will provide the reader with a list of different symplectic realizations of invariant Poisson brackets on the phase space of TL, as well as of the flow TL and of the map dTL. 

\paragraph{General form of equations of motion.} 
All symplectic realizations of TL share the following general form:
\begin{equation}\label{Toda New gen}
\ddot{x}_k=r(\dot{x}_k)\big(f(x_{k+1}-x_k)-f(x_k-x_{k-1})\big).
\end{equation}
System \eqref{Toda New gen} is Lagrangian, with the Lagrange function
\begin{equation}\label{Toda New gen Lagr}
\rL(x,\dot{x})=\sum_{k=1}^N K(\dot{x}_k)-\sum_{k=1}^N U(x_{k+1}-x_k),
\end{equation}
where $K''(v)=1/r(v)$ and $U'(u)=f(u)$. By the Legendre transformation, one easily finds the Hamilton function $\rH(x,p)$. 

All symplectic realizations of dTL$(h)$ have the following general structure:
\begin{equation}\label{dToda New gen}
\psi(\wx_k-x_k;h)-\psi(x_k-\undertilde{x}_k;h)=\phi(\undertilde{x}_{k+1}-x_k;h)-\phi(x_k-\wx_{k-1};h).
\end{equation}
These Newtonian equations of motion admit a Lagrangian formulation. They are identified as {\em discrete time Euler-Lagrange equations}
\begin{equation}\label{dEL gen}
\frac{\partial}{\partial x_k} \big(\Lambda(x,\wx;h)+\Lambda(\undertilde{x},x;h)\big)=0
\end{equation}
for the {\em discrete time Lagrange function} $\Lambda:\mathbb R^N\times \mathbb R^N\times \mathbb R\to\mathbb R$,
\begin{equation}\label{dToda New gen Lagr}
\Lambda(x,\wx;h)=\sum_{k=1}^N\Psi(\wx_k-x_k;h)-\sum_{k=1}^{N} \Phi(x_{k+1}-\wx_k;h),
\end{equation}
where $\Psi'(\xi;h)=\psi(\xi;h)$, $\Phi'(\xi;h)=\phi(\xi;h)$. The discrete time Euler-Lagrange equations generate a symplectic map $F:T^*\mathbb R^N\to T^*\mathbb R^N$, 
$F(x,p)=(\wx,\wip)$ by the formulas
\begin{equation}\label{dToda map gen}
\left\{\begin{array}{ll}
p_k =  -\partial\Lambda(x,\wx;h)/\partial x_k & =\psi(\wx_k-x_k;h)+\phi(x_k-\wx_{k-1};h),\vspace{1.5mm}\\
\wip_k = \partial\Lambda(x,\wx;h)/\partial \wx_k & = \psi(\wx_k-x_k;h)+\phi(x_{k+1}-\wx_{k};h).
\end{array}\right.
\end{equation}
In order that equations \eqref{dToda map gen} define a map $(x,p)\mapsto (\wx,\wip)$, the first of these equation should be solvable for $\widetilde x$ (at least locally), i.e., the matrix of the mixed partial derivatives of the Lagrange function $\Lambda$ should be non-degenerate, $\det(\partial^2 \Lambda/\partial x_i\partial\widetilde x_j)\neq 0$.

Function $\Lambda(x,\wx;h)$ from \eqref{dToda New gen Lagr} is a difference approximation to $\rL(x,\dot x)$ from (\ref{Toda New gen Lagr}) as $h\to 0$. More precisely, assuming that $\wx=x+h\dot{x}+O(h^2)$, 
we find:
\[
h^{-1}\Lambda(x,\wx;h)=\rL(x,\dot{x})+O(h),
\]
provided $h^{-1}\Psi(hv;h)=K(v)+O(h)$ and $h^{-1}\Phi(u;h)=U(u)+O(h)$.

\paragraph{Realization of the linear bracket: exponential Toda lattice.} Symplectic realization of the bracket $\{\cdot,\cdot\}_1$:
$$
a_k=\eto{x_{k+1}\nm x_k}, \quad b_k=p_k.
$$

$$
\begin{array}{ccc}
 \mathrm{TL}: & \qquad  r(v)=1,  &\quad f(u)=\eto{u},   \vspace{2mm}\\
 \mathrm{dTL}(h): & \qquad   \psi(v;h)=\dfrac{1}{h}(\eto{v}-1),  &\quad \phi(u;h)=h\eto{u}. 
\end{array}
$$

\paragraph{Realization of the linear bracket: dual Toda lattice.} Another symplectic realization of the bracket $\{\cdot,\cdot\}_1$:
\begin{equation}\label{TL dual par}
a_k=\eto{p_k},\quad b_k=x_k-x_{k-1}.
\end{equation}

$$
\begin{array}{ccc}
 \mathrm{TL}: & \qquad  r(v)=v,  &\quad f(u)=u,   \vspace{2mm}\\
 \mathrm{dTL}(h): & \qquad   \psi(v;h)=\log\dfrac{v}{h},  &\quad \phi(u;h)=\log(1+hu). 
\end{array}
$$

\paragraph{Realization of the quadrtic bracket: modified exponential Toda lattice.} Symplectic realization of the bracket $\{\cdot,\cdot\}_2$:
\begin{equation}\label{TL q par}
a_k=\eto{x_{k+1}\nm x_k\np p_k}, \quad b_k=\eto{p_k}+\eto{x_k\nm x_{k-1}}
\end{equation}

$$
\begin{array}{ccc}
 \mathrm{TL}: & \qquad  r(v)=v,  &\quad f(u)=\eto{u},   \vspace{2mm}\\
 \mathrm{dTL}(h): & \qquad   \psi(v;h)=\log\dfrac{\eto{v}-1}{h},  &\quad \phi(u;h)=\log(1+h\eto{u}). 
\end{array}
$$

\paragraph{Realization of the linear-quadratic  bracket: modified exponential Toda lattice with parameter.}
Symplectic realization of the linear combination $\{\cdot,\cdot\}_1+\epsilon\{\cdot,\cdot\}_2$
of the brackets \eqref{TL l br} and \eqref{TL q br}:
\begin{equation}\label{TL m1 par}
a_k=\eto{x_{k+1}\nm x_k\np\epsilon p_k},\quad 
b_k=\epsilon^{-1}\big(\eto{\epsilon p_k}-1\big)+\epsilon\eto{x_k\nm x_{k-1}}.
\end{equation}

$$
\begin{array}{ccc}
 \mathrm{TL}: & \qquad  r(v)=1+\epsilon v,  &\quad f(u)=\eto{u},   \vspace{2mm}\\
 \mathrm{dTL}(h): & \qquad   \psi(v;h)=\dfrac{1}{\epsilon}\log\Big(1+\dfrac{\epsilon}{h}\big(\eto{v}-1\big)\Big),  &\quad \phi(u;h)=\dfrac{1}{\epsilon}\log(1+h\epsilon\eto{u}). 
\end{array}
$$
Particular case $\epsilon=h$:
$$
\begin{array}{ccc}
 \mathrm{dTL}(h): & \qquad   \psi(v;h)=\dfrac{v}{h},  &\quad \phi(u;h)=\dfrac{1}{h}\log(1+h^2\eto{u}). 
\end{array}
$$

\paragraph{Realization of the cubic-quadratic bracket: multiplicative hyperbolic Toda lattice.} 
Symplectic realization of the  linear combination  $-\{\cdot,\cdot\}_3-4\beta\{\cdot,\cdot\}_2$ of the brackets 
\eqref{TL c br} and \eqref{TL q br}: 
\begin{equation}\label{TL cq par}
\left\{\begin{array}{rcl}
a_k & = & 
\displaystyle\frac{\beta^2\big(\cth(x_k-x_{k-1})+1\big)
\big(\cth(x_{k+1}-x_k)-1\big)}
{\sh^2(\nu p_k)}\;,\\ \\
b_k & = & -\beta\big(\cth(\beta p_k)+1\big)\big(\cth(x_{k}-x_{k-1})+1\big)\\
 & & -\beta\big(\cth(\beta p_{k-1})-1\big)\big(\cth(x_k-x_{k-1})-1\big).
\end{array}\right.
\end{equation}

$$
\begin{array}{ccc}
 \mathrm{TL}: & \qquad  r(v)=-(v^2-\beta^2),  &\quad f(u)=\cth(u),   \vspace{3mm}\\
 \mathrm{dTL}(h): & \qquad   \psi(v;h)=\dfrac{1}{2\beta}\log\dfrac{\sinh(v+\beta h_0)}{\sinh(v-\beta h_0)},  &\quad 
                                       \phi(u;h)=\dfrac{1}{2\beta}\log\dfrac{\sinh(u+\beta h_0)}{\sinh(u-\beta h_0)},
\end{array}
$$
where
$
h_0=-(1/4\beta)\log(1-4\beta h)=h+O(\beta h^2).
$

\paragraph{Realization of the cubic bracket. I: multiplicative rational Toda lattice.}
A symplectic realization of the bracket $-\{\cdot,\cdot\}_3$, see \eqref{TL c br}:
\begin{equation}\label{TL c2 par}
a_k=\frac{1}{(x_k-x_{k-1})(x_{k+1}-x_k)\,\sh^2(p_k)}\,, \quad
b_k=-\frac{\cth(p_{k-1})+\cth(p_k)}{x_k-x_{k-1}}
\end{equation}

$$
\begin{array}{ccc}
 \mathrm{TL}: & \qquad  r(v)=-(v^2-1),  &\quad f(u)=\dfrac{1}{u},   \vspace{2mm}\\
 \mathrm{dTL}(h): & \qquad   \psi(v;h)=\dfrac{1}{2}\log\dfrac{v+h}{v-h},  &\quad \phi(u;h)=\dfrac{1}{2}\log\dfrac{u+h}{u-h}.
\end{array}
$$

\paragraph{Realization of the cubic bracket. II: additive rational Toda lattice.}
Another symplectic realization of the bracket $-\{\cdot,\cdot\}_3$, see \eqref{TL c br}: 
\begin{equation}\label{TL c3 par}
a_k=\frac{1}{(x_k-x_{k-1})(x_{k+1}-x_k)\,p_k^2}, \quad
b_k=-\frac{1}{x_k-x_{k-1}}\left(\frac{1}{p_{k-1}}+\frac{1}{p_k}\right).
\end{equation}

$$
\begin{array}{ccc}
 \mathrm{TL}: & \qquad  r(v)=-v^2,  &\quad f(u)=\dfrac{1}{u},   \vspace{2mm}\\
 \mathrm{dTL}(h): & \qquad   \psi(v;h)=\dfrac{h}{v},  &\quad \phi(u;h)=\dfrac{h}{u}.
\end{array}
$$

\paragraph{Bibliographical remarks.} Integrable systems of the form (\ref{Toda New gen}) were classified in \cite{Ya89}. Yamilov's list coincides with the list of the present section. The fact that all items of this list are various symplectic realizations of the flow TL, was observed in \cite{Su97b, Su03}. The latter references contain also discretizations of all items of the Yamilov's list, as well as the fact that they all are various symplectic realiizations of the map dTL$(h)$. 

It should be mentioned that, if one generalizes the ansatz (\ref{Toda New gen}) by allowing functions $f$ to depend on $x_k$, $x_{k+1}$ not necessarily through the differences $x_{k+1}-x_k$, then Yamilov's list contains one further system, the so called {\em elliptic Toda lattice},
\[
\ddot x_k=(\dot x_k^2-1)\Big(\zeta(x_{k+1}+x_k)-\zeta(x_{k+1}-x_k)+\zeta(x_k+x_{k-1})+\zeta(x_k-x_{k-1})-2\zeta(2x_k)\Big).
\]
Here $\zeta(u)$ is the Weierstrass zeta-function. This system was independently found in \cite{Kr00}. Its discretization, 
$$
 \phi(x_k,\wx_k;h)+\phi(x_k,\undertilde{x}_k;h)-\phi(x_k,\undertilde{x}_{k+1};h)-\phi(x_k,\wx_{k-1};h)=0,
$$
where
\begin{equation}\label{phi elliptic}
 \phi(x_0,x_1;\alpha)=\frac{1}{2}\log\frac{\sigma(x_0+x_1+\alpha)\sigma(x_0-x_1+\alpha)}{\sigma(x_0+x_1-\alpha)\sigma(x_0-x_1-\alpha)},
\end{equation}
was found in \cite{A00}, \cite{ASu04}. The elliptic Toda lattice and its discrete time counterpart admit a hyperbolic degeneration ($\zeta(u)\to\coth(u)$, $\sigma(u)\to\sinh(u)$) and a rational degeneration ($\zeta(u)\to1/u$, $\sigma(u)\to u$).

\setcounter{equation}{0}
\section[Relativistic Toda lattice]{Relativistic Toda lattice in Flaschka-Manakov variables:\\
\quad equations of motion, Lax representation and tri-Hamiltonian structure}
\label{Sect RTL tri-Ham}

There exists a very remarkable generalization of TL, called {\em relativistic Toda lattice}. It is a one-parameter perturbation of TL, and in a certain physical interpretation this (small) parameter $\alpha$ has the meaning of the inverse speed of light. In all our considerations, a special attention is payed to an immediate and transparent limit  $\alpha\to 0$.
Actually, there are two simplest flows which are perturbations of TL, which we will call the first and the ``negative first'' flows of the RTL hierarchy. The first flow, denoted hereafter 
$\mathrm{RTL}_+(\alpha)$, reads:
\begin{equation}\label{RTL+ param}
\left\{\begin{array}{l}
\dot{b}_k=(1+\alpha b_k)(a_k-a_{k-1}), \\  
\dot{a}_k=a_k(b_{k+1}-b_k+\alpha a_{k+1}-\alpha a_{k-1}),
\end{array}\right. \quad 1\le k\le N. 
\end{equation}
The negative first flow, denoted hereafter  $\mathrm{RTL}_-(\alpha)$, reads:
\begin{equation}\label{RTL- param}
\left\{\begin{array}{l}
\dot{b}_k=\displaystyle\frac{a_k}{1+\alpha b_{k+1}}-
\displaystyle\frac{a_{k-1}}{1+\alpha b_{k-1}}, \vspace{2mm}\\ 
\dot{a}_k=
a_k\left(\displaystyle\frac{b_{k+1}}{1+\alpha b_{k+1}}-
\displaystyle\frac{b_k}{1+\alpha b_k}\right), 
\end{array}\right. \quad 1\le k\le N. 
\end{equation}
All conventions about boundary conditions (open-end or periodic) remain valid for the relativistic case. 

Remarkably, both flows $\mathrm{RTL}_{\pm}(\alpha)$ admit Lax representations falling into the AKS scheme. 
We define:
\begin{align}
L(a,b,\lambda)  = & \sum_{k=1}^N (1+\alpha b_k)E_{kk}+
\alpha\lambda\sum_{k=1}^N E_{k+1,k},
\label{RTL param L}\\
U(a,b,\lambda) = & I-\alpha\lambda^{-1}\sum_{k=1}^N a_kE_{k,k+1},
\label{RTL param U}
\end{align}
and
\begin{equation}\label{RTL param Ts}
T_1(a,b,\lambda)=L(a,b,\lambda)U^{-1}(a,b,\lambda),\quad
T_2(a,b,\lambda)=U^{-1}(a,b,\lambda)L(a,b,\lambda).
\end{equation}
For both matrices $T_{1,2}(a,b,\lambda)$ we have an asymptotic formula:
\begin{equation}\label{RTL param T as}
T_{1,2}(a,b,\lambda)=I+\alpha T(a,b,\lambda)+O(\alpha^2)\;,
\end{equation}
where $T(a,b,\lambda)$ is the Lax matrix of the TL hierarchy.
\index{relativistic Toda lattice!second Lax representation}
\begin{proposition}\label{Lax triads for RTL+}
Equations of motion \eqref{RTL+ param} of the flow $\mathrm{RTL}_+(\alpha)$ are equivalent to the ``Lax triads'':
\begin{equation}\label{RTL+ param triads}
\dot{L}=LB_2-B_1L, \quad \dot{U}=UB_2-B_1U,
\end{equation}
which also imply usual Lax equations for the matrices $T_{1,2}(a,b,\lambda)$:
\begin{equation}\label{RTL+ param Lax}
\dot{T_i}=[T_i,B_i], \quad i=1,2.
\end{equation}
Here the auxiliary matrices 
\begin{align}
B_1(a,b,\lambda) = & \sum_{k=1}^N(b_k+\alpha a_{k-1})E_{kk}+\lambda
\sum_{k=1}^N E_{k+1,k}, \label{RTL+ param A1}\\
B_2(a,b,\lambda) = & \sum_{k=1}^N(b_k+\alpha a_k)E_{kk}+\lambda
\sum_{k=1}^N E_{k+1,k} \label{RTL+ param A2}
\end{align}
admit the following expressions:
\begin{equation}\label{RTL+ param As}
B_i=\pi_+\big((T_i-I)/\alpha\big), \quad i=1,2\;.
\end{equation}
\end{proposition}

\begin{proposition}\label{Lax triads for RTL-}
Equations of motion \eqref{RTL- param} of the flow
$\mathrm{RTL}_-(\alpha)$ are equivalent to the ``Lax triads'':
\begin{equation}
\dot{L}=C_1L-LC_2, \quad \dot{U}=C_1U-UC_2,
\end{equation}
which also imply usual Lax equations for the matrices $T_{1,2}(a,b,\lambda)$:
\begin{equation}\label{RTL- param Lax}
\dot{T_i}=[C_i,T_i], \quad i=1,2.
\end{equation}
Here the auxiliary matrices 
\begin{align}
C_1(a,b,\lambda) = &\ \lambda^{-1}\sum_{k=1}^N 
\frac{a_k}{1+\alpha b_{k+1}}\,E_{k,k+1}, \label{RTL- param C}\\
C_2(a,b,\lambda) = &\ \lambda^{-1}\sum_{k=1}^N 
\frac{a_k}{1+\alpha b_k}\,E_{k,k+1} \label{RTL- param D}
\end{align}
admit the following expressions:
\begin{equation}\label{RTL- param Cs}
C_i=\pi_-\big((I-T_i^{-1})/\alpha\big), \quad i=1,2.
\end{equation}
\end{proposition}

The Lax representations of the RTL flows lives in the same algebra $\g$ as the Lax representation of TL; moreover, Lax representation \eqref{RTL+ param Lax} of RTL$_+(\alpha)$ is an $O(\alpha)$-perturbation  of the $\pi_+$ version of the Lax equation \eqref{Tdot spec} for TL, while Lax representation \eqref{RTL- param Lax} of RTL$_-(\alpha)$ is an $O(\alpha)$-perturbation of the $\pi_-$ version of the Lax equation \eqref{Tdot spec} for TL.

The flows RTL$_\pm(\alpha)$, like their non-relativistic counterpart TL, are tri-Hamiltonian. That is, the phase space of RTL can be equipped with three local Poisson brackets which are preserved by these flows which is therefore Hamiltonian with respect to any of these Poisson structures. These brackets are compatible in the sense that their linear combinations are invariant Poisson brackets, as well. 

\paragraph{Linear bracket:}
\begin{equation}\label{RTL param l br}
\begin{array}{lcl}
\{b_k,a_k\}_1=-a_k, &  & \{a_k,b_{k+1}\}_1=-a_k,\\ 
\{b_k,b_{k+1}\}_1=\alpha a_k. & & 
\end{array}
\end{equation}
The Hamilton functions of RTL$_{\pm}(\alpha)$ are $\rH_2(a,b)$ and $-\alpha^{-1}\rH_0(a,b)$, respectively, where
\begin{align}
\rH_2= & \sum_{k=1}^N \left(\frac{1}{2}\ b_k^2+a_k\right)
+\alpha\sum_{k=1}^N(b_k+b_{k+1})a_k+
\alpha^2\sum_{k=1}^N\left(\frac{1}{2}\,a_k^2+a_ka_{k+1}\right),
 \label{RTL param H2}\\
\rH_0 = & \ \alpha^{-1}\,\sum_{k=1}^N 
\log(1+\alpha b_k).\label{RTL param H0}
\end{align}
Notice that the Hamilton function $-\alpha^{-1}\rH_0(a,b)$ is singular in $\alpha$, but it 
becomes regular upon adding the Casimir function $\alpha^{-1}\rH_1(a,b)$, where
\begin{equation}\label{RTL param H1}
\rH_1=\sum_{k=1}^N b_k+\alpha \sum_{k=1}^N a_k.
\end{equation}
Indeed:
\[
-\alpha^{-1}\rH_0+\alpha^{-1}\rH_1=\sum_{k=1}^N \left(\frac{1}{2}\ b_k^2+ a_k\right)+O(\alpha)\;.
\]

\paragraph{Quadratic bracket:}
\begin{equation}\label{RTL param q br}
\begin{array}{lcl}
\{b_k,a_k\}_2=-b_ka_k, &  &
\{a_k,b_{k+1}\}_2=-a_kb_{k+1}, \\ 
\{b_k,b_{k+1}\}_2=-a_k, &  & 
\{a_k,a_{k+1}\}_2=-a_ka_{k+1}. 
\end{array}
\end{equation}
Amazingly, this bracket coincides with the invariant Poisson bracket \eqref{TL q br} of TL. The Hamilton functions of the flows RTL$_\pm(\alpha)$ are $\rH_1(a,b)$ and
$-\alpha\,\rH_0(a,b)$, respectively.

\paragraph{Cubic bracket:}

\begin{equation}\label{RTL param c br}
\begin{array}{l}
\{b_k,a_k\}_3     = -a_k(b_k^2+a_k)-\alpha b_ka_k^2, \\ 
\{a_k,b_{k+1}\}_3 = -a_k(b_{k+1}^2+a_k)-\alpha a_k^2b_{k+1},  \\ 
\{b_k,b_{k+1}\}_3 = -a_k(b_k+b_{k+1})-\alpha b_ka_kb_{k+1}, \\ 
\{a_k,a_{k+1}\}_3 = -2a_kb_{k+1}a_{k+1}-\alpha a_ka_{k+1}(a_k+a_{k+1}), \\ 
\{b_k,a_{k+1}\}_3 = -a_ka_{k+1}-\alpha b_ka_ka_{k+1}, \\ 
\{a_k,b_{k+2}\}_3 = -a_ka_{k+1}-\alpha a_ka_{k+1}b_{k+2},\\ 
\{a_k,a_{k+2}\}_3 = -\alpha a_ka_{k+1}a_{k+2}.
\end{array}
\end{equation}
The Hamilton functions of the flows RTL$_\pm(\alpha)$ in this bracket are $\frac{1}{2}\rC(a,b)$ and $\frac{1}{2}\rC(a,b)-\alpha\rH_0(a,b)$, where
\begin{equation}\label{RTL+ Ham in c br}
\rC(a,b)=\sum_{k=1}^N\log(a_k).
\end{equation}

\paragraph{Bibliographical remarks.}

Relativistic Toda lattice was introduced in \cite{Rui90}, as the Newtonian equations (\ref{RTL+ New introd}). 
Early references, concerning inverse scattering and finite gap solutions, Lax representations and tri-Hamiltonian structure include: \cite{BR88, BR89a, BR89b},
\cite{OFZR89}, \cite{ZTOF91}.

The Lax representation in terms of $(L,U)\in\g\otimes\g$ was given in \cite{Su91}, where it was pointed out that the corresponding Lax matrices $T_{1,2}$ build an orbit of a
Lie--Poisson group. This was further developed in \cite{FM97}, \cite{FG00}, and generalized for all simple Lie groups 
in \cite{HKKR00}.

\setcounter{equation}{0}
\section{Discretization of the relativistic Toda lattice in the Flaschka-Manakov variables}
\label{Sect RTL param discretization}

To find integrable discretization of the flows RTL$_{\pm}(\alpha)$, we apply the recipe of Section \ref{Sect recipe}, i.e., we consider the maps
\[
\wT=\Pi_+^{-1}(F(T))\cdot T\cdot\Pi_+(F(T))=
\Pi_-(F(T))\cdot T\cdot\Pi_-^{-1}(F(T))
\]
with $T=T_1$ or $T_2$, and with
\[
F(T)=I+\frac{h}{\alpha}\,(T-I)\;\;\mathrm{for\;\; RTL}_+(\alpha), \quad
F(T)=I+\frac{h}{\alpha}\,(I-T^{-1})\;\;\mathrm{for\;\;RTL}_-(\alpha).
\]
As a matter of fact, it is much more convenient to work with the discrete time Lax triads, which take the form
\begin{align*}
\wL = & \ \Pi_+^{-1}(F(T_1))\cdot L\cdot\Pi_+(F(T_2))\ =\
\Pi_-(F(T_1))\cdot L\cdot\Pi_-^{-1}(F(T_2)),\\
\wU = & \ \Pi_+^{-1}(F(T_1))\cdot U\cdot\Pi_+(F(T_2))\ =\
\Pi_-(F(T_1))\cdot U\cdot\Pi_-^{-1}(F(T_2)).
\end{align*}
It turns out that for the flow $\mathrm{RTL}_+(\alpha)$ the version with the 
$\Pi_+$ factors is more suitable, while for the $\mathrm{RTL}_-(\alpha)$ flow 
the version with the $\Pi_-$ factors is preferable.  

\begin{lemma}\label{discr RTL+ param factors}
The factors $\mbB_{1,2}=\Pi_+\Big(I+\displaystyle\frac{h}{\alpha}
(T_{1,2}-I)\Big)$ are of the form
\begin{align}
\mbB_1(a,b,\lambda)= & \ \sum_{k=1}^N\ga_kE_{kk}+h\lambda\sum_{k=1}^N E_{k+1,k},
\label{dRTL+ param A} \\
\mbB_2(a,b,\lambda)= & \ \sum_{k=1}^N\gb_kE_{kk}+h\lambda\sum_{k=1}^N E_{k+1,k},
\label{dRTL+ param B}
\end{align}
where coefficients $\ga_k=\ga_k(a,b)$ are defined by the recurrent relations 
\begin{equation}\label{dRTL+ param a}
\ga_k=1+hb_k+\frac{h(\alpha-h)a_{k-1}}{\ga_{k-1}},
\end{equation}
and coefficients $\gb_k=\gb_k(a,b)$  are defined by  either of the relations
\begin{equation}\label{dRTL+ param b}
\gb_k=\ga_{k-1}\,\frac{\ga_k+h\alpha a_k}{\ga_{k-1}+h\alpha a_{k-1}}=\ga_k\frac{\alpha\ga_{k+1}-h(1+\alpha b_{k+1})}{\alpha\ga_k-h(1+\alpha b_k)},
\end{equation}
which are equivalent by virtue of \eqref{dRTL+ param a}.
\end{lemma}

\noindent
\textbf{Remark.} As in the case of the discretization of TL, in the open--end 
case due to $a_0=0$ one can solve (\ref{dRTL+ param a}) explicitly in terms of 
finite continued fractions:
\[
\ga_1=1+hb_1;\quad 
\ga_2=1+hb_2+\frac{h(\alpha-h) a_1}{1+hb_1};\quad\ldots\quad;
\]
\[
\ga_N=1+hb_N+\frac{h(\alpha-h)a_{N-1}}{1+hb_{N-1}+
\displaystyle\frac{h(\alpha-h)a_{N-2}}{1+hb_{N-2}+\;
\raisebox{-3mm}{$\ddots$}
\raisebox{-4.5mm}{$\;+\displaystyle\frac{h(\alpha-h)a_1}{1+hb_1}$}}}.
\]
In the periodic case $\ga_k$ can be expressed as infinite $N$-periodic continued fractions, and therefore they are double-valued functions, with one of the branches satisfying $\ga_k=1+h(b_k+\alpha a_{k-1})+O(h^2)$ as $h\to 0$.

\begin{theorem} \label{discrete RTL+ param}
Consider the discrete time Lax triads
\begin{equation}\label{dRTL+ param Lax}
\wL=\mbB_1^{-1}L\mbB_2,\quad \wU=\mbB_1^{-1}U\mbB_2,
\end{equation}
with $\mbB_{1,2}=\Pi_+\Big(I+\displaystyle\frac{h}{\alpha}(T_{1,2}-I)\Big)$.
They serve as a Lax representation of the map $(a,b)\mapsto(\wa,\wb)$ given by 
\begin{equation}\label{dRTL+ param}
1+\alpha\wb_k=(1+\alpha b_k)\,\frac{\gb_k}{\ga_k}, \quad 
\wa_k=a_k\,\frac{\gb_{k+1}}{\ga_k},
\end{equation}
where functions $\ga_k$, $\gb_k$ are defined by \eqref{dRTL+ param a}, \eqref{dRTL+ param b}.
Map \eqref{dRTL+ param} will be denoted by $\mathrm{dRTL}_+(\alpha,h)$. 
\end{theorem}
\begin{proof}
Matrix equations \eqref{dRTL+ param Lax}, written entry-wise, are equivalent to equations \eqref{dRTL+ param} in conjunction with
\begin{equation}\label{dRTL+ param addition}
h(1+\alpha\wb_k)+\alpha\ga_{k+1}=h(1+\alpha b_{k+1})+\alpha\gb_k, \quad 
\ga_k-h\alpha\wa_{k-1}=\gb_k-h\alpha a_k.
\end{equation}
A direct computation shows that, given \eqref{dRTL+ param}, equations \eqref{dRTL+ param addition} follow from \eqref{dRTL+ param a}, \eqref{dRTL+ param b}. This finishes the proof. For later reference, we mention the following formula:
\begin{equation}\label{dRTL+ d+c 1}
h(\wb_k+\alpha\wa_{k-1})-h(b_{k+1}+\alpha a_k) = \ga_k-\ga_{k+1}.
\end{equation}
It follows immediately by eliminating $\gb_k$ between equations \eqref{dRTL+ param addition}.
\end{proof}

Next, we turn to discretization of the flow RTL$_-(\alpha)$. 
\begin{lemma}\label{discr RTL- param factors}
The factors $\mbC_{1,2}=\Pi_-\Big(I+\displaystyle\frac{h}{\alpha}
(I-T_{1,2}^{-1})\Big)$ are of the form
\begin{align}
\mbC_1(a,b,\lambda) = &\  I+h\lambda^{-1}\sum_{k=1}^N\gc_kE_{k,k+1},
\label{dRTL- param C}\\
\mbC_2(a,b,\lambda) = &\ I+h\lambda^{-1}\sum_{k=1}^N\gd_kE_{k,k+1},
\label{dRTL- param D}
\end{align}
where coefficients $\gd_k=\gd_k(a,b)$ are defined by the recurrent relations
\begin{equation}\label{dRTL- param d}
\gd_k=\frac{a_k}{1+(\alpha+h)(b_k-h\gd_{k-1})},
\end{equation}
and coefficients $\gc_k=\gc_k(a,b)$ are given by either of two expressions
\begin{equation}\label{dRTL- param c}
\gc_k=\gd_k\,\frac{1+\alpha (b_k-h\gd_{k-1})}{1+\alpha (b_{k+1}-h\gd_k)}=\gd_{k+1}\,\frac{\alpha a_k+h\gd_k}{\alpha a_{k+1}+h\gd_{k+1}},
\end{equation}
which are equivalent by virtue of \eqref{dRTL- param d}.
\end{lemma}

\noindent
\textbf{Remark.} In the open-end case we have the following finite continued
fractions expressions for $\gd_k$:
\[
\gd_1=\frac{a_1}{1+(\alpha+h)b_1};\quad
\gd_2=\frac{a_2}{1+(\alpha+h)b_2-\displaystyle\frac{h(\alpha+h)a_1}
{1+(\alpha+h)b_1}};\quad
\ldots\quad;
\]
\[
\gd_{N-1}=\frac{a_{N-1}}{1+(\alpha+h)b_{N-1}-\displaystyle
\frac{h(\alpha+h)a_{N-2}}{1+(\alpha+h)b_{N-2}-\;
\raisebox{-3mm}{$\ddots$}
\raisebox{-4.5mm}{$\;-\displaystyle\frac{h(\alpha+h)a_1}{1+(\alpha+h)b_1}$}}}.
\]

\noindent
In the periodic case these continued fractions are replaced by $N$-periodic
ones, which represent double-valued functions, with one of the branches satisfying 
$\gd_k= a_k/(1+\alpha b_k)+O(h)$ as $h\to 0$.

\begin{theorem}\label{discrete RTL- param}
Consider the discrete time Lax triads
\begin{equation}\label{dRTL- param Lax triads}
\wL=\mbC_1L\, \mbC_2^{-1},\quad \wU=\mbC_1U\mbC_2^{-1},
\end{equation}
with $\mbC_{1,2}=\Pi_-\Big(I+\displaystyle\frac{h}{\alpha}(I-T_{1,2}^{-1})\Big)$.
They serve as a Lax representation of the map $(a,b)\mapsto(\wa,\wb)$ given by
\begin{equation}\label{dRTL- param}
1+\alpha\wb_k=(1+\alpha b_{k+1})\,\frac{\gc_k}{\gd_k}, \quad 
\wa_k=a_{k+1}\,\frac{\gc_k}{\gd_{k+1}},
\end{equation}
where functions $\gd_k$, $\gc_k$ are defined by \eqref{dRTL- param d}, \eqref{dRTL- param c}.
Map \eqref{dRTL- param} will be denoted $\mathrm{dRTL}_-(\alpha,h)$. 
\end{theorem}
\begin{proof}
Matrix equations \eqref{dRTL- param Lax triads}, written entry-wise, are equivalent to a system of scalar equations consisting of \eqref{dRTL- param} and of 
\begin{equation}\label{dRTL- param addition}
\wb_k+h\gd_{k-1}=b_k+h\gc_k, \quad \alpha\wa_k-h\gd_k=\alpha a_k-h\gc_k.
\end{equation}
Given \eqref{dRTL- param}, equations \eqref{dRTL- param addition} are consequences of \eqref{dRTL- param d}, \eqref{dRTL- param c}. For the later reference, we observe the following formula:
\begin{equation}\label{dRTL- d+c 1}
(\wb_k+\alpha\wa_{k-1})-(b_k+\alpha a_{k-1}) = h\gc_k-h\gc_{k-1},
\end{equation}
which follows by eliminating $\gd_k$ between equations \eqref{dRTL- param addition}.
\end{proof}

By construction, both maps dRTL$_\pm(\alpha,h)$ are Poisson with respect to 
each one of compatible Poisson brackets \eqref{RTL param l br},  \eqref{RTL param q br}, \eqref{RTL param c br}, and hence with respect to 
their arbitrary linear combination. Moreover, these maps share integrals of motion with the flows RTL$_\pm(\alpha)$, commute with those flows and with one another,  and are interpolated by certain flows of the RTL$(\alpha)$ hierarchy.

\paragraph{Bibliographical remarks.}
Discretization of the relativistic Toda lattice was performed in \cite{Su96} as one of the first applications of our general recipe of integrable discretization.

%%%%%%%%%%%%%%%%%%%%%%%%%%%%%%%%%%%%%%%%%%%%%%%%%%%
%%%%%%%%%%%%%%%%%%%%%%%%%%%%%%%%%%%%%%%%%%%%%%%%%%%

\setcounter{equation}{0}
\section{Symplectic realization of the linear bracket: additive exponential relativistic Toda lattice}
\label{Sect additive exp rel Toda}

\begin{proposition}\label{rel Toda linear parametrization}
The map
\begin{equation}\label{RTL l par}
b_k=p_k-\alpha\eto{x_k\nm x_{k-1}},\quad a_k=\eto{x_{k+1}\nm x_k}
\end{equation}
 is a symplectic realization of the linear bracket $\{\cdot,\cdot\}_{1}$, see \eqref{RTL param l br}, on the phase space of $\mathrm{RTL}_\pm(\alpha)$.
\end{proposition}

\begin{theorem} \label{rel Toda in linear parametrization}
Pull-back of the flow $\mathrm{RTL}_+(\alpha)$ under 
parametrization \eqref{RTL l par} is a Hamiltonian system
with the Hamilton and Lagrange functions
\begin{align}
\rH(x,p)=& \ \frac{1}{2}\sum_{k=1}^Np_k^2+
\sum_{k=1}^N(1+\alpha p_k)\,\eto{x_{k+1}\nm x_k}, 
\label{RTL+ l H in xp}\\
\rL(x,\dot{x})= & \ \frac{1}{2}\;\sum_{k=1}^N\dot{x}_k^2-
\sum_{k=1}^N(1+\alpha\dot{x}_k)\,\eto{x_{k+1}\nm x_k}
+\frac{\alpha^2}{2}\,\sum_{k=1}^N\eto{{\scriptstyle 2}(x_{k+1}\nm x_k)}.
\label{RTL+ l Lagr}
\end{align}
The corresponding Newtonian equations of motion read:
\begin{equation}\label{RTL+ l New}
\ddot{x}_k  =  (1+\alpha\dot{x}_{k+1})\,\eto{x_{k+1}\nm x_k}-
(1+\alpha\dot{x}_{k-1})\,\eto{x_k\nm x_{k-1}}-\alpha^2\eto{{\scriptstyle 2}(x_{k+1}\nm x_k)}+\alpha^2\eto{{\scriptstyle 2}(x_k\nm x_{k-1})}.
\end{equation}
\end{theorem}
\begin{proof}
Hamilton function \eqref{RTL+ l H in xp} is obtained from \eqref{RTL param H2} upon substitution \eqref{RTL l par}. Straightforward computations lead to Lagrange function \eqref{RTL+ l Lagr} and to its Euler-Lagrange equations \eqref{RTL+ l New}.
\end{proof}

\begin{theorem}\label{discr rel Toda in linear parametrization}
Pull-back of the map $\mathrm{dRTL}_+(\alpha,h)$ under
parametrization \eqref{RTL l par} is the following
symplectic map:
\begin{equation} \label{dRTL+ l}
\left\{ \begin{array}{l}  p_k  =  
\dfrac{1}{h}\big(\eto{\wx_k\nm x_k}-1\big)-
\displaystyle\frac{(\alpha-h)\,\eto{x_k\nm \wx_{k-1}}}
{\,\raisebox{-2mm}{$1-h\alpha\eto{x_k\nm \wx_{k-1}}$}}
+\alpha\eto{x_k\nm x_{k-1}}-\alpha\eto{x_{k+1}\nm x_k},
\vspace{2mm}\\
\wip_k  = 
\dfrac{1}{h}\big(\eto{\wx_k\nm x_k}-1\big)-
\displaystyle\frac{(\alpha-h)\,\eto{x_{k+1}\nm \wx_k}}
{\,\raisebox{-2mm}{$1-h\alpha\eto{x_{k+1}\nm \wx_k}$}}.
\end{array}\right.
\end{equation}
The corresponding Newtonian equations of motion:
\begin{equation}\label{dRTL+ l New}
\ueto{\wx_k\nm x_k}-\ueto{x_k\nm\undertilde{x}_k}  = 
h\alpha\ueto{x_{k+1}\nm x_k}-h\alpha\eto{x_k\nm x_{k-1}}
-\displaystyle\frac{h(\alpha-h)\,\ueto{\undertilde{x}_{k+1}-x_k}}
{\,\raisebox{-2mm}{$1-h\alpha\ueto{\undertilde{x}_{k+1}-x_k}$}}
+\displaystyle\frac{h(\alpha-h)\,\eto{x_k\nm \wx_{k-1}}}
{\,\raisebox{-2mm}{$1-h\alpha\eto{x_k\nm \wx_{k-1}}$}}.
\end{equation}
\end{theorem}

\begin{proof} The second equation of motion in
(\ref{dRTL+ param}) together with equation (\ref{dRTL+ param b}) yield:
\[
\wa_k=a_k\frac{\ga_{k+1}+h\alpha a_{k+1}}{\ga_k+h\alpha a_k}.
\]
In the parametrization $a_k=\eto{x_{k+1}\nm x_k}$ this is equivalent
to the following quantity being constant, i.e., not depending on $k$:
\[
\left.\eto{\wx_k\nm x_k}\right/(\ga_k+h\alpha a_k)=\mathrm{const}.
\]
Choosing this constant to be equal to 1, we get:
\begin{equation}\label{dRTL+ l aux1}
\ga_k+h\alpha a_k=\eto{\wx_k\nm x_k},
\end{equation}
hence
\begin{equation}\label{dRTL+ l a}
\ga_k=\eto{\wx_k\nm x_k}\Big(1-h\alpha\eto{x_{k+1}\nm\wx_k}\Big),
\end{equation}
and (from (\ref{dRTL+ param b})):
\begin{equation}\label{dRTL+ l b}
\gb_k=\eto{\wx_k\nm x_k}\Big(1-h\alpha\eto{x_k\nm\wx_{k-1}}\Big).
\end{equation}
Further, (\ref{dRTL+ l aux1}), (\ref{dRTL+ l a}) allow us to derive
from the recurrent relation (\ref{dRTL+ param a}):
\begin{equation}\label{dRTL+ l aux2}
\ga_k-hb_k=1+\frac{h(\alpha-h)a_{k-1}}{\ga_{k-1}}
=1+\frac{h(\alpha-h)\,\eto{x_k\nm\wx_{k-1}}}
{\,\raisebox{-2mm}{$1-h\alpha\eto{x_k\nm\wx_{k-1}}$}},
\end{equation}
which, taking into account formulas (\ref{dRTL+ l a}) and $p_k=b_k+\alpha\eto{x_k\nm x_{k-1}}$, leads to the first formula (for $p_k$) in \eqref{dRTL+ l}.
To obtain the second formula (for $\wip_k$) in \eqref{dRTL+ l}, we make use of \eqref{dRTL+ d+c 1}.
This completes derivation of the Lagrangian equations of motion \eqref{dRTL+ l}.
The Newtonian ones follow immediately. 
\end{proof}

\begin{theorem}
Pull-back of the flow $\mathrm{RTL}_-(\alpha)$ under 
parametrization \eqref{RTL l par} is a Hamiltonian system
with the Hamilton and the Lagrange functions
\begin{align}
\rH(x,p)= & \ \alpha^{-1}\sum_{k=1}^N p_k
-\alpha^{-2}\sum_{k=1}^N\log\big(1+\alpha p_k-\alpha^2\eto{x_k\nm x_{k-1}}\big),
\label{RTL- l H in xp}\\
\rL(x,\dot{x})= & \ -\alpha^{-2}\,\sum_{k=1}^N\big(\alpha\dot{x}_k+\log (1-\alpha\dot{x}_k )\big)
-\sum_{k=1}^N(1-\alpha\dot{x}_k)\,\eto{x_k\nm x_{k-1}}.
\label{RTL- l Lagr}
\end{align}
The corresponding Newtonian equations of motion:
\begin{equation}\label{RTL- l New}
\ddot{x}_k=(1-\alpha\dot{x}_k)^2\big((1-\alpha\dot{x}_{k+1})\,
\eto{x_{k+1}\nm x_k}- (1-\alpha\dot{x}_{k-1})\,\eto{x_k\nm x_{k-1}}\big).
\end{equation}
\end{theorem}
\begin{proof}
Hamilton function \eqref{RTL- l H in xp} is obtained from $\alpha^{-1}(\rH_1-\rH_0)$, see \eqref{RTL param H0}, \eqref{RTL param H1}, under substitution \eqref{RTL l par}. Lagrange function \eqref{RTL- l Lagr} and Newtonian equations \eqref{RTL- l New} follow.
\end{proof}

\begin{theorem}
Pull-back of the map $\mathrm{dRTL}_-(\alpha,h)$ under parametrization \eqref{RTL l par} is the following symplectic
map:
\begin{equation}\label{dRTL- l}
\left\{\begin{array}{l}
p_k = 
\displaystyle\frac{h^{-1}\big(\eto{\wx_k\nm x_k}-1\big)}
{1-\alpha h^{-1}\big(\eto{\wx_k\nm x_k}-1\big)}
+(\alpha+h)\,\eto{x_k\nm\wx_{k-1}},  \vspace{3mm}\\
 \wip_k = 
\displaystyle\frac{h^{-1}\big(\eto{\wx_k\nm x_k}-1\big)} 
{1-\alpha h^{-1}\big(\eto{\wx_k\nm x_k}-1\big)}
+(\alpha+h)\,\eto{x_{k+1}\nm\wx_k}-\alpha\eto{\wx_{k+1}\nm\wx_k}+
\alpha\eto{\wx_k\nm\wx_{k-1}}. 
\end{array}\right.
\end{equation}
The corresponding Newtonian equations of motion:
\[
\frac{h^{-1}\big(\eto{\wx_k\nm x_k}-1\big)}{1-\alpha h^{-1}\big(\eto{\wx_k\nm x_k}-1\big)}-
\frac{h^{-1}\big(\eto{x_k\nm \undertilde{x}_k}-1\big)}{1-\alpha h^{-1}\big(\eto{x_k\nm \undertilde{x}_k}-1\big)}= \qquad\qquad\qquad\qquad\quad
\]
\begin{equation}\label{dRTL- l New}
=(\alpha+h)\big(\ueto{\undertilde{x}_{k+1}-x_k}-\ueto{x_k\nm \wx_{k-1}}\big)-\alpha\big(\ueto{x_{k+1}\nm x_k}-\ueto{x_k\nm x_{k-1}}\big).
\end{equation}
\end{theorem}
\begin{proof} The second equation of motion in
(\ref{dRTL- param}), together with the second expression for
$\gc_k$ in (\ref{dRTL- param c}), reads:
\[
\wa_k=a_{k+1}\,\frac{\alpha a_k+h\gd_k}{\alpha a_{k+1}+h\gd_{k+1}}.
\]
In the parametrization $a_k=\exp(x_{k+1}-x_k)$ this is equivalent to the
following quantity being constant, i.e., not depending on $k$:
\[
\left.(\alpha a_k+h\gd_k)\right/\eto{x_{k+1}\nm\wx_k}=\mathrm{const}.
\]
Choosing this constant equal to $\alpha+h$, we obtain:
\begin{equation}\label{dRTL- l d}
h\gd_k=(\alpha+h)\,\eto{x_{k+1}\nm\wx_k}-\alpha\eto{x_{k+1}\nm x_k}.
\end{equation}
The recurrent relation (\ref{dRTL- param d}) implies:
\[
1+(h+\alpha)(b_k-h\gd_{k-1})=\frac{a_k}{\gd_k}=
\frac{\eto{\wx_k\nm x_k}}{1-\alpha h^{-1}\big(\eto{\wx_k\nm x_k}-1\big)},
\]
which is equivalent to
\begin{equation}\label{dRTL- l aux2}
1+\alpha(b_k-h\gd_{k-1})=
\frac{1}{1-\alpha h^{-1}\big(\eto{\wx_k\nm x_k}-1\big)}.
\end{equation}
To compute $\gc_k$, we substitute (\ref{dRTL- l d}), (\ref{dRTL- l aux2})
into (\ref{dRTL- param c}), and arrive at:
\begin{equation}\label{dRTL- l c}
h\gc_k=(\alpha+h)\,\eto{x_{k+1}\nm\wx_k}-\alpha\eto{\wx_{k+1}\nm\wx_k}.
\end{equation}
Formula \eqref{dRTL- l aux2}, together with (\ref{dRTL- l d}) and $p_k=b_k+\alpha\eto{x_k\nm x_{k-1}}$, immediately implies the first formula (for $p_k$) in 
(\ref{dRTL- l}). To derive the second formula in (\ref{dRTL- l}) (for $\wip_k$), we make use of (\ref{dRTL- d+c 1}). Now Newtonian equations (\ref{dRTL- l New}) follow readily. 
\end{proof}

\paragraph{Bibliographical remarks.} Results of this section are from \cite{Su97a, Su03}.

%%%%%%%%%%%%%%%%%%%%%%%%%%%%%%%%%%%%%%%%%%%%%%%%%%%%%%%%%%%%%
%%%%%%%%%%%%%%%%%%%%%%%%%%%%%%%%%%%%%%%%%%%%%%%%%%%%%%%%%%%%%

\setcounter{equation}{0}
\section{The variety of symplectic realizations of RTL$_\pm(\alpha)$ and dRTL$_\pm(\alpha,h)$}
\label{Sect Newtonian rel Toda}

\paragraph{General form of equations of motion.} 
Results of the previous section illustrate a long list of Newtonian equations with continuous and discrete time, which can be obtained as symplectic realizations of the flows RTL$_\pm(\alpha)$ and of the maps dRTL$_\pm(\alpha,h)$. All the continuous time systems share the following general form of Newtonian equations of motion:
\begin{equation}\label{rel Toda New gen}
\ddot{x}_k=r(\dot{x}_k)\Big(\dot{x}_{k+1}g(x_{k+1}-x_k)-\dot{x}_{k-1}g(x_k-x_{k-1})+f(x_{k+1}-x_k)-f(x_k-x_{k-1})\Big).
\end{equation}
Systems \eqref{rel Toda New gen} are Lagrangian, with the Lagrange function
\begin{equation}\label{rel Toda New gen Lagr}
\rL(x,\dot{x})=\sum_{k=1}^N K(\dot{x}_k)-\sum_{k=1}^N\dot{x}_kQ(x_{k+1}-x_k)-\sum_{k=1}^N U(x_{k+1}-x_k),
\end{equation}
where $K''(v)=1/r(v)$, $Q'(u)=g(u)$, and $U'(u)=f(u)$. 

For all symplectic realizations of dRTL$_\pm(\alpha,h)$, Newtonian equations of motion have the following general structure:
\begin{equation}\label{d rel Toda New gen}
\psi(\wx_k-x_k)-\psi(x_k-\undertilde{x}_k)=\phi(\undertilde{x}_{k+1}-x_k)-\phi(x_k-\wx_{k-1})+\psi_0(x_{k+1}-x_k)-\psi_0(x_k-x_{k-1}).
\end{equation}
There are two natural ways to realize these equations as discrete Euler-Lagrange equations. 

The first one refers to the discrete Lagrange function 
\begin{equation}\label{d rel Toda New gen Lagr +}
\Lambda_+(x,\wx)=\sum_{k=1}^N\Psi(\wx_k-x_k)-\sum_{k=1}^{N} \Phi(x_{k+1}-\wx_k)-\sum_{k=1}^N\Psi_0(x_{k+1}-x_k),
\end{equation}
where $\Psi'(\xi)=\psi(\xi)$, $\Psi_0'(\xi)=\psi_0(\xi)$, $\Phi'(\xi)=\phi(\xi)$.
Corresponding symplectic map $(x,p)\mapsto(\wx,\wip)$ is given by
\begin{equation} \label{d rel Toda Lagr +}
\left\{\begin{array}{l}
p_k  =  \psi(\wx_k-x_k)+\phi(x_k-\wx_{k-1})+\psi_0(x_k-x_{k-1})-\psi_0(x_{k+1}-x_k), \vspace{1mm}\\
\wip_k = \psi(\wx_k-x_k)+\phi(x_{k+1}-\wx_{k}). 
\end{array}\right.
\end{equation}
This is the general formula for symplectic realizations of dRTL$_+(\alpha,h)$. Corresponding Newtonian equations coincide with \eqref{d rel Toda New gen}. 

The second possibility is 
\begin{equation}\label{d rel Toda New gen Lagr -}
\Lambda_-(x,\wx)=\sum_{k=1}^N\Psi(\wx_k-x_k)-\sum_{k=1}^{N} \Phi(x_{k+1}-\wx_k)-\sum_{k=1}^N\Psi_0(\wx_{k+1}-\wx_k), 
\end{equation}
which generates map $(x,p)\mapsto(\wx,\wip)$ given by
\begin{equation} \label{d rel Toda Lagr -}
\left\{\begin{array}{l}
p_k =  \psi(\wx_k-x_k)+\phi(x_k-\wx_{k-1}), \vspace{1mm}\\
\wip_k = \psi(\wx_k-x_k)+\phi(x_{k+1}-\wx_{k}) -\psi_0(\wx_k-\wx_{k-1})+\psi_0(\wx_{k+1}-\wx_k). 
\end{array}\right.
\end{equation}
This is the general formula for symplectic realizations of dRTL$_-(\alpha,h)$. Corresponding Newtonian equations coincide with \eqref{d rel Toda New gen}, as well.

Note that in all our examples the non-relativistic limit $\alpha=0$ corresponds to $g(u)=0$, resp. $\psi_0(u)=0$, which leads to symplectic realizations of TL, resp. dTL$(h)$.

\paragraph{Realization of the linear bracket: additive exponential relativistic Toda lattice.} Here we reproduce the results of the previous section. A symplectic realization of the bracket $\{\cdot,\cdot\}_{1}$:
\begin{equation}\label{RTL l par again}
b_k=p_k-\alpha\eto{x_k\nm x_{k-1}},\quad a_k=\eto{x_{k+1}\nm x_k}.
\end{equation}

$$
\begin{array}{cl}
 \mathrm{RTL}_+(\alpha): & \  r(v)=1,  \quad f(u)=\eto{u}-\alpha^2\eto{{\scriptstyle 2}u},  \quad  g(u)=\alpha\eto{u},\vspace{2mm}\\
 \mathrm{dRTL}_+(\alpha,h): & \   \psi(v)=\dfrac{1}{h}(\eto{v}-1),  \quad \phi(u)=\dfrac{(h-\alpha)\eto{u}}{1-h\alpha\eto{u}}, \quad \psi_0(u)=\alpha\eto{u},\vspace{3mm}\\
 \mathrm{RTL}_-(\alpha): & \  r(v)=(1-\alpha v)^2,  \quad f(u)=\eto{u},  \quad  g(u)=-\alpha\eto{u},\vspace{2mm}\\
 \mathrm{dRTL}_-(\alpha,h): & \   \psi(v)=\dfrac{h^{-1}(\eto{v}-1)}{1-\alpha h^{-1}(\eto{v}-1)},  \quad \phi(u)=(h+\alpha)\eto{u}, \quad \psi_0(u)=-\alpha\eto{u}. 
\end{array}
$$

\paragraph{Realization of the linear-quadratic bracket: Ruijsenaars Toda lattice.} Relativistic deformation of the exponential Toda lattice discovered by Ruijsenaars comes from the following symplectic realization of the bracket $\{\cdot,\cdot\}_{1}+\alpha\{\cdot,\cdot\}_{2}$, see \eqref{RTL param l br}, \eqref{RTL param q br}:
\begin{equation}\label{RTL m par}
b_k=\frac{1}{\alpha}\big(\eto{\alpha p_k}-1\big),\quad a_k=\eto{x_{k+1}\nm x_k\np\alpha p_k}.
\end{equation}

$$
\begin{array}{cl}
 \mathrm{RTL}_+(\alpha): & \  r(v)=1+\alpha v,  \quad f(u)=\dfrac{\eto{u}}{1+\alpha^2\eto{u}},  \quad  g(u)=\dfrac{\alpha\eto{u}}{1+\alpha^2\eto{u}},\vspace{2mm}\\
 \mathrm{dRTL}_+(\alpha,h): & \   \psi(v)=\dfrac{1}{\alpha}\log\Big(1+\dfrac{\alpha}{h}(\eto{v}-1)\Big),  \quad \phi(u)=-\dfrac{1}{\alpha}\log\Big(1-\alpha(h-\alpha)\eto{u}\Big), \vspace{3mm} \\
 & \  \psi_0(u)=\dfrac{1}{\alpha}\log(1+\alpha^2\eto{u}).
 \end{array}
$$
The ingredients of Newtonian equations for RTL$_-(\alpha)$, dRTL$_-(\alpha,h)$ are obtained from those for RTL$_+(\alpha)$, dRTL$_+(\alpha,h)$ by changing $\alpha$ to $-\alpha$.

\paragraph{Realization of the linear bracket: dual relativistic Toda lattice.} An alternative symplectic realization of of the linear bracket $\{\cdot,\cdot\}_{1}$, see \eqref{RTL param l br}:
\begin{equation}\label{RTL dual par}
b_k=x_k-x_{k-1}-\alpha\eto{p_{k-1}},\quad a_k=\eto{p_k}.
\end{equation} 

$$
\begin{array}{cl}
 \mathrm{RTL}_+(\alpha): & \  r(v)=v,  \quad f(u)=u,  \quad  g(u)=\dfrac{\alpha}{1+\alpha u},\vspace{2mm}\\
 \mathrm{dRTL}_+(\alpha,h): & \   \psi(v)=\log\dfrac{v}{h},  \quad \phi(u)=\log\dfrac{1+hu}{1+\alpha u}, \quad \psi_0(u)=\log(1+\alpha u),\vspace{3mm}\\
 \mathrm{RTL}_-(\alpha): & \  r(v)=v(1+\alpha^2 v),  \quad f(u)=\dfrac{u}{1+\alpha u},  \quad  g(u)=-\dfrac{\alpha}{1+\alpha u},\vspace{2mm}\\
 \mathrm{dRTL}_-(\alpha,h): & \   \psi(v)=\log\dfrac{h^{-1}v}{1+h^{-1}\alpha(\alpha+h)v},  \quad \phi(u)=\log(1+(h+\alpha)u), \vspace{2mm}\\
  & \  \psi_0(u)=-\log(1+\alpha u). 
\end{array}
$$

\paragraph{Realization of the quadratic bracket: modified relativistic Toda lattice.} A symplectic realization of the bracket $\{\cdot,\cdot\}_{2}$,
see \eqref{RTL param q br}:
\begin{equation}\label{RTL q par}
b_k=\eto{p_k}+\eto{x_k\nm x_{k-1}},\quad
a_k=\eto{x_{k+1}\nm x_k\np p_k}.
\end{equation}

$$
\begin{array}{cl}
 \mathrm{RTL}_+(\alpha): & \  r(v)=v,  \quad f(u)=\eto{u},  \quad  g(u)=\dfrac{\alpha\eto{u}}{1+\alpha\eto{u}},\vspace{2mm}\\
 \mathrm{dRTL}_+(\alpha,h): & \   \psi(v)=\log\dfrac{\eto{v}-1}{h},  \quad \phi(u)=\log\dfrac{1+h\eto{u}}{1+\alpha\eto{u}}, \quad \psi_0(u)=\log(1+\alpha\eto{u}),\vspace{3mm}\\
 \mathrm{RTL}_-(\alpha): & \  r(v)=v(1-\alpha v),  \quad f(u)=\dfrac{\eto{u}}{1+\alpha\eto{u}},  \quad  g(u)=-\dfrac{\alpha\eto{u}}{1+\alpha\eto{u}},\vspace{2mm}\\
 \mathrm{dRTL}_-(\alpha,h): & \   \psi(v)=\log\dfrac{h^{-1}(\eto{v}-1)}{1-\alpha h^{-1}(\eto{v}-1)},  \quad \phi(u)=\log\big(1+(h+\alpha)\eto{u}\big), \vspace{2mm}\\
  & \ \psi_0(u)=-\log(1+\alpha\eto{u}). 
\end{array}
$$

\paragraph{Realization of the linear-quadratic bracket II: general exponential relativistic Toda lattice.}
A symplectic realization of the bracket $\{\cdot,\cdot\}_{1}+\epsilon\{\cdot,\cdot\}_{2}$, see \eqref{RTL param l br}, \eqref{RTL param q br}:
\begin{equation}\label{RTL m2 par}
b_k=\frac{1}{\epsilon}\big(\eto{\epsilon p_k}-1\big)+(\epsilon-\alpha)\eto{x_k\nm x_{k-1}},\quad
a_k=\eto{x_{k+1}\nm x_k\np\epsilon p_k}
\end{equation}
The corresponding one-parameter family of symplectic realizations interpolates between the additive exponential relativistic Toda lattices for $\epsilon=0$ and the Ruijsenaars Toda lattices for $\epsilon=\alpha$. This family is the most general relativistic Toda lattice with exponential interactions.

$$
\begin{array}{cl}
 \mathrm{RTL}_+(\alpha): & \  r(v)=1+\epsilon v,  \quad f(u)=\dfrac{\eto{u}+\alpha(\epsilon-\alpha)\eto{{\scriptstyle 2}u}}{1+\epsilon\alpha\eto{u}},  \quad  
                                         g(u)=\dfrac{\alpha\eto{u}}{1+\epsilon\alpha\eto{u}},\vspace{2mm}\\
 \mathrm{dRTL}_+(\alpha,h): & \   \psi(v)=\dfrac{1}{\epsilon}\log\Big(1+\dfrac{\epsilon}{h}(\eto{v}-1)\Big),  \quad 
                                               \phi(u)=\dfrac{1}{\epsilon}\log\dfrac{1+h(\epsilon-\alpha)\eto{u}}{1+\alpha(\epsilon-h)\eto{u}}, \vspace{3mm}\\
                                        & \   \psi_0(u)=\dfrac{1}{\epsilon}\log(1+\epsilon\alpha\eto{u}),\vspace{3mm}\\
 \mathrm{RTL}_-(\alpha): & \  r(v)=(1-\alpha v)\big(1+(\epsilon-\alpha)v\big),  \quad f(u)=\dfrac{\eto{u}}{1+\epsilon\alpha\eto{u}},  \quad  
                                         g(u)=-\dfrac{\alpha\eto{u}}{1+\epsilon\alpha\eto{u}},\vspace{2mm}\\
 \mathrm{dRTL}_-(\alpha,h): & \   \psi(v)=\dfrac{1}{\epsilon}\log\dfrac{1+(\epsilon-\alpha)h^{-1}(\eto{v}-1)}{1-\alpha h^{-1}(\eto{v}-1)},  \quad 
                                              \phi(u)=\dfrac{1}{\epsilon}\log\big(1+\epsilon(h+\alpha)\eto{u}\big), \vspace{2mm}\\
                                        & \  \psi_0(u)=-\dfrac{1}{\epsilon}\log(1+\epsilon\alpha\eto{u}). 
\end{array}
$$

\paragraph{Realization of  the cubic-quadratic bracket: multiplicative hyperbolic relativistic Toda lattice.}
In this paragraph, we restrict ourselves to the ``first'' flow of the
hierarchy (and its discretization), since the resulting equations for the ``negative first'' one are essentially the same (obtained by the change of $\alpha$ to $-\alpha$. 
We assume that $\beta$ is a small parameter and use the following notation:
\[
\alpha_0=-\frac{1}{4\beta}\log(1-4\beta\alpha)=\alpha+O(\beta\alpha^2), \quad 
h_0=-\frac{1}{4\beta}\,\log(1-4h\beta)=h+O(\beta h^2), \quad \epsilon=\frac{1}{4\beta}.
\]
A symplectic realization of the bracket $-\{\cdot,\cdot\}_{3}-4\beta\{\cdot,\cdot\}_{2}$, see  \eqref{RTL param q br}, \eqref{RTL param c br}, is given by
$$
b_k=u_k+v_{k-1}, \quad a_k=u_kv_k,
$$
where
$$
u_k=\frac{y_k(1+\epsilon z_{k-1})}{1-\epsilon\alpha y_{k-1}z_{k-1}}, \; v_k=\frac{z_k(1+\epsilon y_k)}{1-\epsilon\alpha y_kz_k},
$$
and
\begin{equation}\label{MRVL cq par}
y_k=-2\beta\big(\cth(\beta p_k+\beta\alpha_0)+1\big),\quad
z_k=2\beta\big(\cth(x_{k+1}-x_k-\beta\alpha_0)-1\big).
\end{equation}

$$
\begin{array}{cl}
 \mathrm{RTL}_+(\alpha): & \  r(v)=-(v^2-\beta^2), \vspace{2mm} \\
  & \ f(u)=\dfrac{1}{2}\cdot\dfrac{\sh(2u)}{\sh^2(u)-\sh^2(\beta\alpha_0)},   \quad
 g(u)=-\dfrac{1}{2\beta} \cdot\dfrac{\sh(2\beta\alpha_0)}{\sh^2(u)-\sh^2(\beta\alpha_0)},\vspace{3mm}\\
 \mathrm{dRTL}_+(\alpha,h): & \   \psi(v)=\dfrac{1}{2\beta}\log\dfrac{\sinh(v+\beta h_0)}{\sinh(v-\beta h_0)},  \quad 
                                               \phi(u)=\dfrac{1}{2\beta}\log\dfrac{\sinh(u+\beta h_0-\beta\alpha_0)}{\sinh(u-\beta h_0+\beta\alpha_0)},\vspace{3mm}\\
                                        & \   \psi_0(u)=\dfrac{1}{2\beta}\log\dfrac{\sinh(u+\beta\alpha_0)}{\sinh(u-\beta\alpha_0)}.  
\end{array}
$$

\paragraph{Realization of the cubic bracket. I: multiplicative rational relativistic Toda lattice.} A symplectic realization of the bracket $-\{\cdot,\cdot\}_{3}$, see \eqref{RTL param c br}:
\begin{align}\label{RTL c2 par}
a_k  = & \ \frac{1}{\sh^2(p_k)\big(x_k-x_{k-1}+\alpha\,\cth(p_{k-1})\big)\big(x_{k+1}-x_k+\alpha\,\cth(p_k)\big)}, 
\nonumber\\
b_k  = & \ -\frac{\cth(p_{k-1})+\cth(p_k)}{x_k-x_{k-1}+\alpha\,\cth(p_{k-1})}.
\end{align}

$$
\begin{array}{cl}
 \mathrm{RTL}_+(\alpha): & \  r(v)=-(v^2-1),  \quad  f(u)=\dfrac{u}{u^2-\alpha^2},   \quad   g(u)=-\dfrac{\alpha}{u^2-\alpha^2},\vspace{3mm}\\
 \mathrm{dRTL}_+(\alpha,h): & \   \psi(v)=\dfrac{1}{2}\log\dfrac{v+h}{v-h},  \quad 
                                               \phi(u)=\dfrac{1}{2}\log\dfrac{u+h-\alpha}{u-h+\alpha}, \quad
                                               \psi_0(u)=\dfrac{1}{2}\log\dfrac{u+\alpha}{u-\alpha}.  
\end{array}
$$

\paragraph{Realization of the cubic bracket. II: additive rational relativistic Toda lattice.} Another symplectic realization of the bracket $-\{\cdot,\cdot\}_{3}$, see \eqref{RTL param c br}:
\begin{equation}\label{RTL c3 par}
a_k  = \frac{1}{p_k^2\Big(x_k-x_{k-1}+\displaystyle\frac{\alpha}{p_{k-1}}\Big)
\Big(x_{k+1}-x_k+\displaystyle\frac{\alpha}{p_k}\Big)},
\quad b_k =-\frac{1}{x_k-x_{k-1}+\displaystyle\frac{\alpha}{p_{k-1}}}\,
\left(\frac{1}{p_{k-1}}+\frac{1}{p_k}\right)
\end{equation}

$$
\begin{array}{cl}
 \mathrm{RTL}_+(\alpha): & \  r(v)=-v^2,  \quad  f(u)=\dfrac{1}{u},   \quad   g(u)=-\dfrac{\alpha}{u^2},\vspace{3mm}\\
 \mathrm{dRTL}_+(\alpha,h): & \   \psi(v)=\dfrac{h}{v},  \quad   \phi(u)=\dfrac{h-\alpha}{u}, \quad   \psi_0(u)=\dfrac{\alpha}{u}.  
\end{array}
$$

\paragraph{Bibliographic remarks.} In the paper \cite{Rui90} Ruijsenaars  introduced  ``relativistic Toda lattice'' as the system (\ref{RTL+ New introd}).
Discretization of this system was performed in \cite{Su96}. 
In \cite{Su97a, Su03} further systems of relativistic Toda type, both in continuos and in discrete time, were found and identified as symplectic realizations of RTL$_\pm(\alpha)$, resp. of dRTL$_\pm(\alpha,h)$. The latter result provides a complete understanding for all these systems, including the full set of integrals of motion, Lax representations etc. 

In \cite{ASh97a, ASh97b}, a classification of ``integrable'' systems of the type \eqref{rel Toda New gen} was achieved. The notion of ``integrability'' there is based 
on the requirement that the form of equations of motion is preserved under a sort of Legendre transformation. This  is, {\em \`a priori}, unrelated to the classical notion of Liouville-Arnold integrability. However, the resulting list coincides with the list of symplectic realizations of RTL$_\pm(\alpha)$ given in the present section. Thus, {\em \`a posteriori}, all these systems are integrable also in the classical sense. In \cite{A99}, a similar approach was used to perform a classification of ``integrable'' discrete time systems
of the type \eqref{d rel Toda New gen}. Again, the resulting list coincides with the list of symplectic realizations of dRTL$_\pm(\alpha)$ given in the present section.

In the course of classification in \cite{ASh97a, ASh97b, A99}, the authors discovered also systems of a more general form, where the functions $f$, $g$, $\psi$, $\phi$ and $\psi_0$ not necessarily depend on the difference of the arguments. These more general systems are: the {\em relativistic elliptic Toda lattice}, 
$$
\ddot x_k=(\dot x_k^2-1)\Big(
 \dot{x}_{k+1}\frac{\partial \phi(x_k,x_{k+1};\alpha)}{\partial x_{k+1}}
 -\dot{x}_{k-1}\frac{\partial \phi(x_k,x_{k-1};\alpha)}{\partial x_{k-1}} +\frac{\partial \phi(x_k,x_{k+1};\alpha)}{\partial\alpha}
 +\frac{\partial \phi(x_k,x_{k-1};\alpha)}{\partial\alpha}-2\zeta(2x)\Big),
$$
and its time discretization,
$$
 \phi(x_k,\wx_k;h)+\phi(x_k,\undertilde{x}_k;h)+\phi(x_k,\undertilde{x}_{k+1};\alpha-h)+\phi(x_k,\wx_{k-1};\alpha-h)-\phi(x_k,x_{k+1};\alpha)-\phi(x_k,x_{k-1};\alpha)=0,
$$
where $\phi(x_0,x_1;\alpha)$ is given in \eqref{phi elliptic}. The theory of the latter systems was developed in \cite{ASu04}. 

%%%%%%%%%%%%%%%%%%%%%%%%%%%%%%%%%%%%%%%%%%%%%%%%%%
%%%%%%%%%%%%%%%%%%%%%%%%%%%%%%%%%%%%%%%%%%%%%%%%%%

\section{Explicit discretizations of the Toda lattice}
\label{sect explicit}

An important particular case of the map dRTL$_+(\alpha,h)$ appears when the discrete time step $h$ coincides with the relativistic parameter $\alpha$. Continued fractions for $\ga_k$, defined by \eqref{dRTL+ param a}, obviously degenerate to explicit expressions $\ga_k=1+hb_k$. The reason for this is clear, since in this case $\mbB_{1,2}=\Pi_+(T_{1,2})$, so that, in particular, $\mbB_1=\Pi_+(LU^{-1})=L$. Formulas \eqref{dRTL+ param} turn into
\begin{equation}
1+h\wb_k=(1+hb_{k-1}) \frac{1+hb_k+h^2a_k}{1+hb_{k-1}+h^2a_{k-1}}, \quad \wa_k=a_k \frac{1+hb_{k+1}+h^2a_{k+1}}{1+hb_k+h^2a_k}.
\end{equation}
Thus, the map dRTL$_+(h,h)$ is an explicit rational map. Moreover, it is birational; to found the inverse map, it is sufficient to observe that $1+h\wb_k+h^2\wa_{k-1}=1+hb_k+h^2a_k$. Of course, the map dRTL$_+(h,h)$ belongs to the hierarchy RTL$(h)$; however, it serves as a discretization of TL. Indeed, as $h\to 0$, we have:
$$
\wb_k=b_k+h(a_k-a_{k-1})+O(h^2), \quad \wa_k=a_k+ha_k(b_{k+1}-b_k)+O(h^2).
$$

A similar simplification takes place for the map dRTL$_-(\alpha,h)$ if $h=-\alpha$. Continued fractions for $\gd_k$, defined by \eqref{dRTL- param d}, obviously degenerate to explicit expressions $\gd_k=a_k$. The reason is that in this case $\mbC_{1,2}=\Pi_-(T_{1,2}^{-1})$, so that, in particular, $\mbC_2=\Pi_-(L^{-1}U)=U$. Formulas \eqref{dRTL- param} turn into
\begin{equation}
1-h\wb_k=(1-hb_{k+1}) \frac{1-hb_k+h^2a_{k-1}}{1-hb_{k+1}+h^2a_{k}}, \quad \wa_k=a_k \ \frac{1-hb_{k}+h^2a_{k-1}}{1-hb_{k+1}+h^2a_k}.
\end{equation}
This map dRTL$_-(-h,h)$ is an explicit rational map which belongs to the hierarchy RTL$(-h)$ and serves as a discretization of TL.

We now consider various symplectic realizations of these explicit maps. Naturally, they are explicit discretizations of the corresponding symplectic realizations of the flow TL. We restrict ourselves to the map dRTL$_+(h,h)$. Newtonian equations of motion have the following general structure:
\begin{equation}\label{d rel Toda spec New gen}
\psi(\wx_k-x_k;h)-\psi(x_k-\undertilde{x}_k;h)=\psi_0(x_{k+1}-x_k;h)-\psi_0(x_k-x_{k-1};h).
\end{equation}
They serve as Euler-Lagrange equations for the discrete time Lagrange function 
\begin{equation}\label{d rel Toda spec New gen Lagr}
\Lambda(x,\wx;h)=\sum_{k=1}^N\Psi(\wx_k-x_k;h)-\sum_{k=1}^N\Psi_0(x_{k+1}-x_k;h).
\end{equation}
The corresponding symplectic map $(x,p)\mapsto(\wx,\wip)$ is given by
\begin{equation}\label{d rel Toda spec Lagr}
\left\{\begin{array}{l}
p_k  =  \psi(\wx_k-x_k;h)+\psi_0(x_k-x_{k-1};h)-\psi_0(x_{k+1}-x_k;h), \vspace{1mm}\\
\wip_k = \psi(\wx_k-x_k;h).
\end{array}\right.
\end{equation}

a) For the pull-back of the map $\mathrm{dRTL}_+(h,h)$ under parametrization \eqref{RTL l par again} with $\alpha=h$:
\begin{equation}\label{dRTL+ l spec}
\psi(v;h)=\frac{1}{h}(\eto{v}-1), \quad \psi_0(u;h)=h\eto{u}.
\end{equation}

b) For the pull-back of the map $\mathrm{dRTL}_+(h,h)$ under parametrization \eqref{RTL m par} with $\alpha=h$:
\begin{equation} \label{dRTL+ m spec}
\psi(v;h)=\frac{v}{h}, \quad \psi_0(u;h)=\frac{1}{h}\log(1+h^2\eto{u}).
\end{equation}

c) For the pull-back of the map $\mathrm{dRTL}_+(h,h)$ under parametrization \eqref{RTL dual par} with $\alpha=h$:
\begin{equation}\label{dRTL+ dual spec}
\psi(v;h)=\log\frac{v}{h},\quad \psi_0(u;h)=\log(1+hu).
\end{equation}

d) For the pull-back of the map $\mathrm{dRTL}_+(h,h)$ under  parametrization \eqref{RTL q par} with $\alpha=h$:
\begin{equation}\label{dRTL+ q spec} 
\psi(v;h)=\log\frac{\eto{v}-1}{h}, \quad \psi_0(u;h)=\log(1+h\eto{u}).
\end{equation}

e) For the pull-back of the map $\mathrm{dRTL}_+(h,h)$ under parametrization \eqref{MRVL cq par} with $\alpha=h$:
\begin{equation}\label{dRTL cq spec}
\psi(v;h)=\frac{1}{2\beta}\log\frac{\sinh(v+\beta h_0)}{\sinh(v-\beta h_0)},\quad \psi_0(u;h)=\frac{1}{2\beta}\log\frac{\sinh(u+\beta h_0)}{\sinh(u-\beta h_0)}.
\end{equation}

f) For the pull-back of the map $\mathrm{dRTL}_+(h,h)$ under parametrization \eqref{RTL c2 par} with $\alpha=h$:
\begin{equation}\label{dRTL c2 spec}
\psi(v;h)=\frac{1}{2}\log\frac{v+h}{v-h},\quad \psi_0(u;h)=\frac{1}{2}\log\frac{u+h}{u-h}.
\end{equation}

g) For the pull-back of the map $\mathrm{dRTL}_+(h,h)$ under parametrization \eqref{RTL c3 par} with $\alpha=h$:
\begin{equation}\label{dRTL c3 spec}
\psi(v;h)=\frac{h}{v},\quad \psi_0(u;h)=\frac{h}{u}.
\end{equation}

\paragraph{Bibliographical remarks.} The fact that explicit discretization of Toda lattice belongs to the relativistic Toda hierarchy was pointed out in \cite{Su90b}. It was used there to introduce Lie-algebraic generalizations of relativistic Toda systems.

\setcounter{equation}{0}
\section{Discrete 2D equations of Laplace type}
\label{sect: discr Toda}

We elaborate now on the common features of our discrete time systems.

\paragraph{Implicit discretizations of Toda lattices.}

In equation \eqref{dToda New gen}, we assume that 
\begin{equation}\label{ident 2D}
x_k=x_k(t)=x_{k,n}, \quad \wx_k=x_k(t+h)=x_{k,n+1}, \quad \undertilde{x}_k=x_k(t-h)=x_{k,n-1}
\end{equation}
for $t=nh$ with $n\in \mathbb Z$. Thus, equation \eqref{dToda New gen} can be interpreted as equation 
\begin{equation}\label{dToda New gen 2D}
\psi(x_{k,n+1}-x_{k,n};h)-\psi(x_{k,n}-x_{k,n-1};h)=\phi(x_{k+1,n-1}-x_{k,n};h)-\phi(x_{k,n}-x_{k-1,n+1};h)
\end{equation}
for $x:\mathbb Z^2\to \bbR$. This is a  {\em 2D lattice equation}. If we visualize such an equation by connecting all pairs of vertices which appear in its individual terms, then we end up with a stencil shown on Figure \ref{Fig: square lattice skew}\subref{Fig: dTL stencil}.

Introduce the graph $\Gamma$ with the set of vertices $V(\Gamma)=\mathbb Z^2$ and with the set of edges $E(\Gamma)$ connecting nearest neighbors in the south-to-north direction and in the south-east-to-north-west direction. Clearly, this graph is combinatorially nothing but the regular square lattice, but drawn in $\mathbb R ^2$ in a non-standard way. Equation \eqref{dToda New gen 2D} lives on vertex stars of $\Gamma$. 
\medskip

Moreover, equation \eqref{dToda New gen 2D} admits a 2D Lagrangian formulation. 
\begin{definition} \label{def: discr Toda}
Let $\Gamma$ be a graph with the set of vertices $V=V(\Gamma)$ and with the set of (directed) edges $E=E(\Gamma)$.
Let $\mathcal L: E(\Gamma)\to \bbR$ be {\itbf edge Lagrangians} defined on edges $e=(v_1,v_2)$ of the graph $\Gamma$ and depending on the values $x(v_1)$, $x(v_2)$ of a function $x: V(\Gamma)\to \bbR$ at the endpoints of $e$. Introduce the {\itbf action functional} $S[x]$ for any function $x: V(\Gamma)\to \bbR$, as a (formal) sum of Lagrangians over all edges of $\Gamma$,
\begin{equation}  \label{action over Gamma}
S[x]=\sum_{e=(v_1,v_2)\in E(\Gamma)} \mathcal L(e),
\end{equation}
A {\itbf discrete Laplace type system} on the graph $\Gamma$ consists of Euler-Lagrange equations for critical points $x:V(\Gamma)\to \bbR$ of the action $S$, that is, of equations
\begin{equation} \label{dToda over Gamma}
\frac{\partial S}{\partial x(v)}=0\quad \forall v\in V(\Gamma).
\end{equation}
\end{definition}

Clearly, for a given vertex $v\in V(\Gamma)$, the expression $\partial S/\partial x(v)$ involves the derivatives of discrete Lagrangians for all edges incident to $v$, i.e., for the vertex star of $v$.  

Equation \eqref{dToda New gen 2D} is Euler-Lagrange equation for the action functional on $\mathbb Z^2$,
\begin{equation}
S[x]=\sum_{(k,n)\in\mathbb Z^2}\Big(\Psi(x_{k,n+1}-x_{n,k};h)-\Phi(x_{k+1,n}-x_{k,n};h)\Big),
\end{equation}
where edge Lagrangians are given by
$$
\cL(e)=\left\{\begin{array}{ll}
\Psi(x_{k,n+1}-x_{k,n};h) & \mathrm{for\;\;} e=\big((k,n),(k,n+1)\big),\\
-\Phi(x_{k,n+1}-x_{k+1,n};h) & \mathrm{for\;\;} e=\big((k+1,n),(k,n+1)\big).
\end{array}\right.
$$

%%%%%%%%%%%%%%%%%%%%%%%%%%%%%%%%%%%%%%%%%%%%%%%%%%%%%%%%%%%%%%%%%
\begin{figure}[tbp]
   \centering
   \subfloat[]{\label{Fig: dTL stencil}
   \begin{tikzpicture}[scale=0.8,inner sep=2]
      \node (x) at (2,0) [circle,draw] {};
      \node (x1) at (4,0) [circle,fill,label=-90:${x}_{k,n-1}$] {};
      \node (x11) at (6,0) [circle,fill,label=-90:${x}_{k+1,n-1}$] {};
      \node (x2) at (2,2) [circle,draw,label=45:$x_{k-1,n}$] {};
      \node (x12) at (4,2) [circle,fill,label=45:$x_{k,n}$] {};
      \node (x112) at (6,2) [circle,draw,label=45:$x_{k+1,n}$] {};
      \node (x22) at (2,4) [circle,fill,label=90:${x}_{k-1,n+1}$] {};
      \node (x122) at (4,4) [circle,fill,label=90:${x}_{k,n+1}$] {};
      \node (x1122) at (6,4) [circle,draw]{};     
      \draw [dashed] (1.5,0.5) to (x)  to (x2) to (x1); 
      \draw [very thick] (x1) to (x12) to (x11);
      \draw [dashed] (x11) to (x112) to (6.5,1.5);
      \draw [dashed](1.5,2.5) to  (x2) to (x22) ;
      \draw [very thick] (x22) to (x12) to (x122);
      \draw [dashed] (x122) to (x112) to (x1122) to (6.5,3.5);
   \end{tikzpicture}
   }\qquad
   \subfloat[]{\label{Fig: dTL Lagrangian}
   \begin{tikzpicture}[scale=0.8,inner sep=2]
      \node (x1) at (2,0) [circle,draw] {};
      \node (x11) at (4,0) [circle,draw,label=-90:$\undertilde{x}_{k}$] {};
      \node (x111) at (6,0) [circle,draw,label=-90:$\undertilde{x}_{k+1}$] {};
      \node (x1111) at (8,0) [circle,draw] {};
      \node (x2) at (0,2) [circle,fill] {};
      \node (x12) at (2,2) [circle,fill,label=45:$x_{k-1}$] {};
      \node (x112) at (4,2) [circle,fill,label=45:$x_{k}$] {};
      \node (x1112) at (6,2) [circle,fill,label=45:$x_{k+1}$] {};
      \node (x11112) at (8,2) [circle,fill] {};
      \node (x22) at (0,4) [circle,fill] {};
      \node (x122) at (2,4) [circle,fill,label=90:$\widetilde{x}_{k-1}$] {};
      \node (x1122) at (4,4) [circle,fill,label=90:$\widetilde{x}_{k}$] {};
      \node (x11122) at (6,4) [circle,fill] {};
      \draw [dashed] (0,1.5) to (x2) to (x1) to (x12) to (x11) to (x112) to (x111) to (x1112) to (x1111) to (x11112) to (8.5,1.5);
      \draw [very thick](-0.5,2.5) to (x2) to (x22) to (x12) to (x122) to (x112) to (x1122) to (x1112) to (x11122) to (x11112) to (8,2.5);
   \end{tikzpicture}
   }
   
   \caption{Combinatorics of implicit discretizations of Toda lattices. \protect\subref{Fig: dTL stencil} 5-point vertex star of the graph $\Gamma$ supporting equation \eqref{dToda New gen 2D};  \protect\subref{Fig: dTL Lagrangian} Cauchy slice $\Gamma^n$ of the 2D square lattice supporting the discrete time Lagrange function $\Lambda(x,\wx)$ from \eqref{dToda New gen Lagr}.}
  \label{Fig: square lattice skew} 
\end{figure}
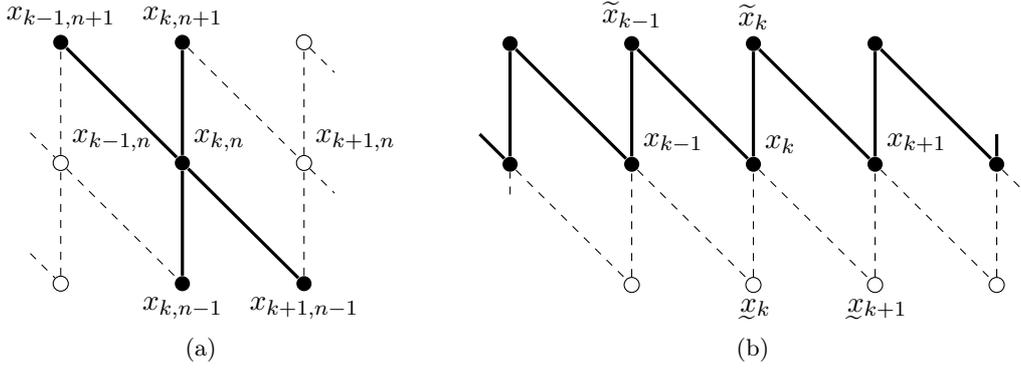
%%%%%%%%%%%%%%%%%%%%%%%%%%%%%%%%%%%%%%%%%%%%%%%%%%%%%%%%%%%%%%%%%

\paragraph{From 2D to 1D: discrete time evolution.} Formulation of Definition \ref{def: discr Toda} is covariant in the sense that there is no a priori privileged evolution direction on a general graph $\Gamma$ with $V(\Gamma)=\mathbb Z^2$. The time evolution arises if we consider a specific {\em Cauchy problem} on a graph. This can be formalized as follows.  

\begin{definition}\label{def: slicing} 
Let the vertices of a graph $\Gamma$ be represented as a disjoint union $V(\Gamma)=\bigsqcup_{n\in \mathbb Z}V^n$ in such a way that any two $V^n$ and $V^m$ are in a one-to-one correspondence. A {\itbf space-time splitting (slicing)} of $\Gamma$ is a sequence of subgraphs $\{\Gamma^n\}_{n\in\mathbb Z}$ ({\itbf Cauchy slices}) such that: 
\begin{itemize}
\item $V(\Gamma^n)=V^{n}\cup V^{n+1}$, 
\item the sets $E(\Gamma^n)$ are disjoint,  
\item $E(\Gamma)=\bigsqcup_{n\in\mathbb Z}E(\Gamma^n)$. 
\end{itemize}
The {\itbf discrete time Lagrange function} corresponding to the slicing $\{\Gamma^n\}$ is obtained by summing up discrete Lagrangians over the edges of one slice $\Gamma^n$:
\begin{equation}\label{Lagr slice}
 \Lambda(x,\wx)=\sum_{e\in E(\Gamma^n)}\cL(e).
\end{equation}
\end{definition}
Let us comment on the notation used on the left-hand side of \eqref{Lagr slice}.
For any fixed $n$, we denote the space of functions on $V^n$ by $\cX=\{x:V^n\to\mathbb R\}$ (this space does not depend on $n$ because any two $V^n$, $V^m$ are in a bijection). We take $\cX$ as the configuration space of our discrete time system.
Functions on $V(\Gamma^n)$ are naturally identified with pairs of functions $(x,\wx)$, where $x=x(n)$ is a function on $V^{n}$
and $\wx=x(n+1)$ is a function on $V^{n+1}$. Thus, the space of functions on on $V(\Gamma^n)$ is identified with $\cX\times\cX$. We have:
\[
S=\sum_{n\in \mathbb Z} \Lambda\big(x(n),x(n+1)\big).
\]
There follows from Definition \ref{def: slicing} that each edge of $\Gamma$ incident to a vertex $v\in V^n$ belongs to exactly one of the slices $E(\Gamma^n)$ and $E(\Gamma^{n-1})$. Therefore, for any vertex $v\in V^n$, the part of the action $S[x]$ containing edges incident to $v$ is 
\[
\sum_{e\in E(\Gamma^n)\sqcup E(\Gamma^{n-1})}\cL(e) =\Lambda\big(x(n),x(n+1)\big)+\Lambda\big(x(n-1),x(n)\big)=\Lambda(x,\wx)+\Lambda(\undertilde{x},x).
\]
It follows that discrete Laplace type equations \eqref{dToda over Gamma} coincide with the discrete time Euler-Lagrange equations \eqref{dEL gen}.
Equations \eqref{dToda map gen} produce out of this system of second order difference equations a symplectic map on $T^*\mathcal X$.
\medskip

For discrete Toda systems of the type \eqref{dToda New gen} with open-end or periodic boundary conditions, $V(\Gamma)=\{1,2,\ldots,N\}\times \mathbb Z$, and $V^n=\{1,2,\ldots,N\}\times\{n\}$. The edges of $\Gamma$ connect nearest neighbors in the south-to-north direction and in the south-east-to-north-west direction. The relevant slicing is shown on Figure \ref{Fig: square lattice skew}\subref{Fig: dTL Lagrangian}. The space $\cX=\mathbb R^N$ consists of finite sequences $x=(x_k)_{k=1}^N$. Equations \eqref{dToda map gen} define a symplectic map on $T^*\bbR^N$.

\paragraph{\textbf{Explicit discretizations of Toda lattices.}}
We now give a similar consideration of symplectic realizations of the explicit rational map dRTL$_+(h,h)$, listed in Section \ref {sect explicit}. 
Under identification \eqref{ident 2D}, equation \eqref{d rel Toda spec New gen} becomes a {\em 2D lattice equation}
\begin{equation}\label{d rel Toda spec New gen 2D}
\psi(x_{k,n+1}-x_{k,n};h)-\psi(x_{k,n}-{x}_{k,n-1};h)= \psi_0(x_{k+1,n}-x_{k,n};h)-\psi_0(x_{k,n}-x_{k-1,n};h)
\end{equation}
for a function $x:\mathbb Z^2\to\mathbb R$. The stencil supporting this equation is shown on Figure \ref{Fig: square lattice straight}\subref{Fig: dRTL spec stencil}.

%%%%%%%%%%%%%%%%%%%%%%%%%%%%%%%%%%%%%%%%%%%%%%%%%%%%%%%%%%%%%%%%%
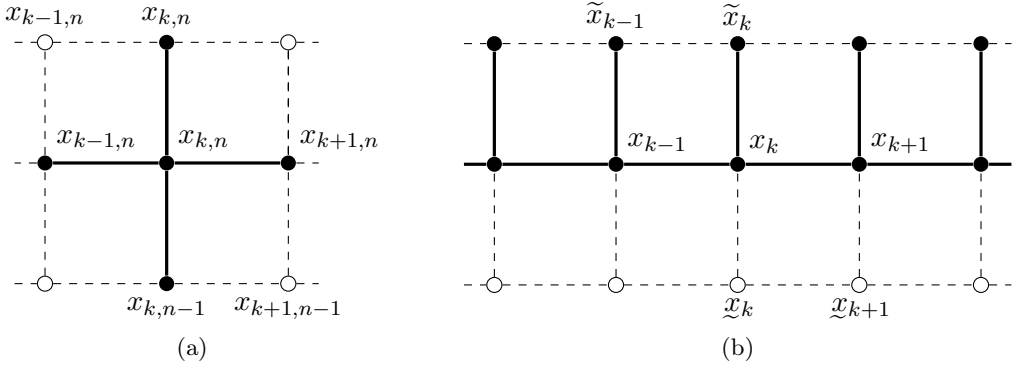
\begin{figure}[htbp]
   \centering
   \subfloat[]{ \label{Fig: dRTL spec stencil}
   \begin{tikzpicture}[scale=0.8,inner sep=2]  
      \node (x) at (2,0) [circle,draw] {};
      \node (x1) at (4,0) [circle,fill,label=-90:${x}_{k,n-1}$] {};
      \node (x11) at (6,0) [circle,draw,label=-90:${x}_{k+1,n-1}$] {};
      \node (x2) at (2,2) [circle,fill,label=45:$x_{k-1,n}$] {};
      \node (x12) at (4,2) [circle,fill,label=45:$x_{k,n}$] {};
      \node (x112) at (6,2) [circle,fill,label=45:$x_{k+1,n}$] {};
      \node (x22) at (2,4) [circle,draw,label=90:${x}_{k-1,n}$] {};
      \node (x122) at (4,4) [circle,fill,label=90:${x}_{k,n}$] {};
      \node (x1122) at (6,4) [circle,draw]{};     
      \draw [dashed] (1.5,0) to (x) to (x1) to (x11) to (6.5,0);
      \draw [dashed] (x)  to (x2) to (1.5,2) ;
      \draw [very thick] (x1) to (x12);
      \draw [dashed] (x11) to (x112)  to (6.5,2);
      \draw [very thick]  (x2) to (x12) to (x112);
      \draw [dashed] (x2) to (x22) ;
      \draw [very thick] (x122) to (x12) to (x122);
      \draw [dashed] (x1122) to (x112) to (x1122);
      \draw [dashed] (1.5,4) to (x22) to (x122) to (x1122) to (6.5,4);
   \end{tikzpicture}
   }
   \qquad
   \subfloat[]{\label{Fig: dRTL spec Lagrangian}
   \begin{tikzpicture}[scale=0.8,inner sep=2]
      \node (x) at (0,0) [circle,draw] {};
      \node (x1) at (2,0) [circle,draw] {};
      \node (x11) at (4,0) [circle,draw,label=-90:$\undertilde{x}_{k}$] {};
      \node (x111) at (6,0) [circle,draw,label=-90:$\undertilde{x}_{k+1}$] {};
      \node (x1111) at (8,0) [circle,draw] {};
      \node (x2) at (0,2) [circle,fill] {};
      \node (x12) at (2,2) [circle,fill,label=45:$x_{k-1}$] {};
      \node (x112) at (4,2) [circle,fill,label=45:$x_{k}$] {};
      \node (x1112) at (6,2) [circle,fill,label=45:$x_{k+1}$] {};
      \node (x11112) at (8,2) [circle,fill] {};
      \node (x22) at (0,4) [circle,fill] {};
      \node (x122) at (2,4) [circle,fill,label=90:$\widetilde{x}_{k-1}$] {};
      \node (x1122) at (4,4) [circle,fill,label=90:$\widetilde{x}_{k}$] {};
      \node (x11122) at (6,4) [circle,fill] {};
       \node (x111122) at (8,4) [circle,fill] {};
      \draw [dashed] (-0.5,0) to (x) to (x1) to (x11) to (x111) to (x1111) to (8.5,0);
      \draw [dashed] (x) to (x2)    (x1) to (x12)   (x11) to (x112)   (x111) to (x1112)   (x1111) to (x11112) ;
      \draw  [very thick](-0.5,2) to (x2) to (x12) to (x112) to (x1112) to (x11112) to (8.5,2);
      \draw [very thick] (x2) to (x22)   (x12) to (x122)   (x112) to (x1122)   (x1112) to (x11122)   (x11112) to (x111122);
      \draw [dashed] (-0.5,4) to (x22) to (x122) to (x1122) to (x11122) to (x111122) to (8.5,4);
   \end{tikzpicture}
   }
 \caption{Combinatorics of explicit discretizations of Toda lattices. \protect\subref{Fig: dRTL spec stencil} 5-point vertex star supporting equation \eqref{d rel Toda spec New gen 2D} \protect\subref{Fig: dRTL spec Lagrangian} Cauchy slice supporting discrete time Lagrange function \eqref{d rel Toda spec New gen Lagr} }  
 \label{Fig: square lattice straight}
\end{figure}
%%%%%%%%%%%%%%%%%%%%%%%%%%%%%%%%%%%%%%%%%%%%%%%%%%%%%%%%%%%%%%%%%

Equation \eqref{d rel Toda spec New gen} is a {\em discrete Laplace type system} on the regular square lattice, according to Definition \ref{def: discr Toda}, because it is Euler-Lagrange equation for the {\em action functional} on $\mathbb Z^2$,
\begin{equation} \label{action dRTL spec 2d}
S[x]=\sum_{(k,n)\in \mathbb Z^2} \Big(\Psi(x_{k,n+1}-x_{k,n};h)-\Psi_0(x_{k+1,n}-x_{k,n};h)\Big),
\end{equation}
where $\Psi'(\xi;h)=\psi(\xi;h)$, $\Psi_0'(\xi;h)=\psi_0(\xi;h)$.
Thus, discrete Lagrangians are given by
$$
\cL(e)=\left\{\begin{array}{ll}
\Psi(x_{k,n+1}-x_{k,n};h) & \mathrm{for\;\;} e=\big((k,n), (k,n+1)\big),\\
-\Psi_0(x_{k+1,n}-x_{k,n};h) & \mathrm{for\;\;} e=\big((k,n),(k+1,n)\big).
\end{array}\right.
$$
The space-time splitting of the regular square lattice leading to the discrete time Lagrange function \eqref{d rel Toda spec New gen Lagr}
is shown on Figure \ref{Fig: square lattice straight}\subref{Fig: dRTL spec Lagrangian}.

%%%%%%%%%%%%%%%%%%%%%%%%%%%%%%%%

\paragraph{\textbf{Discretizations of relativistic Toda lattices.}}

Under the usual identification \eqref{ident 2D}, equation \eqref{d rel Toda New gen} becomes a {\em 2D lattice equation}
\begin{align}\label{d rel Toda New gen 2D}
\psi(x_{k,n+1}-x_{k,n})-\psi(x_{k,n}-{x}_{k,n-1})= & \phi({x}_{k+1,n-1}-x_{k,n})-\phi(x_{k,n}-x_{k-1,n+1})\nonumber\\
 &\quad +\psi_0(x_{k+1,n}-x_{k,n})-\psi_0(x_{k,n}-x_{k-1,n})
\end{align}
for a function $x:\mathbb Z^2\to\mathbb R$. If we visualize this equation by connecting all pairs of vertices which appear in its individual terms, then we arrive at the stencil shown on Figure \ref{Fig: dRTL stencil}.

%%%%%%%%%%%%%%%%%%%%%%%%%%%%%%%%%%%%%%%%%%%%%%%%%%%%%%%%%%%%%%%%%
\begin{figure}[tbp]
   \centering
   \begin{tikzpicture}[scale=0.8,inner sep=2]  
      \node (x1) at (2,0) [circle,draw] {};
      \node (x11) at (4,0) [circle,fill,label=-90:${x}_{k,n-1}$] {};
      \node (x111) at (6,0) [circle,fill,label=-90:${x}_{k+1,n-1}$] {};
      \node (x12) at (2,2) [circle,fill,label=45:$x_{k-1,n}$] {};
      \node (x112) at (4,2) [circle,fill,label=45:$x_{k,n}$] {};
      \node (x1112) at (6,2) [circle,fill,label=45:$x_{k+1,n}$] {};
      \node (x122) at (2,4) [circle,fill,label=90:${x}_{k-1,n+1}$] {};
      \node (x1122) at (4,4) [circle,fill,label=90:${x}_{k,n+1}$] {};
      \node (x11122) at (6,4) [circle,draw]{};     
      \draw [dashed] (1.5,0) to (x1) to (x11) to (x111) to (6.5,0);
      \draw [dashed] (1.5,0.5) to (x1)  to (x12) to (x11); 
      \draw [very thick] (x11) to (x112) to (x111);
      \draw [dashed] (x111) to (x1112) to (6.5,1.5);
      \draw [dashed]  (1.5,2) to (x12)  (x1112) to (6.5,2);
      \draw [very thick]  (x12) to (x112) to (x1112);
      \draw [dashed](1.5,2.5) to  (x12) to (x122) ;
      \draw [very thick] (x122) to (x112) to (x1122);
      \draw [dashed] (x1122) to (x1112) to (x11122) to (6.5,3.5);
      \draw [dashed] (1.5,4) to (x122) to (x1122) to (x11122) to (6.5,4);
   \end{tikzpicture}
        \caption{7-point vertex star of the graph $\Gamma$ supporting equation \eqref{d rel Toda New gen 2D}}
           \label{Fig: dRTL stencil}
\end{figure}
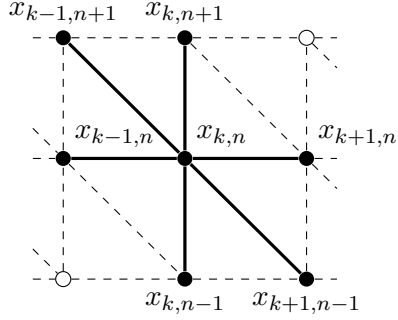
%%%%%%%%%%%%%%%%%%%%%%%%%%%%%%%%%%%%%%%%%%%%%%%%%%%%%%%%%%%%%%%%%

Consider the graph $\Gamma$ with the set of vertices $V(\Gamma)=\mathbb Z^2$, and with the set of edges $E(\Gamma)$ connecting nearest neighbors in the south-to-north direction, in the west-to-east direction, and in the south-east-to-north-west direction. Clearly, this graph is combinatorially nothing but the regular triangular lattice. Equation
\eqref{d rel Toda New gen 2D} lives on vertex stars of $\Gamma$. 
Moreover, equation \eqref{d rel Toda New gen} is a {\em discrete Laplace type system} on the regular triangular lattice, according to Definition \ref{def: discr Toda}, because it is Euler-Lagrange equation for the {\em action functional} on $\mathbb Z^2$,
\begin{equation} \label{action dRTL 2d}
S[x]=\sum_{(k,n)\in \mathbb Z^2} \Big(\Psi(x_{k,n+1}-x_{k,n})-\Psi_0(x_{k+1,n}-x_{k,n})-\Phi(x_{k+1,n}-x_{k,n+1})\Big),
\end{equation}
where $\Psi'(\xi)=\psi(\xi)$, $\Psi_0'(\xi)=\psi_0(\xi)$, $\Phi'(\xi)=\phi(\xi)$.
Thus, discrete Lagrangians are given by
$$
\cL(e)=\left\{\begin{array}{ll}
\Psi(x_{k,n+1}-x_{k,n}) & \mathrm{for\;\;} e=\big((k,n), (k,n+1)\big),\\
-\Psi_0(x_{k+1,n}-x_{k,n}) & \mathrm{for\;\;} e=\big((k,n),(k+1,n)\big),\\
-\Phi(x_{k,n+1}-x_{k+1,n}) & \mathrm{for\;\;} e=\big((k+1,n),(k,n+1)\big).
\end{array}\right.
$$

There are two natural ways to introduce a space-time splitting of the regular triangular lattice $\Gamma$.  A slice of the first slicing is shown on Figure \ref{Fig: dRTL slices}\subref{Fig: dRTL Lagrangian +}. Summing up discrete Lagrangians over the edges of one slice $\Gamma^n$, we come to the discrete time Lagrange function \eqref{d rel Toda New gen Lagr +} for maps dRTL$_+(\alpha,h)$. A slice of an alternative slicing is shown on Figure \ref{Fig: dRTL slices}\subref{Fig: dRTL Lagrangian -}. It supports discrete time Lagrange function \eqref{d rel Toda New gen Lagr -} for maps dRTL$_-(\alpha,h)$.

%%%%%%%%%%%%%%%%%%%%%%%%%%%%%%%%%%%%%%%%%%%%%%%%%%%%%%%%%%%%%%%%%
\begin{figure}[tbp]
   \centering
   \subfloat[Slice supporting Lagrange function \eqref{d rel Toda New gen Lagr +} ]{\label{Fig: dRTL Lagrangian +}
   \begin{tikzpicture}[scale=0.7,inner sep=2]
      \node (x1) at (2,0) [circle,draw] {};
      \node (x11) at (4,0) [circle,draw,label=-90:$\undertilde{x}_{k}$] {};
      \node (x111) at (6,0) [circle,draw,label=-90:$\undertilde{x}_{k+1}$] {};
      \node (x1111) at (8,0) [circle,draw] {};
      \node (x2) at (0,2) [circle,fill] {};
      \node (x12) at (2,2) [circle,fill,label=45:$x_{k-1}$] {};
      \node (x112) at (4,2) [circle,fill,label=45:$x_{k}$] {};
      \node (x1112) at (6,2) [circle,fill,label=45:$x_{k+1}$] {};
      \node (x11112) at (8,2) [circle,fill] {};
      \node (x22) at (0,4) [circle,fill] {};
      \node (x122) at (2,4) [circle,fill,label=90:$\widetilde{x}_{k-1}$] {};
      \node (x1122) at (4,4) [circle,fill,label=90:$\widetilde{x}_{k}$] {};
      \node (x11122) at (6,4) [circle,fill] {};
      \draw [dashed] (1.5,0) to (x1) to (x11) to (x111) to (x1111) to (8.5,0);
      \draw [dashed] (0,1.5) to (x2) to (x1) to (x12) to (x11) to (x112) to (x111) to (x1112) to (x1111) to (x11112) to (8.5,1.5);
      \draw  [very thick](-0.5,2) to (x2) to (x12) to (x112) to (x1112) to (x11112) to (8.5,2);
      \draw [very thick](-0.5,2.5) to (x2) to (x22) to (x12) to (x122) to (x112) to (x1122) to (x1112) to (x11122) to (x11112) to (8,2.5);
      \draw [dashed] (-0.5,4) to (x22) to (x122) to (x1122) to (x11122) to (6.5,4);
   \end{tikzpicture}
   }\qquad
   \subfloat[Slice supporting Lagrange function \eqref{d rel Toda New gen Lagr -}]{\label{Fig: dRTL Lagrangian -}
   \begin{tikzpicture}[scale=0.7,inner sep=2]
      \node (x1) at (2,0) [circle,draw] {};
      \node (x11) at (4,0) [circle,draw,label=-90:$\undertilde{x}_{k}$] {};
      \node (x111) at (6,0) [circle,draw,label=-90:$\undertilde{x}_{k+1}$] {};
      \node (x1111) at (8,0) [circle,draw] {};
      \node (x2) at (0,2) [circle,fill] {};
      \node (x12) at (2,2) [circle,fill,label=45:$x_{k-1}$] {};
      \node (x112) at (4,2) [circle,fill,label=45:$x_{k}$] {};
      \node (x1112) at (6,2) [circle,fill,label=45:$x_{k+1}$] {};
      \node (x11112) at (8,2) [circle,fill] {};
      \node (x22) at (0,4) [circle,fill] {};
      \node (x122) at (2,4) [circle,fill,label=90:$\widetilde{x}_{k-1}$] {};
      \node (x1122) at (4,4) [circle,fill,label=90:$\widetilde{x}_{k}$] {};
      \node (x11122) at (6,4) [circle,fill] {};
      \draw [dashed] (1.5,0) to (x1) to (x11) to (x111) to (x1111) to (8.5,0);
      \draw [dashed] (0,1.5) to (x2) to (x1) to (x12) to (x11) to (x112) to (x111) to (x1112) to (x1111) to (x11112) to (8.5,1.5);
      \draw  [dashed](-0.5,2) to (x2) to (x12) to (x112) to (x1112) to (x11112) to (8.5,2);
      \draw [very thick](-0.5,2.5) to (x2) to (x22) to (x12) to (x122) to (x112) to (x1122) to (x1112) to (x11122) to (x11112) to (8,2.5);
      \draw [very thick] (-0.5,4) to (x22) to (x122) to (x1122) to (x11122) to (6.5,4);
   \end{tikzpicture}
   }   
   \caption{Cauchy slices $\Gamma^n$ of the regular triangular lattice supporting discrete time Lagrange functions for dRTL$_\pm(\alpha,h)$.}
  \label{Fig: dRTL slices}
\end{figure}
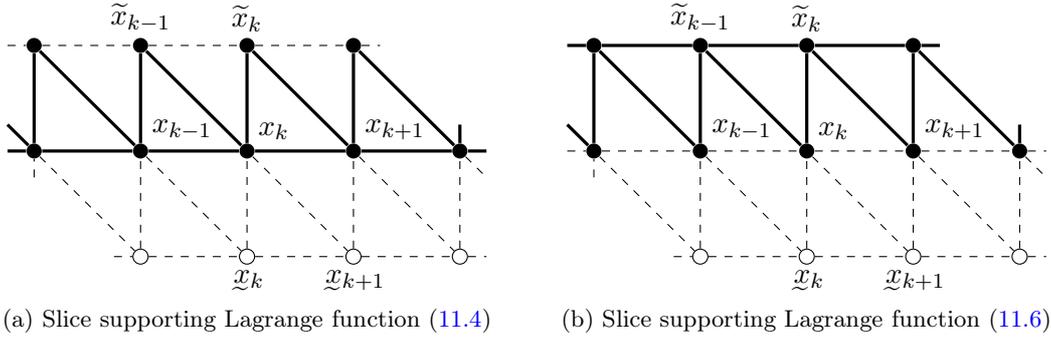
%%%%%%%%%%%%%%%%%%%%%%%%%%%%%%%%%%%%%%%%%%%%%%%%%%%%%%%%%%%%%%%%%

\paragraph{Bibliographical remarks.} Relation between lattice 2D systems and initial value problems for discrete 1D systems was discussed in the literature on many occasions, cf. \cite{PNC90, QCPN91}. Lagrangian aspects of this relation  were considered in \cite{CNP91}. The fact that different initial value problems for one and the same 2D lattice system lead to different 1D systems was instrumental in \cite{SR99} for a construction of a novel 1D integrable system, the relativistic Volterra lattice.

Some of the 2D Laplace type equations appearing as time discretizations of Toda lattices enjoy applications in different areas of mathematics. See, for instance, \cite{BH03} for an application of the discrete additive rational relativistic Toda  lattice to integrable circle patterns with the combinatorics of the regular hexagonal lattice and with prescribed intersection angles.

\section{Integrable discrete Laplace type equations and integrable quad-equations}
\label{Sect: dToda}

\paragraph{Discrete Laplace type equations.} 
We can slightly generalize definition of discrete Laplace type system, by removing the requirement that they come from a variational principle.

\begin{definition}\label{Dfn:dLaplace}
Let $\Gamma$ be a graph, with the set of vertices $V(\Gamma)$ and the
set of (directed) edges $E(\Gamma)$. A {\itbf discrete Laplace type system} on
$\Gamma$ for a function $x:V(\Gamma)\to\mathbb R$ consists of equations
\begin{equation}\label{eq:nlin Laplace}
\sum_{e_i=(v_0,v_i)\in \mathrm{star}(v)} \phi(x(v_0),x(v_i);e_i)=0.
\end{equation}
There is one equation for every vertex $v_0\in V(\Gamma)$; the
summation is extended over $\mathrm{star}(v_0)$, the set of edges
$e_i\in E(\Gamma)$ adjacent to $v_0$. The functions
$\phi=\phi(x_0,x_i;e_i)$ are possibly depending on $e_i$, often through
parameters $\alpha:E(\Gamma)\to\mathbb R$, assigned to the edges
of $\,\Gamma$.
\end{definition}

We will be mainly working with {\em planar} graphs $\Gamma$, that is, with graphs coming from a (strongly regular) polytopal cell decomposition
of an oriented surface. In this case, one can introduce the {\em dual graph} (cell
decomposition) $\Gamma^*$. Each edge $e\in E(\Gamma)$ separates two faces
of $\Gamma$, which in turn correspond to two vertices of $\Gamma^*$. A
path between these two vertices is then declared to be the edge
$e^*\in E(\Gamma^*)$ dual to $e$. If one assigns a direction to
an edge $e\in E(\Gamma)$, then it will be assumed that the dual
edge $e^*\in E(\Gamma^*)$ is also directed, in a way consistent
with the orientation of the underlying surface, namely so that the
pair $(e,e^*)$ is positively oriented at its crossing point.
This orientation convention implies that $e^{**}=-e$. Finally,
the faces of $\Gamma^*$ are in a one-to-one correspondence with the
vertices of $\Gamma$: if $v_0\in V(\Gamma)$, and $v_1,\ldots,v_n\in
V(\Gamma)$ are its neighbors connected with $v_0$ by the edges
$e_1=(v_0,v_1),\ldots, e_n=(v_0,v_n)\in E(\Gamma)$, then the face
of $\Gamma^*$ dual to $v_0$ is bounded by the dual edges
$e_1^*=(y_1,y_2),\ldots, e_n^*=(y_n,y_1)$; see Figure
\ref{Fig: graph and dual}\subref{Fig: dual face}.

%%%%%%%%%%%%%%%%%%%%%%%%%%%%%%%%%%%%%%%%%%%%%%%%%%%%%%%%%%%%%%%%%
\begin{figure}[tbp]
   \centering
   \subfloat[Faces of $\Gamma$ adjacent to $v_0$]{\label{Fig: star in a graph}
   \begin{tikzpicture}[scale=1.0,inner sep=1.5]  
      \node (x1) at (4,0) [circle,fill,label=-90:$v_5$] {};
      \node (x11) at (6,0) [circle,fill,label=-90:$v_6$] {};
      \node (x2) at (2,2) [circle,fill,label=135:$v_4$] {};
      \node (x12) at (4,2) [circle,fill,label=45:$v_0$] {};
      \node (x112) at (6,2) [circle,fill,label=45:$v_1$] {};
      \node (x22) at (2,4) [circle,fill,label=90:$v_3$] {};
      \node (x122) at (4,4) [circle,fill,label=90:$v_2$] {}; 
      \draw [very thick]  (x2) to (x1) to (x11);
      \draw [very thick] (x1) to (x12) to (x11) to (x112);
      \draw [very thick]  (x22) to (x122);
      \draw [very thick]  (x2) to (x12) to (x112);
      \draw [very thick]  (x2) to (x22)  to (x12) to (x122) to (x112);
   \end{tikzpicture}
   }\qquad
   \subfloat[Face of $\Gamma^*$ dual to $v_0$ and its oriented boundary]{\label{Fig: dual face}
  \begin{tikzpicture}[scale=1.0,>=stealth',inner sep=1.5]  
      \node (x1) at (4,0) [circle,fill,label=-90:$v_5$] {};
      \node (x11) at (6,0) [circle,fill,label=-90:$v_6$] {};
      \node (x2) at (2,2) [circle,fill,label=135:$v_4$] {};
      \node (x12) at (4,2) [circle,fill,label=45:$v_0$] {};
      \node (x112) at (6,2) [circle,fill,label=45:$v_1$] {};
      \node (x22) at (2,4) [circle,fill,label=90:$v_3$] {};
      \node (x122) at (4,4) [circle,fill,label=90:$v_2$] {}; 
      \node (y5) at (3.3,1.3) [circle,draw,label=-90:$y_5$]{};
      \node (y4) at (2.7,2.7) [circle,draw,label=-135:$y_4$]{};
       \node (y3) at (3.3,3.3) [circle,draw,label=90:$y_3$]{};
       \node (y2) at (4.7,2.7) [circle,draw,label=0:$y_2$]{};
        \node (y1) at (5.3,1.3) [circle,draw,label=45:$y_1$]{};
        \node (y6) at (4.7,0.7) [circle,draw,label=-90:$y_6$]{};
      \draw [very thick]  (x2) to (x1) to (x11) to (x112) to (x122) to (x22) to (x2);
      \draw [very thick,->] (x12) to (x1);
       \draw [very thick,->] (x12) to (x11);
      \draw [very thick,->]  (x12) to (x2);
      \draw [very thick,->] (x12) to (x112);
      \draw [very thick,->] (x12) to (x122);
      \draw [very thick,->] (x12) to (x22);
      \draw [very thick,dashed, ->] (y1) to (y2);
      \draw [very thick,dashed, ->] (y2) to (y3);
      \draw [very thick,dashed, ->] (y3) to (y4);
      \draw [very thick,dashed, ->] (y4) to (y5);
      \draw [very thick,dashed, ->] (y5) to (y6);
      \draw [very thick,dashed, ->] (y6) to (y1);
   \end{tikzpicture}
   }   
   \caption{}
   \label{Fig: graph and dual}
\end{figure}
%%%%%%%%%%%%%%%%%%%%%%%%%%%%%%%%%%%%%%%%%%%%%%%%%%%%%%%%%%%%%%%%%

\paragraph{Integrability of discrete Laplace type equations.} 
We will say that a discrete Laplace type system on $\Gamma$ is {\em integrable} if it possesses a
{\em discrete zero curvature representation}. That means the existence of a collection of
matrices $L(e^*;\lambda)\in G[\lambda]$ from some loop group $G[\lambda]$, associated to directed edges
$e^*\in\vec{E}(\Gamma^*)$ of the dual graph $\Gamma^*$, such that:
\begin{itemize}
 \item the matrix $L(e^*;\lambda)=L(x(v),x(w),\alpha;\lambda)$ depends on the fields
$x(v)$ and $x(w)$ at the vertices of the edge $e=(v,w)\in E(\Gamma)$,
dual to the edge $e^*\in E(\Gamma^*)$, as well as on the parameter
$\alpha=\alpha(e)$;
 \item for any directed edge
$e^*=(y_1,y_2)$, if $-e^*=(y_2,y_1)$, then
\begin{equation}\label{zero curv cond inv}
L(-e^*,\lambda)=\big(L(e^*,\lambda)\big)^{-1};
\end{equation}
\item for any closed path of directed edges $e^*_1=(y_1,y_2)$, $e^*_2=(y_2,y_3)$, $\ldots$, $e^*_n=(y_n,y_1)$,
we have
\begin{equation}\label{zero curv cond}
L(e^*_n,\lambda)\cdots L(e^*_2,\lambda)L(e^*_1,\lambda)=I.
\end{equation}
\end{itemize}
The matrix $L(e^*;\lambda)$ is interpreted as a transition
matrix  along the edge $e^*\in E(\Gamma^*)$, that is, a
transition across the edge $e\in E(\Gamma)$.
Under conditions \eqref{zero curv cond inv}, (\ref{zero curv
cond}) one can define a {\em wave function} $\Psi: V(\Gamma^*)\to
G[\lambda]$ on the vertices of the dual graph $\Gamma^*$, by the
following requirement: for any directed edge $e^*=(y_1,y_2)$,
the values of the wave functions at its ends must be connected via
\begin{equation}\label{wave function evol}
\Psi(y_2,\lambda)=L(e^*,\lambda)\Psi(y_1,\lambda).
\end{equation}

For an arbitrary graph, the analytical consequences of the zero
curvature representation for a given Laplace type system are
not clear. However, in the case of regular graphs, like the square
lattice or the regular triangular lattice, such a representation
may be used to determine conserved quantities for suitably defined
Cauchy problems, as well as to apply powerful analytical methods
for finding concrete solutions.

\paragraph{Quad-graphs and quad-equations.}
\label{Sect quad-graphs}

Although one can consider 2D integrable systems on very different
kinds of planar graphs, there is one kind, namely quad-graphs,
supporting the most fundamental integrable systems.
\begin{definition}\label{Def: quad-graph}
A {\itbf quad-graph} is a planar graph with all quadrilateral faces.
\end{definition}
Quad-graphs are privileged because from an arbitrary planar graph $\Gamma$ one can produce a
certain quad-graph $\cD$, called the double of $\Gamma$. The {\em
double} $\cD$ is a quad-graph, constructed from $\Gamma$ and its dual
$\Gamma^*$ as follows. The set of vertices
of the double $\cD$ is $V(\cD)=V(\Gamma)\sqcup V(\Gamma^*)$. Each pair
of dual edges, say $e=(v_0,v_1)\in E(\Gamma)$ and
$e^*=(y_1,y_2)\in E(\Gamma^*)$, defines a quadrilateral
$(v_0,y_1,v_1,y_2)$. These quadrilaterals constitute the faces of
a quad-graph $\cD$; see Figure \ref{Fig:flower}. Let us stress
that edges of $\cD$ belong neither to $E(\Gamma)$ nor to $E(\Gamma^*)$.

%------------------------------------------------------------------
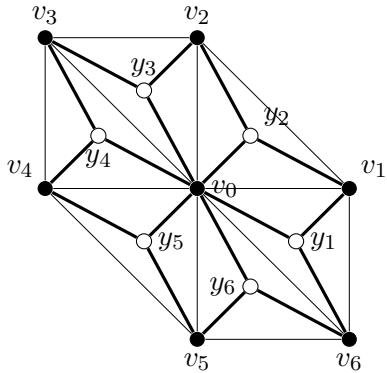
\begin{figure}[htbp]
\begin{center}
 \begin{tikzpicture}[scale=1.0,>=stealth',inner sep=2]  
      \node (x1) at (4,0) [circle,fill,label=-90:$v_5$] {};
      \node (x11) at (6,0) [circle,fill,label=-90:$v_6$] {};
      \node (x2) at (2,2) [circle,fill,label=135:$v_4$] {};
      \node (x12) at (4,2) [circle,fill,label=0:$v_0$] {};
      \node (x112) at (6,2) [circle,fill,label=45:$v_1$] {};
      \node (x22) at (2,4) [circle,fill,label=90:$v_3$] {};
      \node (x122) at (4,4) [circle,fill,label=90:$v_2$] {}; 
      \node (y5) at (3.3,1.3) [circle,draw,label=0:$y_5$]{};
      \node (y4) at (2.7,2.7) [circle,draw,label=-90:$y_4$]{};
       \node (y3) at (3.3,3.3) [circle,draw,label=90:$y_3$]{};
       \node (y2) at (4.7,2.7) [circle,draw,label=30:$y_2$]{};
        \node (y1) at (5.3,1.3) [circle,draw,label=0:$y_1$]{};
        \node (y6) at (4.7,0.7) [circle,draw,label=180:$y_6$]{};
      \draw [very thin]  (x2) to (x1) to (x11) to (x112) to (x122) to (x22) to (x2);
      \draw [very thin] (x12) to (x1);
       \draw [very thin] (x12) to (x11);
      \draw [very thin]  (x12) to (x2);
      \draw [very thin] (x12) to (x112);
      \draw [very thin] (x12) to (x122);
      \draw [very thin] (x12) to (x22);
      \draw [very thick] (x12) to (y5) to (x2) to (y4) to (x12);
      \draw [very thick] (y4) to (x22) to (y3) to (x12);
      \draw [very thick] (y3) to (x122) to (y2) to (x12);
      \draw [very thick] (y2) to (x112) to (y1) to (x12);
      \draw [very thick] (y1) to (x11) to (y6) to (x12);
      \draw [very thick] (y6) to (x1) to (y5);
   \end{tikzpicture}
   \caption{Faces of $\cD$ adjacent to the vertex $v_0$.}
   \label{Fig:flower}
\end{center}
    \end{figure}
%------------------------------------------------------------------------

Quad-graphs $\cD$ coming as doubles are bipartite: the set
$V(\cD)$ may be decomposed into two complementary halves,
$V(\cD)=V(\Gamma)\sqcup V(\Gamma^*)$ (``black'' and ``white'' vertices),
such that the ends of each edge from $E(\cD)$ are of different
colors. Equivalently, any closed loop consisting of edges of $\cD$
has an even length.

The construction of the double can be reversed. Start with a
bipartite quad-graph $\cD$. For instance, any quad-graph embedded
in a plane or in an open disc is automatically bipartite. Any
bipartite quad-graph produces two dual planar graphs $\Gamma$ and $\Gamma^*$, with
$V(\Gamma)$ containing all the ``black'' vertices of $\cD$ and
$V(\Gamma^*)$ containing all the ``white'' ones, and edges of $\Gamma$
(resp. of $\Gamma^*$) connecting ``black'' (resp. ``white'') vertices
along the diagonals of each face of $\cD$. The decomposition of
$V(\cD)$ into $V(\Gamma)$ and $V(\Gamma^*)$ is unique, up to
interchanging the roles of $\Gamma$ and $\Gamma^*$.

A privileged role played by the quad-graphs is reflected in the
privileged role played in the theory of discrete integrable
systems by the so called quad-equations supported by
quad-graphs.
\begin{definition} For a given bipartite quad-graph $\cD$, the system of {\itbf quad-equations} for a function $x:V(\cD)\to\mathbb R$ consists of equations of the type
\begin{equation}\label{eq:2d for 3leg}
Q(x_0,y_1,x_1,y_2)=0.
\end{equation}
There is one equation for every face $(x_0,y_1,x_1,y_2)$ of $\cD$, see Figure \ref{Fig: quad-equation}\subref{Fig: quadrilateral bipartite}. The function
$Q$ is supposed to be {\em multi-affine}, i.e., a polynomial of
degree $\le 1$ in each argument, so that equation (\ref{eq:2d for
3leg}) is uniquely solvable for any of its arguments. Usually, the function $Q=Q(x_0,y_1,x_1,y_2;\alpha,\beta)$
additionally depends on some parameters assigned to the
edges of the quadrilaterals, $\alpha:E(\cD)\to\mathbb C$, 
opposite edges carrying equal parameters:
$\alpha=\alpha(x_0,y_1)=\alpha(y_2,x_1)$ and
$\beta=\alpha(x_0,y_2)=\alpha(y_1,x_1)$.
\end{definition}

%%%%%%%%%%%%%%%%%%%%%%%%%%%%%%%%%%%%%%%%%%%%%%%%%%%%%%%%%%%%%%%%%
\begin{figure}[tbp]
   \centering
   \subfloat[A quad-equation.]{\label{Fig: quadrilateral bipartite}
   \begin{tikzpicture}[scale=0.8,inner sep=2]  
      \node (x0) at (0,2) [circle,fill,label=180:$x_0$] {};
      \node (y1) at (3,0) [circle,draw,label=-90:$y_1$] {};
      \node (x1) at (6,2) [circle,fill,label=0:$x_1$] {};
      \node (y2) at (3,4) [circle,draw,label=90:$y_2$] {};
      \draw [very thick]  (x0) -- (y1) node[pos=.5,sloped,above] {$\alpha$};
      \draw [very thick] (y1) -- (x1) node[pos=.5,sloped,above] {$\beta$};
       \draw [very thick] (x1) -- (y2) node[pos=.5,sloped,above] {$\alpha$}; 
       \draw [very thick] (y2) -- (x0) node[pos=.5,sloped,above] {$\beta$};
   \end{tikzpicture}
   }\qquad
   \subfloat[Three-leg form of a quad-equation.]{\label{Fig: 3leg}
  \begin{tikzpicture}[scale=0.8,inner sep=2]  
      \node (x0) at (0,2) [circle,fill,label=180:$x_0$] {};
      \node (y1) at (3,0) [circle,draw,label=-90:$y_1$] {};
      \node (x1) at (6,2) [circle,fill,label=0:$x_1$] {};
      \node (y2) at (3,4) [circle,draw,label=90:$y_2$] {};
      \draw [very thick]  (x0) to (y1) (x0) to (x1) (x0) to (y2);
   \end{tikzpicture}
   }   
   \caption{Different forms of a quad-equation.}
   \label{Fig: quad-equation}
\end{figure}
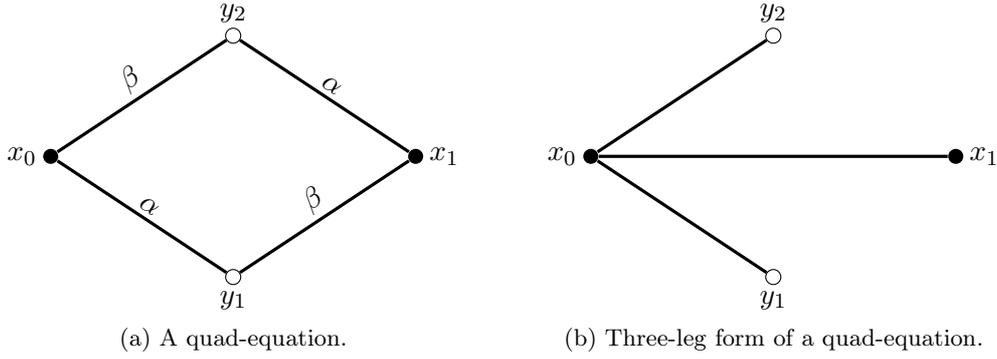
%%%%%%%%%%%%%%%%%%%%%%%%%%%%%%%%%%%%%%%%%%%%%%%%%%%%%%%%%%%%%%%%%

The geometric relation of a given planar graph $\Gamma$ to its
double $\cD$ leads to a relation of discrete Laplace type systems on $\Gamma$ to
systems of quad-equations on $\cD$. The latter relation is based on an intriguing
property of quad-equations to have the so called three-leg form. 

\begin{definition}\label{Def:3leg}
a) A quad-equation \eqref{eq:2d for 3leg} possesses a {\itbf three-leg
form} centered at the vertex $x_0$ if it is equivalent to the
equation
\begin{equation}\label{eq:3leg add}
\psi(x_0,y_1)-\psi(x_0,y_2)=\phi(x_0,x_1)
\end{equation}
with some functions $\psi, \phi$. The terms on the left-hand side
correspond to the ``short'' legs $(x_0,y_1)$, $(x_0,y_2)\in E(\cD)$,
while the right-hand side corresponds to the ``long'' leg
$(x_0,x_1)\in E(\Gamma)$, see Figure \ref{Fig: quad-equation}\subref{Fig: 3leg}.

b) A system of quad-equations \eqref{eq:2d for 3leg} on the faces of a
bipartite quad-graph $\cD$ has the {\itbf legs matching property}, if for 
any edge $(x,y)\in E(\cD)$, the short leg functions $\psi(x,y)$ for both quadrilaterals sharing this edge, coincide. 
\end{definition}

For a system of quad-equations with the leg matching property, consider the following discrete Laplace type equations on the ``black'' graph $\Gamma$, constructed from the ``long'' legs functions:
\begin{equation}\label{eq:Laplace for 3legs}
\sum_{x_k\in V(\Gamma),\,x_k\sim x_0} \phi(x_0,x_k)=0\qquad \forall x_0\in V(\Gamma).
\end{equation}

\begin{theorem}\label{Th: Laplace for 3legs}
a) The restriction of any solution $x:V(\cD)\to\mathbb R$ of the system of quad-equations \eqref{eq:2d for 3leg} to the
``black'' vertices $V(\Gamma)$ satisfies discrete Laplace type equations \eqref{eq:Laplace for 3legs}.

b) If the graph $\Gamma$ is embedded into a simply connected surface, then, conversely, given a solution $x:V(\Gamma)\to\mathbb R$ of
the Laplace type equations \eqref{eq:Laplace for 3legs}, there exists a one-parameter family
of extensions $x:V(\cD)\to\mathbb R$ satisfying quad-equations
\eqref{eq:2d for 3leg} on the double $\cD$. Such an extension is
uniquely determined by the value at one arbitrary vertex of
$V(\Gamma^*)$.
\end{theorem}
\begin{proof} This follows by summation of quad-equations over the vertex star of $\cD$
adjacent to the ``black'' vertex $x_0\in V(\Gamma)$ (see Figure
\ref{Fig:flower}), due to the telescoping
effect.
\end{proof}

\paragraph{Multi-dimensional consistency of quad-equations.} The by now widely accepted notion of integrability of quad-equations is that of multi-dimensional consistency.
Consider a system on $\mathbb Z^m$ consisting of (possibly different) quad-equations $Q(x,x_i,x_{ij},x_j)=0$ on all affine two-dimensional sublattices 
$n_0 + \mathbb Z e_i +\mathbb Z e_j$. Here $x$ stands for $x(n)$ at a generic point $n\in\bbZ^m$, and $x_i=x(n+e_i)$, $x_{ij}=x(n+e_i+e_j)$, 
where $e_i$ is the unit vector of the $i$\textsuperscript{th} coordinate direction. Such a system is called {\em multi-dimensionally consistent}, if it has a solution 
whose restrictions on all two-dimensional sublattices  are generic solutions of corresponding equations. It turns out that the multi-dimensional consistency of quad-equations follows from 3D consistency, and the latter boils down to a local property for one elementary 3D cube.
\begin{definition}
Consider a six-tuple of (a priori different) quad-equations assigned to the faces of a 3D cube:
\begin{eqnarray} \label{system}
A\left(x,x_{1},x_{2},x_{12}\right)=0,& \quad &
\bar{A}\left(x_{3},x_{13},x_{23},x_{123}\right)=0,\nonumber\\
B\left(x,x_{2},x_{3},x_{23}\right)=0,& \quad &
\bar{B}\left(x_{1},x_{12},x_{13},x_{123}\right)=0,\\
C\left(x,x_{1},x_{3},x_{13}\right)=0,& \quad &
\bar{C}\left(x_{2},x_{12},x_{23},x_{123}\right)=0.\nonumber
\end{eqnarray}
Such a six-tuple is called {\itbf 3D consistent} if, for arbitrary initial data $x$, $x_{1}$, $x_{2}$, $x_{3}$, and for $x_{12}$, $x_{13}$, $x_{23}$ determined by using $A=0$, $B=0$, $C=0$, the three values for $x_{123}$ determined by using $\bar{A}=0$, $\bar{B}=0$, or $\bar{C}=0$, coincide. See Figure~\ref{Fig: cube six}. 
\end{definition}

%%%%%%%%%%%%%%%%%%%%%%%%%%%%%%%%%%%%%%%%%%%%%%%%%%%%%%%%%%%%%%%%
\begin{figure}[htbp]
   \centering
   \begin{tikzpicture}[scale=1,inner sep=2]
   \tikzset{square/.style={regular polygon,regular polygon sides=4,inner sep=2}}
      \node (x) at (0,0) [circle,fill,label=-135:$x$] {};
      \node (x1) at (3,0) [circle,fill,label=-45:$x_1$] {};
      \node (x2) at (1,1) [circle,fill,label=-45:$x_2$] {};
      \node (x3) at (0,3) [circle,fill,label=135:$x_3$] {};
      \node (x12) at (4,1) [circle,draw,label=0:$x_{12}$] {};
      \node (x13) at (3,3) [circle,draw,label=0:$x_{13}$] {};
      \node (x23) at (1,4) [circle,draw,label=135:$x_{23}$] {};
      \node (x123) at (4,4) [square,draw,label=45:$x_{123}$] {};
      \draw [ultra thick] (x) to (x1) to (x12);
      \draw [very thick, dashed] (x12) to (x2) to (x);
      \draw [ultra thick] (x) to (x3) to (x13) to (x1);
      \draw [very thick, dashed] (x2) to (x23);
      \draw [ultra thick] (x3) to (x23) to (x123) to (x13);
      \draw [ultra thick] (x12) to (x123);
        \end{tikzpicture}
   \caption{3D consistency of quad-equations}
   \label{Fig: cube six}
\end{figure}
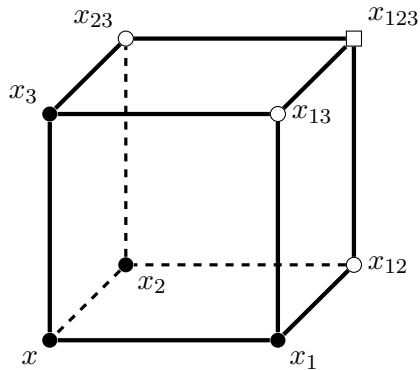
%%%%%%%%%%%%%%%%%%%%%%%%%%%%%%%%%%%%%%%%%%%%%%%%%%%%%%%%%%%%%%%%%%

This notion is relevant for systems of quad-equations on quad-graphs, since, under certain mild conditions, a quad-graph can be realized as a quad-surface in a lattice of a sufficiently high dimension. An example relevant to the discrete relativistic Toda type systems will be given in the next section.

The property of 3D consistency is relevant for integrability, since it implies the more traditional attributes thereof, like the discrete zero curvature representation and the existence of permutable B\"acklund transformations.

\paragraph{Bibliographical remarks.}

Integrability of discrete Laplace type equations on arbitrary planar graphs was discussed in \cite{A01}.

The notion of 3D consistency of quad-equations which can be put into the basis of
the integrability theory, was clearly formulated in \cite{NW01}. A conceptual breakthrough has been made in \cite{BS02}
and \cite{Nij02}, where it was shown that 3D consistency allows one to derive in an algorithmic way such basic
integrability attributes as discrete zero curvature representations and B\"acklund transformations for quad-equations.
In \cite{ABS03}, this property has been put  into the basis of a classification of integrable quad-equations which provided a
finite list of such equations known nowadays as the ``ABS list''.

A relation between integrable quad-equations and discrete Laplace type equations based on the three-leg forms in \cite{BS02}. The existence of the three-leg forms was established for all equations of the ABS list in \cite{ABS03}, and was proved for all quad-equations with multi-affine functions $Q$ by V.~Adler, see Exercise 6.16 in \cite{BS08}.

\section[Discrete relativistic Toda systems from quad-equations]
{Discrete relativistic Toda systems from quad-equations on the dual kagome lattice} 
\label{Sect main}

The non-symmetric discrete relativistic Toda type systems live on
the regular triangular lattice $\cT$ and cannot be directly
generalized to arbitrary graphs. Therefore, we introduce now the
specific notation tailored for the regular triangular lattice. The
double of $\cT$ is the quad-graph $\cK$ known as the {\em dual
kagome lattice} (drawn on Figure \ref{fig:L3} in dashed lines).
The latter graph has vertices of two kinds, black vertices of
valence 6 and white vertices of valence 3, and edges of three
different directions.

%------------------------------------------------------------------
\begin{figure}[htbp]
\begin{center}
 \begin{tikzpicture}[scale=1.2,>=stealth',inner sep=2.5]  
      \node (x1) at (6,0) [circle,fill,label=-90:$\undertilde{x}_k$] {};
      \node (x11) at (9,0) [circle,fill,label=-90:$\undertilde{x}_{k+1}$] {};
      \node (x2) at (3,3) [circle,fill,label=135:$x_{k-1}$] {};
      \node (x12) at (6,3) [circle,fill,label=45:$x_k$] {};
      \node (x112) at (9,3) [circle,fill,label=45:$x_{k+1}$] {};
      \node (x22) at (3,6) [circle,fill,label=90:$\wx_{k-1}$] {};
      \node (x122) at (6,6) [circle,fill,label=90:$\wx_k$] {}; 
      \node (y0) at (4,1) [circle,draw,label=180:$U_k$] {}; 
      \node (y1) at (5,2) [circle,draw,label=0:$V_k$]{};
      \node (y2) at (4,4) [circle,draw,label=180:$\wU_k$]{};
       \node (y3) at (5,5) [circle,draw,label=0:$\wV_k$]{};
       \node (y4) at (7,4) [circle,draw,label=180:$\wU_{k+1}$]{};
        \node (y5) at (8,2) [circle,draw,label=0:$V_{k+1}$]{};
        \node (y6) at (7,1) [circle,draw,label=180:$U_{k+1}$]{};
         \node (y7) at (8,5) [circle,draw,label=0:$\wV_{k+1}$]{};
      \draw [very thick]  (x2) to (x1) to (x11) to (x112) to (x122) to (x22) to (x2);
      \draw [very thick] (x12) to (x1);
       \draw [very thick] (x12) to (x11);
      \draw [very thick]  (x12) to (x2);
      \draw [very thick] (x12) to (x112);
      \draw [very thick] (x12) to (x122);
      \draw [very thick] (x12) to (x22);
      \draw [thin, dashed] (x12) to (y1) to (x2) to (y2) to (x12);
      \draw [thin, dashed] (y2) to (x22) to (y3) to (x12);
      \draw [thin, dashed] (y3) to (x122) to (y4) to (x12);
      \draw [thin, dashed] (y4) to (x112) to (y5) to (x12);
      \draw [thin, dashed] (y5) to (x11) to (y6) to (x12);
      \draw [thin, dashed] (y6) to (x1) to (y1);
       \draw [thin, dashed] (x2) to (y0) to (x1);  
       \draw [thin, dashed] (x122) to (y7) to (x112);      
   \end{tikzpicture}
   \caption{Fields on the triangular lattice and wave functions on its dual; when considered simultaneously, they live on the dual kagome lattice $\cK$.}
   \label{fig:L3}
\end{center}
    \end{figure}
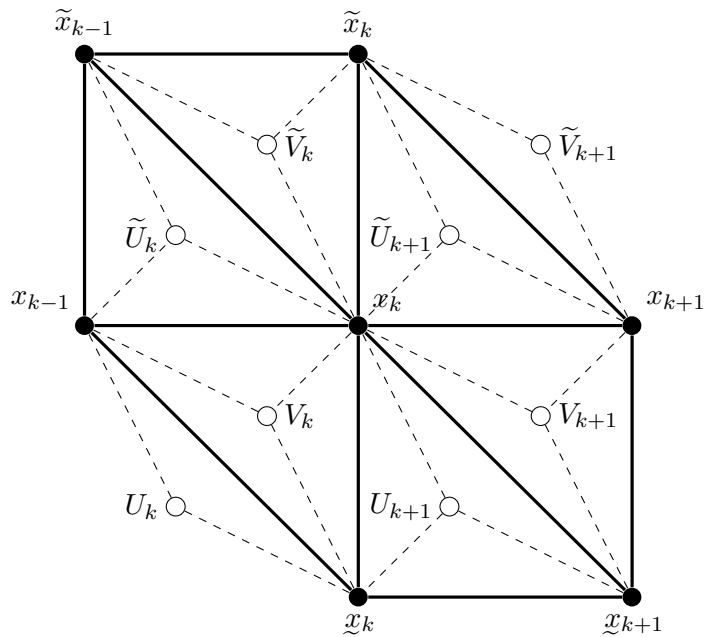
%------------------------------------------------------------------------

Correspondingly, it can be realized as a quad-surface in $\mathbb Z^3$, so that the three edge directions are 
realized as coordinate directions of $\mathbb Z^3$. In this realization, black vertices of $\cK$, that is, vertices of $\cT$, are the points
$(i_1,i_2,i_3)\in\mathbb Z^3$ lying in the plane $i_1+i_2+i_3=0$, while white vertices of $\cK$ are the points of $\mathbb Z^3$
lying in the planes $i_1+i_2+i_3=1$ (vertices $V$) and $i_1+i_2+i_3=-1$ (vertices $U$). See Figure \ref{fig:star6}.

%------------------------------------------------------------------
\begin{figure}[tbp]
\begin{center}
 \begin{tikzpicture}[scale=0.43,inner sep=2]  
      \node (x111) at (10,6) [circle,fill,label=15:$x_k$] {};
      \node (x100) at (6,-2) [circle,fill,label=-90:$\undertilde{x}_{k}$] {};
      \node (x210) at (16,-2) [circle,fill,label=-45:$\undertilde{x}_{k+1}$] {};
      \node (x001) at (0,6) [circle,fill,label=180:$x_{k-1}$] {};
      \node (x221) at (20,6) [circle,fill,label=0:$x_{k+1}$] {};
      \node (x012) at (4,14) [circle,fill,label=135:$\wx_{k-1}$] {};
      \node (x122) at (14,14) [circle,fill,label=45:$\wx_k$] {}; 
      \node (x101) at (6,4) [circle,draw,label=-135:$V_k$]{};
      \node (x211) at (16,4) [circle,draw,label=-45:$V_{k+1}$]{};
      \node (x011) at (4,8) [circle,draw,label=135:$\wU_k$]{};
       \node (x121) at (14,8) [circle,draw,label=10:$\wU_{k+1}$]{};
       \node (x112) at (10,12) [circle,draw,label=90:$\wV_k$]{};
        \node (x110) at (10,0) [circle,draw,label=-90:$U_{k+1}$]{};  
        \node (x000) at (0,0) [circle,draw,label=-90:$U_{k}$]{}; 
        \node (x222) at (20,12) [circle,draw,label=0:$\wV_{k+1}$]{};
 % horizontal squares 2x2           
        \draw [thick,dotted]  (x000) to (x100) to (12,-4) to (x210) to (20,0) to (14,2) to (8,4) to (4,2) to (x000);
        \draw [thick,dotted]  (x001) to (x101) to (12,2) to (x211) to (x221) to (x121) to (8,10) to (x011) to (x001);
        \draw [thick,dotted]  (0,12) to (6,10) to (12,8) to (16,10) to (x222) to (x122) to (8,16) to (x012) to (0,12);
 % midlines of horizontal squares       
        \draw [thick,dotted] (x100) to (x110) to (14,2);
        \draw [thick,dotted] (4,2) to (x110) to (x210);
        \draw [thick,dotted] (x011) to (x111) to (x211);
        \draw [thick,dotted] (x101) to (x111) to (x121);
         \draw [thick,dotted] (6,10) to (x112) to (x122);
        \draw [thick,dotted] (x012) to (x112) to (16,10);
% vertical lines        
        \draw [thick,dotted] (x000) to (x001) to (0,12);
        \draw [thick,dotted] (x100) to (x101) to (6,10);
        \draw [thick,dotted] (12,-4) to (12,2) to (12,8);
        \draw [thick,dotted] (x210) to (x211) to (16,10);
        \draw [thick,dotted] (20,0) to (x221) to (x222);
        \draw [thick,dotted] (14,2) to (x121) to (x122);
        \draw [thick,dotted] (8,4) to (8,10) to (8,16);
        \draw [thick,dotted] (4,2) to (x011) to (x012);
        \draw [thick,dotted] (x110) to (x111) to (x112);
% edges of RTL
      \draw [very thick] (x100) to (x111) to (x122);
      \draw [very thick] (x001) to (x111) to (x221);
      \draw [very thick] (x012) to (x111) to (x210);
      \draw [very thick] (x100) to (x210) to (x221) to (x122) to (x012) to (x001) to (x100);
   \end{tikzpicture}
   \caption{Embedding of the triangular lattice and the dual kagome lattice into $\mathbb Z^3$}.
   \label{fig:star6}
\end{center}
    \end{figure}
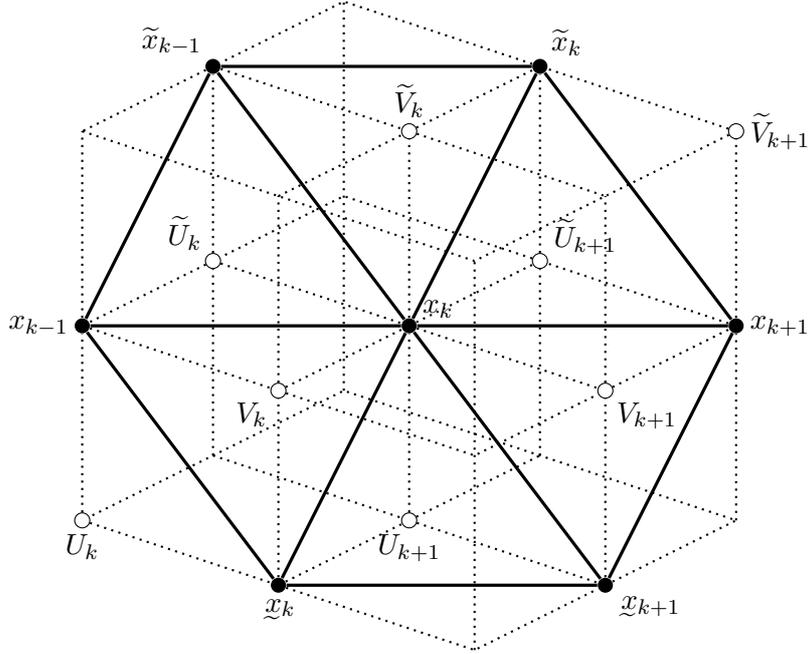
%------------------------------------------------------------------------

\begin{theorem}\label{Th main}
Each discrete relativistic Toda type system is a restriction to
the triangular lattice $\cT$ of a certain 3D consistent system of
quad-equations on the dual kagome lattice $\cK$ considered as a
quad-surface in $\mathbb Z^3$.
\end{theorem}
\begin{proof}  The corresponding systems of quad-equations are constructed case by case. The quadrilateral
faces of the dual kagome lattice are of three different types (corresponding to three different directions of coordinate planes of $\mathbb Z^3$). We will denote them by type I, II, and III. The systems are specified by giving the quad-equations explicitly for each type of quadrilateral faces separately in notation of Figure \ref{Fig: quads}. One has to: a) check the 3D consistency of the quad-equations, and b) find the three-leg forms, centered at $x_k$, of quad-equations for all six quadrilaterals around $x_k$, and then check that adding these three-leg forms results in the corresponding discrete Toda equation. All this is a matter of direct computations. 

%%%%%%%%%%%%%%%%%%%%%%%%%%%%%%%%%%%
\begin{figure}[htbp]
   \centering
   \subfloat[Quadrilateral (NW) of type III]{\label{Fig: north-west quad}
   \begin{tikzpicture}[scale=0.9,inner sep=1.5]  
    \node (Y) at (0,3) [circle,fill,label=90:$Y{=}\eto{\wx_{k-1}}$] {};
    \node (U) at (2,2) [circle,draw,label=0:$U{=}\wV_k$] {};
    \node (X) at (3,0) [circle,fill,label=-90:$X{=}\eto{x_k}$] {};
    \node (V) at (1,1) [circle,draw,label=180:$V{=}\wU_k$] {};
    \draw [very thick] (X) to (Y);
    \draw [thin] (X) to (U) to (Y) to (V) to (X);
   \end{tikzpicture}
   }\;
   \subfloat[Quadrilateral (N) of type I]{\label{Fig: north quad}
  \begin{tikzpicture}[scale=0.9,inner sep=1.5]  
   \node (Y) at (1,3) [circle,fill,label=90:$Y{=}\eto{\wx_k}$] {};
    \node (U) at (2,1) [circle,draw,label=0:$U{=}\wU_{k+1}$] {};
    \node (X) at (1,0) [circle,fill,label=-90:$X{=}\eto{x_k}$] {};
    \node (V) at (0,2) [circle,draw,label=180:$V{=}\wV_k$] {};
    \draw [very thick] (X) to (Y);
    \draw [thin] (X) to (U) to (Y) to (V) to (X);
        \end{tikzpicture}
    }\;
   \subfloat[Quadrilateral (E) of type II]{\label{Fig: east quad}
  \begin{tikzpicture}[scale=0.9,inner sep=1.5]  
   \node (Y) at (3,1) [circle,fill,label=0:$Y{=}\eto{x_{k+1}}$] {};
    \node (V) at (1,2) [circle,draw,label=90:$V{=}\wU_{k+1}$] {};
    \node (X) at (0,1) [circle,fill,label=180:$X{=}\eto{x_k}$] {};
    \node (U) at (2,0) [circle,draw,label=-90:$U{=}V_{k+1}$] {};
    \draw [very thick] (X) to (Y);
    \draw [thin] (X) to (U) to (Y) to (V) to (X);
        \end{tikzpicture}     
   }  \\
   \subfloat[Quadrilateral (W) of type II]{\label{Fig: west quad}
  \begin{tikzpicture}[scale=0.9,inner sep=1.5]  
   \node (Y) at (3,1) [circle,fill,label=0:$Y{=}\eto{x_k}$] {};
    \node (V) at (1,2) [circle,draw,label=90:$V{=}\wU_k$] {};
    \node (X) at (0,1) [circle,fill,label=180:$X{=}\eto{x_{k-1}}$] {};
    \node (U) at (2,0) [circle,draw,label=-90:$U{=}V_k$] {};
    \draw [very thick] (X) to (Y);
    \draw [thin] (X) to (U) to (Y) to (V) to (X);
        \end{tikzpicture}  
     }\;
   \subfloat[Quadrilateral (S) of type I]{\label{Fig: south quad}
  \begin{tikzpicture}[scale=0.9,inner sep=1.5]  
   \node (Y) at (1,3) [circle,fill,label=90:$Y{=}\eto{x_k}$] {};
    \node (U) at (2,1) [circle,draw,label=0:$U{=}U_{k+1}$] {};
    \node (X) at (1,0) [circle,fill,label=-90:$X{=}\eto{\undertilde{x}_k}$] {};
    \node (V) at (0,2) [circle,draw,label=180:$V{=}V_k$] {};
    \draw [very thick] (X) to (Y);
    \draw [thin] (X) to (U) to (Y) to (V) to (X);
        \end{tikzpicture}
    }  \;
    \subfloat[Quadrilateral (SE) of type III]{\label{Fig: south-east quad}
   \begin{tikzpicture}[scale=0.9,inner sep=1.5]  
    \node (Y) at (0,3) [circle,fill,label=90:$Y{=}\eto{x_k}$] {};
    \node (U) at (2,2) [circle,draw,label=0:$U{=}V_{k+1}$] {};
    \node (X) at (3,0) [circle,fill,label=-90:$X{=}\eto{\undertilde{x}_{k+1}}$] {};
    \node (V) at (1,1) [circle,draw,label=180:$V{=}U_{k+1}$] {};
    \draw [very thick] (X) to (Y);
    \draw [thin] (X) to (U) to (Y) to (V) to (X);
   \end{tikzpicture}
   }
    \caption{Notation for single quadrilaterals of the dual kagome lattice around the vertex $x_k$}
     \label{Fig: quads}
\end{figure}
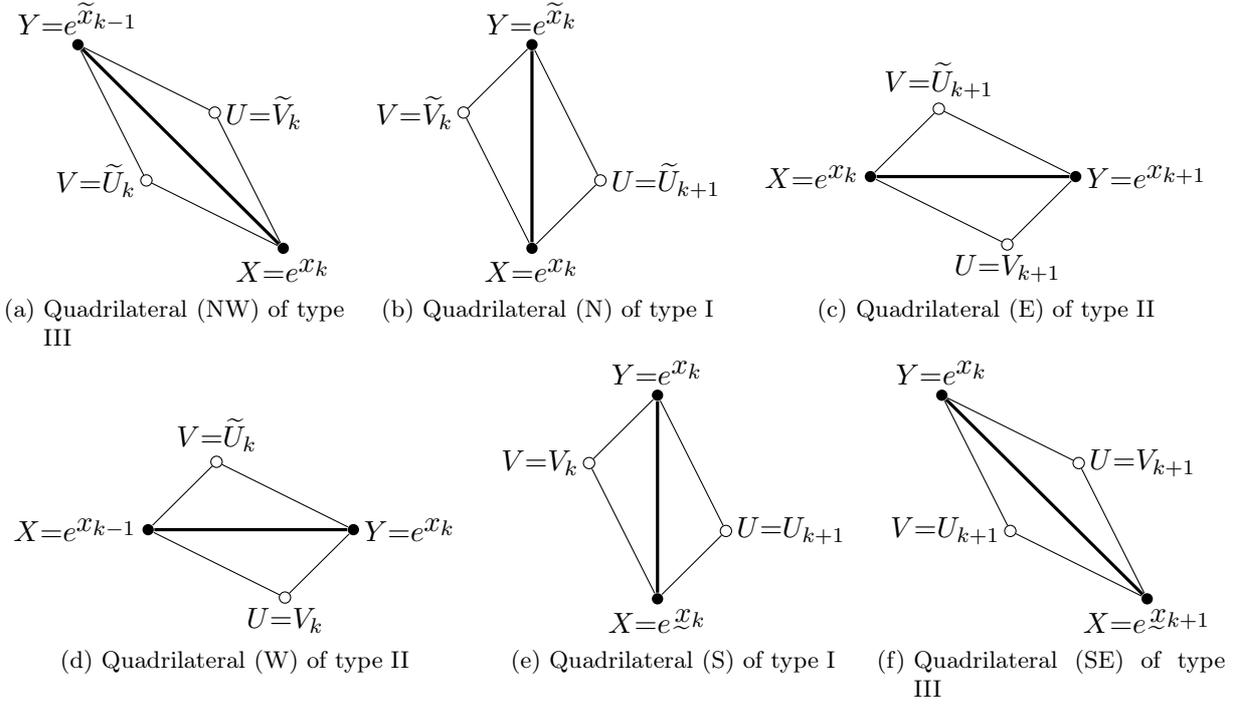
%%%%%%%%%%%%%%%%%%%%%%%%%%%%%%%%%%%%%%%%%%%%%%%%%%%%%%%%%%%%%%%%%

We only give details for the additive exponential relativistic Toda lattice (\ref{dRTL+ l New}).
For this system, we have the following 3D consistent system of quad-equations:
\begin{align}
& h(XY+UV)+YV+h^2XU-(1-h\lambda)XV=0, \tag{I}\\
& \alpha (XY+UV)+XV+\alpha^2YU-(1-\alpha\lambda)XU=0, \tag{II}\\
& (h-\alpha)(XY+UV)+(1-\alpha\lambda)YU-(1-h\lambda)YV  \nonumber\\
& \qquad\qquad +h^2(1-\alpha\lambda)XV-\alpha^2(1-h\lambda)XU=0.
\tag{III}
\end{align}
Three-leg forms of equation (I) centered at $X$ and at $Y$ read:
$$
\frac{Y}{X}+\frac{hU}{X}-\frac{(1-h\lambda)V}{V+hX}=0,\quad \mathrm{resp.}\quad 
\frac{Y}{X}+\frac{hY}{V}-\frac{(1-h\lambda)Y}{Y+hU}=0.
$$
In notation of Figure  \ref{Fig: quads}\subref{Fig: north quad} and \subref{Fig: south quad}, we obtain:
\begin{align}
& \eto{\wx_{k}\nm x_{k}}+
\frac{h\widetilde{U}_{k+1}}{\eto{x_{k}}}
-\frac{(1-h\lambda)\widetilde{V}_k}{\widetilde{V}_{k}+h\eto{x_{k}}}
=0,  \tag{N}\\
& \eto{x_{k}\nm\undertilde{x}_{k}}+\frac{h\eto{x_{k}}}{V_{k}}-
\frac{(1-h\lambda)\eto{x_{k}}}{\eto{x_{k}}+hU_{k+1}}
=0.  \tag{S}
\end{align}
For equation (II) we find the following three-leg forms (in notation of Figure \ref{Fig: quads}\subref{Fig: east quad} and \subref{Fig: west quad}):
\begin{align}
& \alpha\eto{x_{k+1}\nm x_{k}}+\frac{\widetilde{U}_{k+1}}{\eto{x_{k}}}-
\frac{(1-\alpha\lambda)V_{k+1}}{\eto{x_{k}}+\alpha V_{k+1}}
=0,  \tag{E}\\
& \alpha \eto{x_{k}-x_{k-1}}+\frac{\eto{x_{k}}}{V_{k}}
-\frac{(1-\alpha\lambda)\eto{x_{k}}}{\widetilde{U}_{k}+\alpha\eto{x_{k}}}
=0.  \tag{W}
\end{align}
For equation (III) we find the following three-leg forms (in notation of Figure \ref{Fig: quads}\subref{Fig: north-west quad} and \subref{Fig: south-east quad}):
\begin{align}
& \frac{(\alpha-h)\eto{x_{k}\nm \wx_{k-1}}}{1-h\alpha\eto{x_{k}\nm \wx_{k-1}}}
-\frac{(1-\alpha\lambda)\eto{x_{k}}}{\widetilde{U}_{k}+\alpha\eto{x_{k}}}
+\frac{(1-h\lambda)\eto{x_{k}}}{\widetilde{V}_{k}+h\eto{x_{k}}}
=0,   \tag{NW}\\
& \frac{(\alpha-h)\eto{\undertilde{x}_{k+1}\nm x_k}}
{1-h\alpha\eto{\undertilde{x}_{k+1}\nm x_{k}}}-
\frac{(1-\alpha\lambda) V_{k+1}}{\eto{x_{k}}+\alpha V_{k+1}}+
\frac{(1-h\lambda) U_{k+1}}{\eto{x_{k}}+hU_{k+1}}
=0.  \tag{SE}
\end{align}
Adding/subtracting these six equations, we see that all contributions of the short legs (depending on $\lambda$) cancel away, and we are left with the combination of the long legs expressed as (\ref{dRTL+ l New}).
\end{proof}

\paragraph{Zero curvature representations.}
\label{Sect zcr}

The construction of discrete Laplace type systems on graphs from
systems of quad-equations allows one to find, in an algorithmic
way, discrete zero curvature representations for the former.
Indeed, each quad-equation can be viewed as a M\"obius
transformation of the field at one white vertex of a quad into the
field at the other white vertex, with the coefficients dependent
on the fields at the both black vertices. The $PSL(2,\mathbb R)$
matrices representing these M\"obius transformations play then the
role of transition matrices across the edges connecting the black
vertices. The property (\ref{zero curv cond}) is satisfied
automatically, by construction.

Specializing this construction to the case of the regular
triangular lattice, we denote by $L_k$
the transition matrix from $V_k$ to $V_{k+1}$, and by $M_k$
the transition matrix from $V_{k}$ to $\widetilde{V}_{k}$  (see Figure \ref{fig:L3} for notations). Then the discrete zero curvature representation reads:
\begin{equation}\label{eq: dzcr}
  \widetilde{L}_kM_{k}=M_{k+1}L_k,
\end{equation}
both parts representing the transition from $V_k$ to
$\widetilde{V}_{k+1}$ along two different paths. Discrete zero curvature representation depending on a spectral parameter $\lambda$ is one of the central integrability attributes. In particular, it implies that the {\em monodromy matrix}
$$
T_N(x,p,\lambda)=L_N(x,p,\lambda)\cdots L_2(x,p,\lambda)L_1(x,p,\lambda)
$$
remains isospectral under the discrete time evolution (at least in the case of periodic boundary conditions):
$$
\widetilde{T}_N M_1=M_1T_N,
$$
so that its spectral invariants are integrals of motion of the system.

It is clear (see Figure \ref{fig:L3})
that $L_k$ is the product of two matrices, the first corresponding
to the transition from $V_k$ to $U_{k+1}$ across the edge
$[x_k,\undertilde{x}_k]$ (equation (S) on Figure \ref{Fig: quads}\subref{Fig: south quad}), and the second corresponding to the
transition from $U_{k+1}$ to $V_{k+1}$ across the edge
$[x_k,\undertilde{x}_{k+1}]$ (equation (SE) on Figure \ref{Fig: quads}\subref{Fig: south-east quad}). Similarly, $M_k$ can be represented as the product of two
matrices, the first corresponding to the transition from $V_{k}$
to $\widetilde{U}_{k}$ across the edge $[x_k,x_{k-1}]$ (equation (W) on Figure \ref{Fig: quads}\subref{Fig: west quad}), and the
second corresponding to the transition from $\widetilde{U}_{k}$
to $\widetilde{V}_{k}$ across the edge $[x_k,\wx_{k-1}]$ (equation (NW) on Figure \ref{Fig: quads}\subref{Fig: north-west quad}).
The matrices $L_k$, $M_k$ for a given discrete relativistic Toda
type equation can be computed in a straightforward way, as soon as
the generating system of quad-equations mentioned in Theorem \ref{Th
main} is known.

\begin{theorem}
For all discrete relativistic Toda type systems, the transition
matrix $L_k$ is local when expressed in terms of canonically conjugate variables,
$L_k=L(x_k,p_k;\lambda)$, and does not depend on
the time discretization parameter $h$. 
\end{theorem}
\begin{proof} This is obtained via direct
computations on the case-by-case basis. We illustrate the claims of the theorem with the
case (\ref{dRTL+ l New}). For this case, we first eliminate $U_{k+1}$ between equations (S) and (SE). Upon taking into account the second equation of motion in \eqref{dRTL+ l}, downshifted in discrete time, we find:
$$
p_k+\lambda+\frac{\eto{x_k}}{V_k}+\frac{(1-\alpha\lambda)V_{k+1}}{\eto{x_k}+\alpha V_{k+1}}=0.
$$
This can be put as $V_{k+1}=L_k\big[V_k\big]$ with
$$
L_k=L(x_k,p_k;\lambda)=\begin{pmatrix} p_k+\lambda & \eto{x_k} \vspace{2mm}\\ -(1+\alpha p_k)\eto{-x_k} & -\alpha\end{pmatrix}.
$$
Analogously, we eliminate $\wU_k$ between equations (W) and (NW). Taking into account the second equation of motion in \eqref{dRTL+ l} with $k\to k-1$, we find:
$$
h\eto{- \wx_{k-1}}(1+\alpha\wip_{k-1})+\frac{1}{V_k}-\frac{1-h\lambda}{\wV_k+h\eto{x_k}}=0.
$$
This can be put as $\wV_k=M_k\big[V_k\big]$ with
$$
M_k=M(x_k,\wx_{k-1},\wip_{k-1};\lambda)=\begin{pmatrix} 
1-h\lambda-h^2(1+\alpha\wip_{k-1})\eto{x_k\nm \wx_{k-1}} & -h\eto{x_k} \vspace{2mm}\\ h(1+\alpha\wip_{k-1})\eto{-\wx_{k-1}} & 1\end{pmatrix}.
$$
The fact that $L_k$ does not depend on $h$ means that the corresponding
symplectic maps $(x,p)\mapsto(\wx,\widetilde{p})$ belong to the
same integrable hierarchies as their respective continuous time
Hamiltonian counterparts. This confirms once again that these maps
serve as B\"acklund transformations for the respective Hamiltonian
flows, the B\"acklund parameter being the time step $h$.
\end{proof}

\paragraph{Bibliographical remarks.} Our presentation here follows \cite{BolS10}, where one can also find details for other discrete relativistic Toda type systems. For the elliptic Toda systems mentioned in the bibliographical remarks to Sections \ref{sect Toda New}, \ref{Sect Newtonian rel Toda}, see \cite{ASu04}. These discrete Laplace type systems are obtained by the same procedure as discussed in this section from the master integrable equation Q4 of the ABS list (and its hyperbolic and rational degenerations Q3$_{\delta=1}$ and Q2). For these equations, all leg functions (the ``short'' and the ``long'' ones) are given by \eqref{phi elliptic}.

%%%%%%%%%%%%%%%%%%%%%%%%%%%%%%%%%%%%%%%%%%%%%%%%%%%%%%%%%%%%%%%%%%%

\section{General theory of discrete one-dimensional pluri-Lagrangian systems}
\label{sect: discr 1d results}

Recall that symplectic maps describing the discrete time Toda lattice with different step sizes commute. We now turn to the theory which gives a deep insight into the nature of commuting symplectic maps.
\begin{definition} {\itbf (One-dimensional pluri-Lagrangian problem)} \label{def:pluriLagr problem 1d}
Let $\cL$ be a discrete $1$-form on $\bbZ^{m}$ (a function of directed edges $\sigma$ of $\bbZ^{m}$ with $\cL(-\sigma)=-\cL(\sigma)$), depending on a function $x:\bbZ^{m}\to\mathcal{X}$, where $\mathcal{X}$ is some vector space. It is supposed that $\cL(\sigma)$ depends on the values of the field $x$ at the endpoints of the edge $\sigma$.
\begin{itemize}
\item To an arbitrary discrete curve $\Sigma$ in $\bbZ^{m}$ (a concatenation of a sequence of directed edges in $\mathbb Z^m$ such that the endpoint of any edge is the beginning of the next one), there corresponds the {\itbf action functional}
\[
S_{\Sigma}=\sum_{\sigma\in\Sigma}\cL(\sigma)
\]
(it depends only on the fields at vertices of $\Sigma$).
\item We say that the field $x:V(\Sigma)\to\mathcal{X}$ is a critical point of $S_{\Sigma}$, if at any interior point $n$ of the curve $\Sigma$, we have
\[\frac{\partial S_{\Sigma}}{\partial x(n)}=0.
\]
\item We say that the field $x:\bbZ^{m}\to\mathcal{X}$ solves the {\itbf pluri-Lagrangian problem} for the Lagrangian $1$-form $\cL$ if, for any discrete curve $\Sigma$ in $\bbZ^{m}$, the restriction $\left.x\right|_{V(\Sigma)}$ is a critical point of the corresponding action $S_{\Sigma}$.
\end{itemize}
\end{definition}

We use the following notations: $x$ for $x(n)$ at a generic point $n\in\bbZ^m$, and
\begin{alignat}{3}\label{eq: shifts}
    &x_i=x(n+e_i),&\quad& x_{-i}=x(n-e_i),&\quad& i=1,\ldots,m,
\end{alignat}
where $e_i$ is the unit vector of the $i$\textsuperscript{th} coordinate direction, then we assume that
\[
\cL(\sigma_i)=\cL(n,n+e_i)=\Lambda_i(x,x_i)\quad\Leftrightarrow\quad \cL(-\sigma_i)=\cL(n+e_{i},n)=-\Lambda_i(x,x_{i}).
\]
Here $\Lambda_i:{\mathcal X}\times {\mathcal X}\to\mathbb R$ are local Lagrangians corresponding to the edges of the $i$\textsuperscript{th} coordinate direction.

Any interior point of any discrete curve $\Sigma$ in $\bbZ^m$ is of one of the four types shown on Figure \ref{Fig: corners}.\par
%--------------------------------------------------------------------------
\begin{figure}[htbp]
\centering
\subfloat[]{\label{Fig: corners1}
\begin{tikzpicture}[auto,scale=0.6,>=stealth',inner sep=2]
   \node (x-1) at (-2,0) [circle,fill,label=-135:$x_{-i}$] {};
   \node (x) at (0,0) [circle,fill,label=-90:$x$] {};
   \node (x1) at (2,0) [circle,fill,label=-45:$x_i$] {};
   \draw[->] (x) to (x1);
   \draw[->] (x-1) to (x);
\end{tikzpicture}
}\qquad
\subfloat[]{\label{Fig: corners2}
\begin{tikzpicture}[auto,scale=0.6,>=stealth',inner sep=2pt]
   \node (x) at (0,0) [circle,fill,label=-135:$x$] {};
   \node (x1) at (2,0) [circle,fill,label=-45:$x_i$] {};
   \node (x2) at (0,2) [circle,fill,label=135:$x_j$] {};
   \draw[->] (x) to (x1);
   \draw[->] (x2) to (x);
\end{tikzpicture}
}\qquad
\subfloat[]{\label{Fig: corners3}
\begin{tikzpicture}[auto,scale=0.6,>=stealth',inner sep=2pt]
   \node (x) at (0,0) [circle,fill,label=-135:$x_{-i}$] {};
   \node (x1) at (2,0) [circle,fill,label=-45:$x$] {};
   \node (x12) at (2,2) [circle,fill,label=45:$x_{j}$] {};
   \draw[->] (x) to (x1) ;
   \draw[->] (x1) to (x12);
\end{tikzpicture}
}\qquad
\subfloat[]{\label{Fig: corners4}
\begin{tikzpicture}[auto,scale=0.6,>=stealth',inner sep=2pt]
   \node (x1) at (2,0) [circle,fill,label=-45:$x_{-j}$] {};
   \node (x2) at (0,2) [circle,fill,label=135:$x_{-i}$] {};
   \node (x12) at (2,2) [circle,fill,label=45:$x$] {};
   \draw[<-] (x1) to (x12);
   \draw[->] (x2) to (x12);
\end{tikzpicture}
}
\caption{Four types of vertices of a discrete curve. Case \protect\subref{Fig: corners1}: two edges of one coordinate direction meet at $n$. Case \protect\subref{Fig: corners2}: a negatively directed edge followed by a positively directed edge. Case \protect\subref{Fig: corners3}: two equally (positively or negatively) directed edges of two different coordinate directions meet at $n$. Case \protect\subref{Fig: corners4}: a positively directed edge followed by a negatively directed edge. Values of the field $x$ at the points of the curve are indicated.}
\label{Fig: corners}
\end{figure}
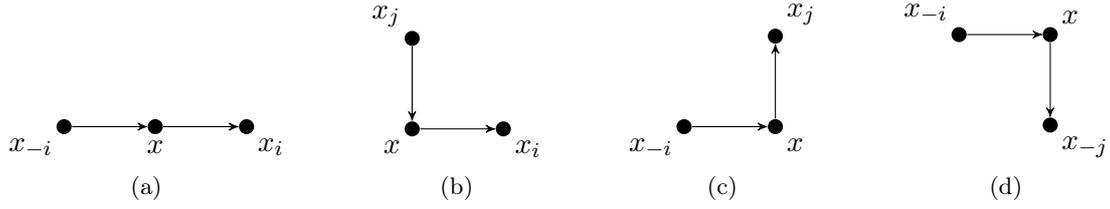
%%%%%%%%%%%%%%%%%%%%%%%%%%%%%%%%%%%%%%%%%%%%%%%%%%%%%%%%%%%%%%%%%%%%%%%%%%%%%%%%%%%%%%
The pieces of discrete curves as on Figures~\ref{Fig: corners}\subref{Fig: corners2}, \subref{Fig: corners3}, and \subref{Fig: corners4} will be called \emph{2D corners}. Observe that a straight piece of a discrete curve, as on Figure \ref{Fig: corners}\subref{Fig: corners1}, is a sum of 2D corners, as on Figures~\ref{Fig: corners}\subref{Fig: corners2} and \subref{Fig: corners3}. The whole variety of Euler-Lagrange equations for a pluri-Lagrangian system with $d=1$ reduces to the following three types of \emph{2D corner equations}:
\begin{align}
\label{eq: E 0}
&\frac{\partial\Lambda_i(x,x_i)}{\partial x}-
\frac{\partial\Lambda_j(x,x_j)}{\partial x}=0,
\\
\label{eq: E1 0}
&\frac{\partial\Lambda_i(x_{-i},x)}{\partial x}+
\frac{\partial\Lambda_j(x,x_j)}{\partial x}=0,
\\
\label{eq: E12 0}
&\frac{\partial\Lambda_i(x_{-i},x)}{\partial x}-
\frac{\partial\Lambda_j(x_{-j},x)}{\partial x} = 0.
\end{align}
In particular, the standard single-time discrete Euler-Lagrange equation,
\begin{equation*}\label{eq: dEL}
\frac{\partial\Lambda_i(x_{-i},x)}{\partial x}+\frac{\partial\Lambda_i(x,x_i)}{\partial x}=0,
\end{equation*}
corresponding to a straight piece of a discrete curve as on Figure~\ref{Fig: corners}\subref{Fig: corners1} is a consequence of equations \eqref{eq: E 0} and \eqref{eq: E1 0}, corresponding to 2D corners as on on Figures~\ref{Fig: corners}\subref{Fig: corners2} and \subref{Fig: corners3}.

To discuss \emph{consistency} of the system of 2D corner equations, it will be more convenient to re-write them with appropriate shifts, as
\begin{align}
\label{eq: E}\tag{$E$}
&\frac{\partial\Lambda_i(x,x_i)}{\partial x}-
\frac{\partial\Lambda_j(x,x_j)}{\partial x}=0,
\\
\label{eq: Ei}\tag{$E_i$}
&\frac{\partial\Lambda_i(x,x_i)}{\partial x_i}+
\frac{\partial\Lambda_j(x_i,x_{ij})}{\partial x_i}=0,
\\
\label{eq: Ej}\tag{$E_j$}
&\frac{\partial\Lambda_j(x,x_j)}{\partial x_j}+
\frac{\partial\Lambda_i(x_j,x_{ij})}{\partial x_j}=0,
\\
\label{eq: Eij}\tag{$E_{ij}$}
&\frac{\partial\Lambda_i(x_j,x_{ij})}{\partial x_{ij}}-
\frac{\partial\Lambda_j(x_i,x_{ij})}{\partial x_{ij}} = 0.
\end{align}
In this form, 2D corner equations \eqref{eq: E}--\eqref{eq: Eij} correspond to the four vertices of an elementary square $\sigma_{ij}$ of the lattice, as on Figure~\ref{Fig: consistency}\subref{Fig: consistency1}.
Consistency of the system of 2D corner equations \eqref{eq: E}--\eqref{eq: Eij} should be understood as follows: start with the fields $x$, $x_i$, $x_j$ satisfying equation \eqref{eq: E}. Then each of equations \eqref{eq: Ei}, \eqref{eq: Ej} can be solved for $x_{ij}$. Thus, we obtain two alternative values for the latter field. Consistency takes place if these values coincide identically (with respect to the initial data), and, moreover, if the resulting field $x_{ij}$ satisfies equation \eqref{eq: Eij}. In other words:
\begin{definition}\label{def: 2D corner eqs consist}
The system of 2D corner equations \eqref{eq: E}--\eqref{eq: Eij} is called \emph{consistent}, if it has the minimal possible rank 2, i.e., if exactly two of these four equations are independent.
\end{definition}

%%%%%%%%%%%%%%%%%%%%%%%%%%%%%%%%%%%%%%%%%%%%%%%%%%%%%%%%%%%%%%%%%%
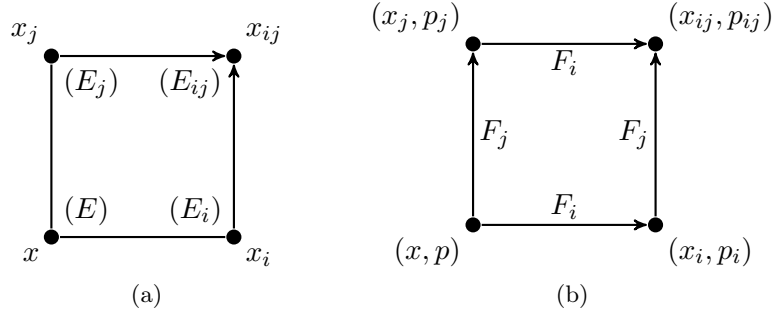
\begin{figure}[tbp]
\centering
\subfloat[]{\label{Fig: consistency1}
\begin{tikzpicture}[auto,scale=1.2,>=stealth',inner sep=2]
   \node (x) at (0,0) [circle,fill,thick,label=-135:$x$,label=45:$(E)$] {};
   \node (x1) at (2,0) [circle,fill,thick,label=135:$(E_{i})$,label=-45:$x_i$] {};
   \node (x2) at (0,2) [circle,fill,thick,label=-45:$(E_{j})$,label=135:$x_j$] {};
   \node (x12) at (2,2) [circle,fill,thick,label=45:$x_{ij}$,label=-135:$(E_{ij})$] {};
   \draw [thick,->] (x) to (x1) to (x12);
   \draw [thick,->] (x) to (x2) to (x12);
\end{tikzpicture}
}\qquad
\subfloat[]{\label{Fig: consistency2}
\begin{tikzpicture}[auto,scale=1.2,>=stealth',inner sep=2]
   \node (x) at (0,0) [circle,fill,thick,{label=-135:$(x,p)$}] {};
   \node (x1) at (2,0) [circle,fill,thick,{label=-45:$(x_i,p_i)$}] {};
   \node (x2) at (0,2) [circle,fill,thick,{label=135:$(x_j,p_j)$}] {};
   \node (x12) at (2,2) [circle,fill,thick,{label=45:$(x_{ij},p_{ij})$}] {};
   \draw [thick,->] (x) to node {$F_{i}$} (x1);
   \draw [thick,->] (x) to node [swap] {$F_{j}$} (x2);
   \draw [thick,->] (x2) to node [swap]{$F_{i}$} (x12);
   \draw [thick,->] (x1) to node {$F_{j}$} (x12);
\end{tikzpicture}
}
\caption{Consistency of 2D corner equations: \protect\subref{Fig: consistency1} Start with data $x$, $x_i$, $x_j$ related by 2D~corner equation \eqref{eq: E}; solve 2D~corner equations \eqref{eq: Ei} and \eqref{eq: Ej} for $x_{ij}$; consistency means that the two values of $x_{ij}$ coincide identically and satisfy 2D~corner equation \eqref{eq: Eij}. \protect\subref{Fig: consistency2} Maps $F_{i}$ and $F_{j}$ commute.}
\label{Fig: consistency}
\end{figure}
%%%%%%%%%%%%%%%%%%%%%%%%%%%%%%%%%%%%%%%%%%%%%%%%%%%%%%%%%%%%%%%%%%

Observe that 2D corner equations \eqref{eq: E}--\eqref{eq: Eij} can be put as
\begin{alignat}{4}\label{eq: 2D corner eqs}
&\frac{\partial S^{ij}}{\partial x}=0,&\quad
&\frac{\partial S^{ij}}{\partial x_i}=0,&\quad
&\frac{\partial S^{ij}}{\partial x_j}=0,&\quad
&\frac{\partial S^{ij}}{\partial x_{ij}}=0,
\end{alignat}
where $S^{ij}$ is the action along the boundary of an oriented elementary square $\sigma_{ij}$ (this action can be identified with the discrete exterior derivative $d\cL$ evaluated at $\sigma_{ij}$),
\[
S^{ij}=d\cL(\sigma_{ij})=\Delta_i\cL(\sigma_j)-\Delta_j\cL(\sigma_i)
=\Lambda_i(x,x_i)+\Lambda_j(x_i,x_{ij})-\Lambda_i(x_j,x_{ij})-\Lambda_j(x,x_j).
\]
Here and in what follows, $\Delta_i=T_i-I$ is the difference operator, $T_i$ being the shift operator in the $i^{\mathrm th}$ coordinate, so that, e.g., $T_i x=T_i x(n)=x(n+e_i)=x_i$ and $T_i x_j=T_i x(n+e_j)=x(n+e_i+e_j)=x_{ij}$.

The main feature of our definition is that the ``almost closedness'' of the 1-form $\cL$ on solutions of the system of 2D~corner equations is, so to say, built-in from the outset.

\begin{theorem}\label{th: discr almost closed}
If the system of 2D~corner equations \eqref{eq: corner eqs} is consistent, then, for any pair of the coordinate directions $i,j$, the action $S^{ij}$ over the boundary of an elementary square of these coordinate directions is constant on solutions:
\[
S^{ij}(x,x_i,x_{ij},x_j)=\ell^{ij}=\mathrm{const}  \pmod{\partial S^{ij}/\partial x=0,\ldots,
\partial S^{ij}/\partial x_{ij}=0}.
\]
\end{theorem}
In particular, if all these constants $\ell^{ij}$ vanish, then the discrete 1-form $\cL$ is closed on solutions of the Euler-Lagrange equations, so that the critical value of the action functional $S_\Sigma$ does not depend on the choice of the curve $\Sigma$ connecting two given points in $\bbZ^m$.
\medskip

Consistency of the system of 2D corner equations \eqref{eq: E 0}--\eqref{eq: E12 0} is equivalent to existence of a function $p:\bbZ^m\to {\mathcal X}$ satisfying all the relations
\begin{alignat}{2}
    p & = \frac{\partial\Lambda_i(x,x_i)}{\partial x},&\qquad& i=1,\ldots,m,
    \label{eq: discr p 1}\\
    p  & = -\frac{\partial\Lambda_i(x_{-i},x)}{\partial x},&\qquad& i=1,\ldots,m.
     \label{eq: discr p 2}
\end{alignat}
Suppose that all the equations (\ref{eq: discr p 1}) can be solved for $x_i$ in terms of $x,p$, so that equations
\begin{alignat}{2}\label{eq: single Lagr map}
    &p=\frac{\partial\Lambda_i(x,x_i)}{\partial x},&\qquad
    &p_i=-\frac{\partial\Lambda_i(x,x_i)}{\partial x_i},
\end{alignat}
define symplectic maps $F_i:(x,p)\mapsto(x_i,p_i)$.
\begin{theorem}\label{thm 1d commute}
For a consistent one-dimensional pluri-Lagrangian system, maps $F_i$ commute:
\begin{equation}\label{eq: commute}
F_i\circ F_j=F_j\circ F_i,
\end{equation}
see Figure~\ref{Fig: consistency}\subref{Fig: consistency2}). Conversely, for a given system of $m$ commuting symplectic maps
$F_i$ admitting Lagrangians (generating functions) $\Lambda_i$, the 1-form $\cL$ on $\bbZ^m$ defined by $\cL(\sigma_i)=\Lambda_i(x,x_i)$, generates a consistent one-dimensional pluri-Lagrangian system.
\end{theorem}

We are mainly interested in one-parameter families of commuting symplectic maps $F_i:(x,p)\mapsto(x_i,p_i)$, depending on the parameter $\lambda$, and admitting a generating function $\Lambda(x,x_i;\lambda)$. To avoid double indices, we will denote its action by a tilde:
\begin{equation}\label{eq: BT1}
F_i:\
p=-\frac{\partial \Lambda(x,\wx;\lambda)}{\partial x},\quad
\wip=\frac{\partial \Lambda(x,\wx;\lambda)}{\partial \wx}.
\end{equation}
When considering a second such map, say $F_j$, corresponding to another parameter value $\mu$, we will denote its action by a hat:
\begin{equation}\label{eq: BT2}
F_j:\
p=-\frac{\partial \Lambda(x,\whx;\mu)}{\partial x},\quad
\whp=\frac{\partial \Lambda(x,\whx;\mu)}{\partial \whx}.
\end{equation}
We assume that \eqref{eq: commute} is satisfied for any two parameter values $\lambda$, $\mu$.
Corner equations \eqref{eq: E}--\eqref{eq: Eij} in these new notations take the form 
\begin{align}
\frac{\partial\Lambda(x,\wx;\lambda)}{\partial x}-\frac{\partial\Lambda(x,\whx;\mu)}{\partial x}= & \ 0,
\label{eq: E tilde}\tag{$E$}\\
\frac{\partial\Lambda(x,\wx;\lambda)}{\partial \wx}+\frac{\partial\Lambda(\wx,\widehat{\wx};\mu)}{\partial \wx}= & \ 0,
\label{eq: E1 tilde}\tag{$E_i$}\\
\frac{\partial\Lambda(x,\whx;\mu)}{\partial \whx}+\frac{\partial\Lambda(\whx,\widehat{\wx};\lambda)}{\partial \whx}= & \ 0,
\label{eq: E2 tilde}\tag{$E_j$}\\
\frac{\partial\Lambda(\whx,\widehat{\wx};\lambda)}{\partial \widehat{\wx}}-
\frac{\partial\Lambda(\wx,\widehat{\wx};\mu)}{\partial \widehat{\wx}} = & \ 0.
\label{eq: E12 tilde}\tag{$E_{ij}$}
\end{align}

\begin{theorem}\label{thm spectrality}
For a consistent system of corner equations  \eqref{eq: E tilde}--\eqref{eq: E12 tilde}, the discrete multi-time Lagrangian 1-form is closed on solutions,  that is, $\ell(\lambda,\mu)=0$, if and only if $\partial \Lambda(x,\wx;\lambda)/\partial \lambda$ is a common integral of motion for all $F_j$. 
\end{theorem}
\begin{proof} According to Theorem \ref{th: discr almost closed}, 
\begin{equation}\label{eq: BT closure}
    \Lambda(x,\wx;\lambda)+\Lambda(\wx,\widehat{\wx};\mu)-
    \Lambda(x,\whx;\mu)-\Lambda(\whx,\widehat{\wx};\lambda)=\ell(\lambda,\mu)
\end{equation}
is constant on solutions of corner equations  \eqref{eq: E tilde}--\eqref{eq: E12 tilde}. This constant is obviously skew-symmetric: $\ell(\mu,\lambda)=-\ell(\lambda,\mu)$. Then $\ell(\lambda,\mu)=0$ is equivalent to $\partial \ell/\partial\lambda=0$. Differentiating equation \eqref{eq: BT closure} with respect to $\lambda$ and taking into account that the terms containing $\partial \wx/\partial\lambda$ etc., appearing due to the chain rule, vanish by virtue of corner equations  \eqref{eq: E tilde}--\eqref{eq: E12 tilde}, we arrive at
$$
\frac{\partial  \Lambda(x,\wx;\lambda)}{\partial\lambda}-\frac{\partial \Lambda(\whx,\widehat{\wx};\lambda)}{\partial\lambda}\, = \,  0.
$$
This is equivalent to $\partial \Lambda(x,\wx;\lambda)/\partial \lambda$ being an integral of motion for $F_j$. 
\end{proof}

\paragraph{Bibliographical remarks.} 
The first example of the pluri-Lagrangian structure for discrete 1D systems was given in \cite{YLN11}. A general theory was developed in \cite{Su13}, on which this section is based. The result of Theorem \ref{thm spectrality}  is a re-formulation of the mysterious ``spectrality property'' of B\"acklund transformations discovered by Kuznetsov and Sklyanin \cite{KS}. Spectrality was originally understood as the property of $\partial \Lambda(x,\wx;\lambda)/\partial \lambda$ to be a spectral invariant of the Lax matrix for the system at hand. Our re-formulation avoids an a priori knowledge of the Lax matrix. We remark that the problem of completeness of the set of the integrals encoded in this quantity requires for a separate study in both approaches.

%%%%%%%%%%%%%%%%%%%%%%%%%%%%%%%
%%%%%%%%%%%%%%%%%%%%%%%%%%%%%%%
\section{Commutativity of B\"acklund transformations for exponential Toda lattice}
\label{sect: BT Toda}
%%%%%%%%%%%%%%%%%%%%%%%%%%%%%%%
%%%%%%%%%%%%%%%%%%%%%%%%%%%%%%%

Here we illustrate the main constructions by the example of B\"acklund transformations for the exponential Toda lattice \eqref{TL l New}. The maps $\mathrm{dTL}(\lambda)=F_i: T^*\bbR^N\to T^*\bbR^N$ are given by equations \eqref{dTL l}:
\begin{equation}\label{eq: BT1 Toda}
F_i:\ \left\{\begin{array}{l}
p_k=\dfrac{1}{\lambda}\left(\eto{\wx_k\nm x_k}-1\right)+\lambda \eto{x_k\nm\wx_{k-1}},\vspace{0.2truecm}\\
\wip_k=\dfrac{1}{\lambda}\left(\eto{\wx_k\nm x_k}-1\right)+\lambda \eto{x_{k+1}\nm\wx_k}.
\end{array}\right.
\end{equation}
The corresponding Lagrangian is given by
\begin{equation}\label{eq: BT Toda Lagr}
\Lambda(x,\wx;\lambda)=\frac{1}{\lambda}\sum_{k=1}^N\left(\eto{\wx_k\nm x_k}-1-(\wx_k\nm x_k)\right)-
    \lambda\sum_{k=1}^N \eto{x_{k+1}\nm\wx_k},
\end{equation}
and the standard single-time Euler-Lagrange equations coincide with (\ref{dTL l New}) with $h=\lambda$.

As discussed in the previous section, commutativity of the maps $F_i$, $F_j$ (in the open-end case, when they are well-defined, i.e., single-valued) is equivalent to consistency of the system of corner equations:
\begin{align}
 \dfrac{1}{\lambda}\Big(\eto{\wx_k\nm x_k}-1\Big)+\lambda \eto{x_k\nm\wx_{k-1}}= & \
 \dfrac{1}{\mu}\Big(\eto{\whx_k\nm x_k}-1\Big)+\mu \eto{x_k\nm \whx_{k-1}}, 
 \label{eq: BT Toda E}\tag{$E$} \\
 \dfrac{1}{\lambda}\Big(\eto{\wx_k\nm x_k}-1\Big)+\lambda \eto{x_{k+1}\nm \wx_k}= & \
 \dfrac{1}{\mu}\Big(\eto{\widehat{\widetilde{x}}_k\nm \wx_k}-1\Big)
 +\mu \eto{\wx_k\nm \widehat{\widetilde{x}}_{k-1}},
  \label{eq: BT Toda E1}\tag{$E_i$}\\
 \dfrac{1}{\mu}\Big(\eto{\whx_k\nm x_k}-1\Big)+\mu \eto{x_{k+1}\nm \whx_k}= & \
 \dfrac{1}{\lambda}\Big(\eto{\widehat{\widetilde{x}}_k\nm \whx_k}-1\Big)
 +\lambda \eto{\whx_k\nm \widehat{\widetilde{x}}_{k-1}},
  \label{eq: BT Toda E2}\tag{$E_j$}\\
 \dfrac{1}{\lambda}\Big(\eto{\widehat{\widetilde{x}}_k\nm \whx_k}-1\Big)
 +\lambda \eto{\whx_{k+1}\nm \widehat{\widetilde{x}}_k}= & \
 \dfrac{1}{\mu}\Big(\eto{\widehat{\widetilde{x}}_k\nm \wx_k}-1\Big)
 +\mu \eto{\wx_{k+1}\nm \widehat{\widetilde{x}}_k}.
  \label{eq: BT Toda E12}\tag{$E_{ij}$}
\end{align}
We have to clarify the meaning of the both notions (commutativity of  $F_i$, $F_j$ and consistency of corner equations) in the periodic case. To do this, we prove the following statement.
\begin{theorem}\label{th: superposition}
Suppose that the fields $x$, $\wx$, $\whx$ satisfy corner equations \eqref{eq: BT Toda E}. Define the fields $\widehat{\widetilde{x}}$ by any of the following two formulas, which are equivalent by virtue of \eqref{eq: BT Toda E}:
\begin{equation}\label{eq: BT Toda S1}\tag{$S1$}
\dfrac{1}{\lambda}\Big(\eto{\widehat{\widetilde{x}}_k\nm\whx_k}-1\Big)
-\dfrac{1}{\mu}\Big(\eto{\widehat{\widetilde{x}}_k\nm\wx_k}-1\Big)
+\lambda \eto{x_{k+1}\nm\wx_k}-\mu \eto{x_{k+1}\nm\whx_k}=0,
\end{equation}
\begin{equation}\label{eq: BT Toda S2}\tag{$S2$}
\dfrac{1}{\lambda}\Big(\eto{\wx_{k+1}\nm x_{k+1}}-1\Big)
-\dfrac{1}{\mu}\Big(\eto{\whx_{k+1}\nm x_{k+1}}-1\Big)
+\lambda \eto{\whx_{k+1}\nm\widehat{\widetilde{x}}_k}
-\mu \eto{\wx_{k+1}\nm\widehat{\widetilde{x}}_k}=0,
\end{equation}
called superposition formulas. Then corner equations \eqref{eq: BT Toda E1}, \eqref{eq: BT Toda E2}, \eqref{eq: BT Toda E12} are satisfied, as well.
\end{theorem}
\begin{proof} First of all, we show that equations (\ref{eq: BT Toda S1}) and (\ref{eq: BT Toda S2}) are indeed equivalent by virtue of (\ref{eq: BT Toda E}). For this, we re-write these equations in algebraically equivalent forms:
\begin{equation}\label{eq: BT Toda S1 3leg xk+1}
\frac{\lambda-\mu}{\eto{\widehat{\widetilde{x}}_k\nm x_{k+1}}-\lambda\mu}=\lambda \eto{x_{k+1}\nm \wx_k}-\mu \eto{x_{k+1}\nm \whx_k},
\end{equation}
and
\begin{equation}\label{eq: BT Toda S2 3leg xk+1}
\frac{(\lambda-\mu)\eto{x_{k+1}\nm \widehat{\wx}_k}}{1-\lambda\mu \eto{x_{k+1}\nm \widehat{\wx}_k}}=
\dfrac{1}{\mu}\Big(\eto{\whx_{k+1}\nm x_{k+1}}-1\Big)-\dfrac{1}{\lambda}\Big(\eto{\wx_{k+1}\nm x_{k+1}}-1\Big),
\end{equation}
respectively. The left-hand sides of the latter two equations are equal. Thus, their difference coincides with (\ref{eq: BT Toda E}).

 Second, we show that equations (\ref{eq: BT Toda S1}) and (\ref{eq: BT Toda S2}) yield (\ref{eq: BT Toda E1}). (For (\ref{eq: BT Toda E2}) everything is absolutely analogous.) For this aim, we re-write these equations in still other algebraically equivalent forms. Namely, (\ref{eq: BT Toda S1}) is equivalent to
\begin{equation}\label{eq: BT Toda S1 3leg wx}
\eto{\widehat{\widetilde{x}}_k\nm \wx_k}=
 \lambda\mu \eto{x_{k+1}\nm \wx_k}+\frac{\mu-\lambda}{\mu \eto{\wx_k\nm \whx_k}-\lambda},
\end{equation}
while (\ref{eq: BT Toda S2}) with $k$ replaced by $k-1$ is equivalent to
\begin{equation}\label{eq: BT Toda S2 3leg wx}
\lambda\mu \eto{\wx_k\nm \widehat{\widetilde{x}}_{k-1}}=
 \eto{\wx_k\nm x_k}+\frac{\lambda-\mu}{\mu -\lambda \eto{\whx_k\nm \wx_k}}.
\end{equation}
An obvious linear combination of these expressions leads to
\[
\frac{1}{\mu} \eto{\widehat{\widetilde{x}}_k\nm \wx_k}+\mu \eto{\wx_k\nm \widehat{\widetilde{x}}_{k-1}}=
\lambda \eto{x_{k+1}\nm \wx_k}+\frac{1}{\lambda} \eto{\wx_k\nm x_k}+\frac{\lambda-\mu}{\lambda\mu},
\]
which is nothing but (\ref{eq: BT Toda E1}).

Third, we observe that the sum of equations (\ref{eq: BT Toda E}), (\ref{eq: BT Toda S1}) and (\ref{eq: BT Toda S2}) is nothing but the corner equation (\ref{eq: BT Toda E12}). \end{proof}

\textbf{Remark.} Observe that equations (\ref{eq: BT Toda S1}) and (\ref{eq: BT Toda S2}) are quad-equations with respect to
\[
\Big(\eto{x_{k+1}},\eto{\wx_k},\eto{\whx_k},\eto{\widehat{\widetilde{x}}_k}\Big),\quad \mathrm{resp.}\quad
\Big(\eto{x_{k+1}},\eto{\wx_{k+1}},\eto{\whx_{k+1}},\eto{\widehat{\widetilde{x}}_k}\Big),
\]
i.e., they can be formulated as vanishing of multi-affine polynomials of the four specified variables. Equations
(\ref{eq: BT Toda S1 3leg xk+1}) and (\ref{eq: BT Toda S2 3leg xk+1}) are then interpreted as the three-leg forms of these quad-equations, centered at $x_{k+1}$. Similarly, equations (\ref{eq: BT Toda S1 3leg wx}) and (\ref{eq: BT Toda S2 3leg wx}) are the three-leg forms of the quad-equations, centered at $\wx_k$.
\medskip

Theorem \ref{th: superposition} allows us to achieve an exhaustive understanding of consistency and commutativity for double-valued B\"acklund transformations. First, suppose that we are given the fields $x$, $\wx$, $\whx$ satisfying corner equation (\ref{eq: BT Toda E}). Each of corner equations (\ref{eq: BT Toda E1}), (\ref{eq: BT Toda E2}) produces two values for $\widehat{\wx}$. Then consistency is reflected in the following fact: one of the values for $\widehat{\wx}$ obtained from (\ref{eq: BT Toda E1}) coincides with one of the values for $\widehat{\wx}$ obtained from (\ref{eq: BT Toda E2}), see Fig.~\ref{Fig: consistency 4-valued}\subref{Fig: commut x}. Indeed, this common value is nothing but $\widehat{\wx}$ obtained from the superposition formulas (\ref{eq: BT Toda S1}), (\ref{eq: BT Toda S2}), as in Theorem \ref{th: superposition}.

The ``loose ends'' on Fig.~\ref{Fig: consistency 4-valued}\subref{Fig: commut x} are best explained by considering the (double-valued) maps $F_i$, $F_j$, i.e., by working with the variables $(x,p)$ rather than with the variables $x$ alone. Indeed, each of the compositions $F_i\circ F_j$ and $F_j\circ F_i$ is four-valued. It follows from Theorem \ref{th: superposition} that their branches must pairwise coincide, as shown on Fig.~\ref{Fig: consistency 4-valued}\subref{Fig: commut x,p}. Indeed, Theorem \ref{th: superposition} delivers four possible values for $(\wx,\whx,\widehat{\wx})$ satisfying all corner equations (\ref{eq: BT Toda E})--(\ref{eq: BT Toda E12}), namely one $\widehat{\wx}$ for each of the four possible combinations of $(x,\wx,\whx)$.

%--------------------------------------------------------------------------
\begin{figure}[tbp]
\centering
\subfloat[]{\label{Fig: commut x}
\begin{tikzpicture}[auto,scale=1.7,>=stealth',inner sep=2]
   \node (x) at (0,0) [circle,fill,thick,label=-135:$x$,label=45:$\left(E\right)$] {};
   \node (x1) at (2,0) [circle,fill,thick,label=135:$\left(E_{i}\right)$,label=-45:$\wx$] {};
   \node (x2) at (0,2) [circle,fill,thick,label=-45:$\left(E_{j}\right)$,label=135:$\whx$] {};
   \node (x12) at (2,2) [circle,fill,thick,label=45:$\widehat{\wx}$,label=-135:$\left(E_{ij}\right)$] {};
   \node (x121) at (2.25,2) [circle,fill,thick] {};
   \node (x122) at (2,2.25) [circle,fill,thick] {};
   \draw [thick,->] (x) to (x1) to (x12);
   \draw [thick,->] (x) to (x2) to (x12);
   \draw [thick,->] (x) to (x1) to (x121);
   \draw [thick,->] (x) to (x2) to (x122);
\end{tikzpicture}
}\qquad
\subfloat[]{\label{Fig: commut x,p}
\begin{tikzpicture}[auto,scale=1.7,>=stealth',inner sep=2]
   \node (x) at (0,0) [circle,fill,thick,{label=-135:$\left(x,p\right)$}] {};
   \node (x1) at (2,0) [circle,fill,thick,{label=-45:$\left(\wx,\wip\right)$}] {};
   \node (x11) at (1.75,0.25) [circle,fill,thick] {};
   \node (x2) at (0,2) [circle,fill,thick,{label=135:$\left(\whx,\whp\right)$}] {};
   \node (x22) at (0.25,1.75) [circle,fill,thick] {};
   \node (x12) at (2,1.75) [circle,fill,thick,{label=45:$\big(\widehat{\wx},\widehat{\wip}\big)$}] {};
   \node (x121) at (2.25,1.5) [circle,fill,thick] {};
   \node (x122) at (1.75,2) [circle,fill,thick] {};
   \node (x123) at (1.5,2.25) [circle,fill,thick] {};
   \draw [thick,->] (x) to (x1);
   \draw [thick,->] (x) to node {$F_{i}$} (x11);
   \draw [thick,->] (x) to (x2);
   \draw [thick,->] (x) to node [swap] {$F_{j}$} (x22);
   \draw [thick,->] (x2) to (x122);
   \draw [thick,->] (x2) to node{$F_{i}$} (x123);
   \draw [thick,->] (x22)  to (x121);
   \draw [thick,->] (x22) to (x12);
   \draw [thick,->] (x1) to  (x122);
   \draw [thick,->] (x1) to node [swap] {$F_{j}$} (x121);
   \draw [thick,->] (x11) to  (x123);
   \draw [thick,->] (x11) to  (x12);
\end{tikzpicture}
}
\caption{Consistency of multi-time Euler-Lagrange equations in the case of periodic boundary conditions. \\
\protect\subref{Fig: commut x} Given the fields $x$, $\wx$, $\whx$ satisfying corner equation (\ref{eq: BT Toda E}), each of corner equations (\ref{eq: BT Toda E1}), (\ref{eq: BT Toda E2}) has two solutions $\widehat{\wx}$. One of the values for $\widehat{\wx}$ obtained from (\ref{eq: BT Toda E1}) coincides with one of the values for $\widehat{\wx}$ obtained from (\ref{eq: BT Toda E2}). This common value of $\widehat{\wx}$, together with $\wx$ and $\whx$, satisfies (\ref{eq: BT Toda E12}). \\
\protect\subref{Fig: commut x,p} Each one of the four branches of $F_i\circ F_j$ coincides with exactly one of the four branches of $F_j\circ F_i$.}
\label{Fig: consistency 4-valued}
\end{figure}
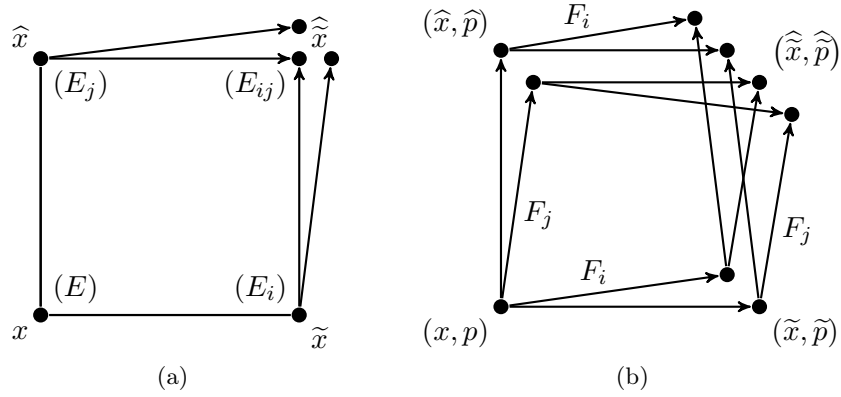

\begin{theorem}\label{th: Toda closure}
The quantity 
\begin{equation}\label{eq: BT Toda P}
P(x,\wx;\lambda)=\prod_{k=1}^N \eto{\wx_k\nm x_k}
\end{equation}
is a common integral of motion for all maps $F_j$. Equivalently, the discrete multi-time Lagrangian 1-form $\cL$ is closed on any solution of the corner equations \eqref{eq: BT Toda E}--\eqref{eq: BT Toda E12}.
\end{theorem}
\begin{proof} First of all, we show that the closure relation $d\cL=\ell(\lambda,\mu)=0$ is equivalent to
\begin{equation}\label{eq: BT Toda for closeness}
\sum_{k=1}^N(\widehat{\widetilde{x}}_k-\wx_k-\whx_k+x_k)=0\quad \Leftrightarrow\quad
\prod_{k=1}^N \eto{\widehat{\widetilde{x}}_k\nm \wx_k\nm \whx_k\np x_k}=1.
\end{equation}
This can be done in two different ways. On one hand, combining (\ref{eq: BT Toda S1}), (\ref{eq: BT Toda S2}) with (\ref{eq: BT Toda E}), we arrive at the two formulas
\begin{align}
  \dfrac{1}{\lambda}\eto{\wx_{k+1}\nm x_{k+1}}-\dfrac{1}{\mu}\eto{\whx_{k+1}\nm x_{k+1}}
  -\dfrac{1}{\lambda} \eto{\widehat{\widetilde{x}}_k\nm \whx_k}
  +\dfrac{1}{\mu} \eto{\widehat{\widetilde{x}}_k\nm \wx_k} =& 0,
  \label{eq: BT Toda superp 1}\\
  \lambda \eto{x_{k+1}\nm \wx_k}-\mu \eto{x_{k+1}\nm \whx_k}
  -\lambda \eto{\whx_{k+1}\nm \widehat{\widetilde{x}}_k}
  +\mu \eto{\wx_{k+1}\nm \widehat{\widetilde{x}}_k}  = & 0.
  \label{eq: BT Toda superp 2}
\end{align}
By virtue of these formulas, most of the terms on the left-hand side of (\ref{eq: BT closure}) with the Lagrange function (\ref{eq: BT Toda Lagr}) cancel, leaving us with
\[
\ell(\lambda,\mu)=\left(\frac{1}{\lambda}-\frac{1}{\mu}\right)
\sum_{k=1}^N(\widehat{\widetilde{x}}_k-\wx_k-\whx_k+x_k).
\]
Alternatively, we can refer to Theorem \ref{thm spectrality} stating that  $\ell(\lambda,\mu)=0$ is equivalent to $\partial \Lambda(x,\wx;\lambda)/\partial \lambda$ being an integral of motion for $F_j$. One easily computes:
\[
\frac{\partial \Lambda(x,\wx;\lambda)}{\partial \lambda}=-\frac{1}{\lambda}\sum_{k=1}^N p_k+\frac{1}{\lambda^2}\sum_{k=1}^N(\wx_k-x_k),
\]
the first sum on the right-hand side being an obvious integral of motion.

Now the desired result (\ref{eq: BT Toda for closeness}) can be derived from the following form of the superposition formula:
\begin{equation}\label{eq: BT Toda superp}
    \eto{\widehat{\widetilde{x}}_k\nm \wx_k\nm \whx_k\np x_{k+1}}=\dfrac{\lambda \eto{\whx_{k+1}}- \mu \eto{\wx_{k+1}}}{\lambda \eto{\whx_{k}}-\mu \eto{\wx_{k}}},
\end{equation}
which is in fact equivalent to either of equations (\ref{eq: BT Toda superp 1}), (\ref{eq: BT Toda superp 2}). In the periodic case (\ref{eq: BT Toda for closeness}) follows directly by multiplying equations (\ref{eq: BT Toda superp}) for $1\le k\le N$, in the open-end case equation (\ref{eq: BT Toda superp}) holds true for $1\le k\le N-1$ and has to be supplemented by the boundary counterparts
\begin{equation}\label{eq: BT Toda superp open boundary}
 \eto{x_{1}}=\dfrac{\lambda \eto{\whx_1}-\mu \eto{\wx_1}}{\lambda-\mu}, \quad
 \eto{\widehat{\widetilde{x}}_N\nm \wx_N\nm \whx_N}=\dfrac{\lambda-\mu}{\lambda \eto{\whx_N}-\mu \eto{\wx_N}},
\end{equation}
which are equivalent to (\ref{eq: BT Toda E}) for $k=1$, resp. to (\ref{eq: BT Toda E12}) for $k=N$.
\end{proof}

The conserved quantity \eqref{eq: BT Toda P}, expressed through $(x,p)$, can be given a beautiful expression in terms of matrices which turn out to be transition matrices of the zero curvature representation for $F_i$ (but the latter notion is not necessary for establishing the result).
\begin{theorem}\label{th: BT Toda zcr}
Set
\begin{equation*}\label{eq: BT Toda zcr L}
    L_k(x,p;\lambda)=\begin{pmatrix} 1+\lambda p_k & -\lambda^2 \eto{x_k\nm x_{k-1}} \vspace{2mm}\\
                                     1 & 0 \end{pmatrix},
\end{equation*}
and
\begin{equation*}\label{eq: BT Toda zcr T}
    T_N(x,p;\lambda)=L_N(x,p;\lambda) \cdots L_2(x,p;\lambda)L_1(x,p;\lambda).
\end{equation*}
Then in the periodic case conserved quantity \eqref{eq: BT Toda P} is an eigenvalue of $T_N(x,p;\lambda)$, while in the open-end case it is equal to $\tr T_N(x,p;\lambda)$.
\end{theorem}
\begin{proof} We use the following notation for the action of matrices from $GL(2,\mathbb C)$ on $\mathbb C$ by M\"obius transformations:
\[
\begin{pmatrix} a & b \\ c & d \end{pmatrix}[z]=\frac{az+b}{cz+d}.
\]
With this notation, we can re-write the first equation in (\ref{eq: BT1 Toda}) as
\[
\eto{\wx_k\nm x_k}=(1+\lambda p_k)-\lambda^2 \eto{x_k\nm \wx_{k-1}}=L_k(x,p;\lambda)\big[\eto{\wx_{k-1}\nm x_{k-1}}\big].
\]
This is equivalent to saying that
\[
L_k(x,p;\lambda)\begin{pmatrix} \gamma_{k-1} \\ 1 \end{pmatrix} \sim
\begin{pmatrix} \gamma_k \\ 1 \end{pmatrix}, \quad \mathrm{where} \quad \gamma_k=\eto{\wx_k\nm x_k}.
\]
The proportionality coefficient is easily determined by comparing the second components of these vectors:
\begin{equation}\label{eq: BT Toda zcr transition}
L_k(x,p;\lambda)\begin{pmatrix} \gamma_{k-1} \\ 1 \end{pmatrix} = \gamma_{k-1} \begin{pmatrix} \gamma_k \\ 1 \end{pmatrix}.
\end{equation}
Now in the periodic case we see that $(\gamma_N, 1)^\mathrm{T}$ is an eigenvector of $T_N(x,p;\lambda)$ with the eigenvalue $\prod_{k=1}^N \gamma_k$. In the open-end case, equation (\ref{eq: BT Toda zcr transition}) holds true for $2\le k\le N$, and has to be supplemented by the following two relations:
\[
L_1(x,p;\lambda)\begin{pmatrix} 1 \\ 0 \end{pmatrix} = \begin{pmatrix} 1+\lambda p_1 \\ 1 \end{pmatrix}=\begin{pmatrix} \gamma_1 \\ 1 \end{pmatrix}, \quad
\mathrm{and}\quad
\begin{pmatrix} 1 & 0 \end{pmatrix}\begin{pmatrix} \gamma_N \\ 1 \end{pmatrix}
= \gamma_N.
\]
As a consequence,
\[
\begin{pmatrix} 1 & 0 \end{pmatrix}T_N(x,p;\lambda)\begin{pmatrix} 1 \\ 0 \end{pmatrix}=\prod_{k=1}^N \gamma_k. 
\]
Thus, the (11)-entry of $T_N(x,p;\lambda)$ equals $\prod_{k=1}^N \gamma_k$. It coincides with $\tr T_N(x,p;\lambda)$, since the (22)-entry of this matrix vanishes.
\end{proof}

\paragraph{Bibliographical remarks.} Presentation here is based on \cite{BPS13}. One finds there similar results  for all discrete Toda type systems given in Section \ref{sect Toda New}.

\section{General theory of discrete two-dimensional pluri-Lagrangian systems}
\label{sect: discr 2d results}

Multi-dimensional consistency is a fundamental integrability concept for quad-equations. In the present section, we address the question about an analog of this property for discrete variational Laplace type systems.
\begin{definition} {\itbf (Two-dimensional pluri-Lagrangian problem)} \label{def:pluriLagr problem d=2}
Let $\cL$ be a discrete 2-form on $\bbZ^{m}$ (a function of oriented elementary squares
$\sigma_{ij}=(n,n+e_{i},n+e_{i}+e_{j},n+e_{j})$
such that $\cL(\sigma_{ij})=-\cL(\sigma_{ji})$), depending on a function $x:\bbZ^{m}\to\mathcal{X}$, where $\mathcal{X}$ is some vector space. It is supposed that $\cL(\sigma_{ij})$ depends on the values $x,x_i,x_j,x_{ij}$ of the field at the four vertices of the elementary square $\sigma_{ij}$. 
\begin{itemize}
\item To an arbitrary quad-surface $\Sigma$ in $\bbZ^{m}$ (an oriented surface in $\mathbb R^m$ composed of elementary squares of $\mathbb Z^m$), there corresponds the {\itbf action functional}
\[
S_{\Sigma}=\sum_{\sigma\in\Sigma}\cL(\sigma)
\]
(it depends only on the fields at vertices of $\Sigma$).
\item We say that the field $x:V(\Sigma)\to\mathcal{X}$ is a critical point of $S_{\Sigma}$, if at any interior point $n\in V(\Sigma)$, we have
\[\frac{\partial S_{\Sigma}}{\partial x(n)}=0.
\]
\item We say that the field $x:\bbZ^{m}\to\mathcal{X}$ solves the {\itbf pluri-Lagrangian problem} for the Lagrangian 2-form $\cL$ if, for any quad-surface $\Sigma$ in $\bbZ^{m}$, the restriction $\left.x\right|_{V(\Sigma)}$ is a critical point of the corresponding action $S_{\Sigma}$.
\end{itemize}
\end{definition}

One can show that the vertex star of any interior vertex of an oriented quad-surface $\Sigma$ in $\bbZ^m$ can be represented as a sum of (oriented) 3D~corners in $\bbZ^{m+1}$, see Figure \ref{Fig: star from corners}. Here, a \emph{3D~corner} is a quad-surface consisting of three elementary squares adjacent to a vertex of valence 3. As a consequence, the action over any vertex star can be represented as a sum of actions for several 3D~corners. Thus, Euler-Lagrange equation for any interior vertex $n$ of $\Sigma$ can be represented as a sum of several Euler-Lagrange equations for 3D~corners.

%%%%%%%%%%%%%%%%%%%%%%%%%%%%%%%%%%%%%%%%%%%%%%%%%%%%%%%%%%%%%%%%%
\begin{figure}[htbp]
   \centering
   \begin{tikzpicture}[scale=0.8, inner sep=2]  
      \node (x12) at (4,1) [label=-45:$n$] {};
      \node (x123) at (4,4) [label=135:$n+e_{m+1}$]{};     
      \draw [thick, dashed] (1,1) to (1,4) to (7,4) to (7,1);
      \draw [thick, dashed] (3,0) to (3,3)  to (5,5) to (5,2); 
      \draw [thick, dashed] (4,1) to (4,4); 
      \draw [very thick] (0,0) to (6,0) to (8,2) to (2,2) to (0,0);
      \draw [very thick] (1,1) to (7,1);
      \draw [very thick] (3,0) to (5,2);
   \end{tikzpicture}
        \caption{Vertex star of a quad-surface in $\mathbb Z^m$ as a sum of 3D corners in $\mathbb Z^{m+1}$, here $m=2$.}
           \label{Fig: star from corners}
\end{figure}
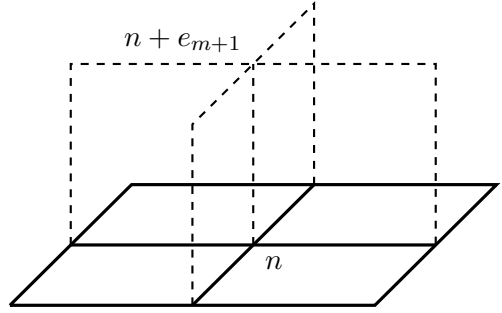
%%%%%%%%%%%%%%%%%%%%%%%%%%%%%%%%%%%%%%%%%%%%%%%%%%%%%%%%%%%%%%%%%

This justifies the following fundamental definition:
\begin{definition}\label{def:pluriLagr system}
The {\itbf system of 3D~corner equations} for a given discrete 2-form $\cL$ consists of discrete Euler-Lagrange equations for all possible 3D~corners in $\bbZ^m$. If the action for the surface of an oriented elementary cube $\sigma_{ijk}$ of the coordinate directions $i,j,k$ (which can be identified with the discrete exterior derivative $d\cL$ evaluated at $\sigma_{ijk}$) is denoted by
\begin{equation}\label{eq: Sijk}
S^{ijk}=d\cL(\sigma_{ijk})=\Delta_k\cL(\sigma_{ij})+\Delta_i\cL(\sigma_{jk})+\Delta_j\cL(\sigma_{ki}),
\end{equation}
then the system of 3D~corner equations consists of the eight equations
\begin{equation}\label{eq: corner eqs}
\begin{alignedat}{4}
&\dfrac{\partial S^{ijk}}{\partial x}=0, &\quad& \dfrac{\partial S^{ijk}}{\partial x_i}=0, &\quad& \dfrac{\partial S^{ijk}}{\partial x_j}=0, &\quad& \dfrac{\partial S^{ijk}}{\partial x_k}=0, \\
&\dfrac{\partial S^{ijk}}{\partial x_{ij}}=0, &\quad& \dfrac{\partial S^{ijk}}{\partial x_{jk}}=0, &\quad& \dfrac{\partial S^{ijk}}{\partial x_{ik}}=0, &\quad& \dfrac{\partial S^{ijk}}{\partial x_{ijk}}=0,
\end{alignedat}
\end{equation}
for each triple $i,j,k$.
\end{definition}
Thus, the system of 3D~corner equations encompasses all possible discrete Euler-Lagrange equations for all possible quad-surfaces $\Sigma$. In other words, solutions of a two-dimensional pluri-Lagrangian problem as introduced in Definition \ref{def:pluriLagr problem d=2} are precisely solutions of the corresponding system of 3D~corner equations.
\medskip

Of course, in order that the above definition be meaningful, the system of 3D~corner equations has to be \emph{consistent}:
\begin{definition} \label{def: corner eqs consist}
The system \eqref{eq: corner eqs} is called \emph{consistent}, if it has the minimal possible rank 2, i.e., if exactly two of these equations are independent.
\end{definition}
Like in the one-dimensional case, the ``almost closedness'' of the 2-form $\cL$ on solutions of the system of 3D~corner equations is built-in from the outset.
\begin{theorem}\label{Th: almost closed 2d}
If  the system of 3D~corner equations \eqref{eq: corner eqs} is consistent, then, for any triple of the coordinate directions $i,j,k$, the action $S^{ijk}$ over an elementary cube of these coordinate directions is constant on solutions:
\[
S^{ijk}\left(x,\ldots,x_{ijk}\right)=c^{ijk}=\mathrm{const}  \pmod{\partial S^{ijk}/\partial x=0,\ldots,\partial S^{ijk}/\partial x_{ijk}=0}.
\]
\end{theorem}
The most interesting case is, of course, when all $c^{ijk}=0$. Then $d\cL=0$, that is, the discrete 2-form $\cL$ is \emph{closed} on solutions of the system of 3D~corner equations, so that the critical value of the action $S_\Sigma$ does not change under perturbations of the quad-surface $\Sigma$ in $\bbZ^m$ fixing its boundary.
\medskip

\paragraph{Case of three-point 2-forms.} We formulated the system of 3D~corner equations for a generic 2-form $\cL$. We now specialize the theory for an important particular ansatz for the discrete 2-form, namely we consider the so called {\em three-point 2-form}:
\begin{equation}\label{eq: 3point}
    \cL(\sigma_{ij})=\Psi_i(x_i-x)-\Psi_j(x_j-x)-\Phi_{ij}(x_j-x_i),
\end{equation}
where the Lagrangians $\Psi_i$ and $\Phi_{ij}$ only depend on the differences of the fields at the end points, and the diagonal Lagrangians are skew-symmetric in the sense that $\Phi_{ij}(\xi)=-\Phi_{ji}(-\xi)$. Thus, one considers as a main building block in \eqref{action dRTL 2d} the discrete 2-form rather than edge dependent Lagrangians. This seemingly minor change of view point turns out to be very important conceptually.

For a 3-point 2-form, expression \eqref{eq: Sijk} specializes to
\begin{equation}\label{eq: 3point S}
\begin{split}
S^{ijk} & =  \Psi_i(x_{ik}-x_k)+\Psi_j(x_{ij}-x_i)+\Psi_k(x_{jk}-x_j)\\
        &   \phantom{=}\ -\Psi_i(x_{ij}-x_j)-\Psi_j(x_{jk}-x_k)-\Psi_k(x_{ik}-x_i)\\
        &   \phantom{=}\ -\Phi_{ij}(x_{jk}-x_{ik})-\Phi_{jk}(x_{ik}-x_{ij})-\Phi_{ki}(x_{ij}-x_{jk})\\
        &   \phantom{=}\ +\Phi_{ij}(x_j-x_i)+\Phi_{jk}(x_k-x_j)+\Phi_{ki}(x_i-x_k).
\end{split}
\end{equation}
Thus, $S^{ijk}$ depends on neither $x$ nor $x_{ijk}$, and its domain of definition is better visualized as an octahedron shown in Figure~\ref{Fig: 3-point case}\subref{Fig: octahedron}.

%%%%%%%%%%%%%%%%%%%%%%%%%%%%%%%%%%%%%%%%%%%%%%%%%%%%%%%%%%%%%%%%%%
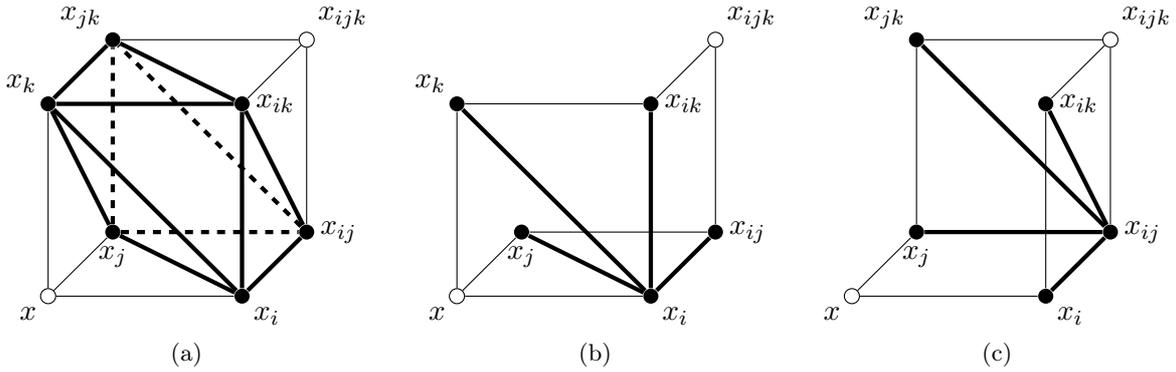
\begin{figure}[htbp]
   \centering
   \subfloat[]{ \label{Fig: octahedron}
   \begin{tikzpicture}[scale=0.85,inner sep=2]
      \node (x) at (0,0) [circle,draw,label=-135:$x$] {};
      \node (x1) at (3,0) [circle,fill,label=-45:$x_{i}$] {};
      \node (x2) at (1,1) [circle,fill,label=-90:$x_{j}$] {};
      \node (x3) at (0,3) [circle,fill,label=135:$x_{k}$] {};
      \node (x12) at (4,1) [circle,fill,label=0:$x_{ij}$] {};
      \node (x13) at (3,3) [circle,fill,label=0:$x_{ik}$] {};
      \node (x23) at (1,4) [circle,fill,label=135:$x_{jk}$] {};
      \node (x123) at (4,4) [circle,draw,label=45:$x_{ijk}$] {};
      \draw (x) to (x1);
      \draw (x) to (x2);
      \draw (x) to (x3);
      \draw [ultra thick] (x1) to (x2);
      \draw [ultra thick] (x1) to (x3);
      \draw [ultra thick] (x1) to (x12);
      \draw [ultra thick] (x1) to (x13);
      \draw [ultra thick] (x2) to (x3);
      \draw [ultra thick,dashed] (x2) to (x12);
      \draw [ultra thick,dashed] (x2) to (x23);
      \draw [ultra thick] (x3) to (x13);
      \draw [ultra thick] (x3) to (x23);
      \draw [ultra thick] (x12) to (x13);
      \draw (x12) to (x123);
      \draw [ultra thick,dashed] (x12) to (x23);
      \draw [ultra thick] (x13) to (x23);
      \draw (x13) to (x123);
      \draw (x23) to (x123);
   \end{tikzpicture}
  }
   \;
   \subfloat[]{\label{fig Ei}
   \begin{tikzpicture}[scale=0.85,inner sep=2]
      \node (x) at (0,0) [circle,draw,label=-135:$x$] {};
      \node (x1) at (3,0) [circle,fill,label=-45:$x_{i}$] {};
      \node (x2) at (1,1) [circle,fill,label=-90:$x_{j}$] {};
      \node (x3) at (0,3) [circle,fill,label=135:$x_{k}$] {};
      \node (x12) at (4,1) [circle,fill,label=0:$x_{ij}$] {};
      \node (x13) at (3,3) [circle,fill,label=0:$x_{ik}$] {};
%      \node (x23) at (1,4) [circle,draw,label=135:$x_{jk}$] {};
      \node (x123) at (4,4) [circle,draw,label=45:$x_{ijk}$] {};
      \draw (x) to (x1);
      \draw (x) to (x2);
      \draw (x) to (x3);
      \draw [ultra thick] (x1) to (x2);
      \draw [ultra thick] (x1) to (x3);
      \draw [ultra thick] (x1) to (x12);
      \draw [ultra thick] (x1) to (x13);
%      \draw  (x2) to (x3);
      \draw  (x2) to (x12);
%      \draw  (x2) to (x23);
      \draw (x3) to (x13);
%      \draw (x3) to (x23);
%      \draw (x12) to (x13);
      \draw (x12) to (x123);
%      \draw  (x12) to (x23);
%      \draw  (x13) to (x23);
      \draw (x13) to (x123);
%      \draw (x23) to (x123);
   \end{tikzpicture}
   }
\;
\subfloat[]{\label{fig Eij}
   \begin{tikzpicture}[scale=0.85,inner sep=2]
      \node (x) at (0,0) [circle,draw,label=-135:$x$] {};
      \node (x1) at (3,0) [circle,fill,label=-45:$x_{i}$] {};
      \node (x2) at (1,1) [circle,fill,label=-90:$x_{j}$] {};
%      \node (x3) at (0,3) [circle,fill,label=135:$x_{k}$] {};
      \node (x12) at (4,1) [circle,fill,label=0:$x_{ij}$] {};
      \node (x13) at (3,3) [circle,fill,label=0:$x_{ik}$] {};
      \node (x23) at (1,4) [circle,fill,label=135:$x_{jk}$] {};
      \node (x123) at (4,4) [circle,draw,label=45:$x_{ijk}$] {};
      \draw (x) to (x1);
      \draw (x) to (x2);
%      \draw (x) to (x3);
%      \draw [ultra thick] (x1) to (x2);
%      \draw [ultra thick] (x1) to (x3);
      \draw [ultra thick] (x1) to (x12);
      \draw (x1) to (x13);
%      \draw [ultra thick] (x2) to (x3);
      \draw [ultra thick] (x2) to (x12);
      \draw (x2) to (x23);
%      \draw [ultra thick] (x3) to (x13);
%      \draw [ultra thick] (x3) to (x23);
      \draw [ultra thick] (x12) to (x13);
      \draw (x12) to (x123);
      \draw [ultra thick] (x12) to (x23);
%      \draw [ultra thick] (x13) to (x23);
      \draw (x13) to (x123);
      \draw (x23) to (x123);
   \end{tikzpicture}
   }
   \caption{\protect\subref{Fig: octahedron} Octahedron supporting $d\cL$ for a three-point discrete 2-form $\cL$;\\
   \protect\subref{fig Ei}  Stencil supporting 3D corner equation \eqref{eq: 3point Ei};
   \protect\subref{fig Eij} Stencil supporting corner equation \eqref{eq: 3point Eij}.}
   \label{Fig: 3-point case}
\end{figure}
%%%%%%%%%%%%%%%%%%%%%%%%%%%%%%%%%%%%%%%%%%%%%%%%%%%%%%%%%%%%%%%%%%

Accordingly, the system of corner equations consists of six equations per elementary 3D cube, which we denote by $({\mathcal E}_i)$, $({\mathcal E}_j)$, $({\mathcal E}_k)$, $({\mathcal E}_{ij})$, $({\mathcal E}_{ik})$, and $({\mathcal E}_{jk})$. To write them down, we set
\begin{equation}\label{eq: phi psi}
    \psi_i(\xi)=\Psi'_i(\xi),\quad  \phi_{ij}(\xi)=\Phi'_{ij}(\xi).
\end{equation}
In particular, we have: $\phi_{ij}(\xi)=\phi_{ji}(-\xi)$. In terms of these functions, corner equations read:
\begin{align}\label{eq: 3point Ei}\tag{${\mathcal E}_i$}
&\psi_j(x_{ij}-x_i)+\phi_{ij}(x_j-x_i)=\psi_k(x_{ik}-x_i)+\phi_{ik}(x_k-x_i),\\
\label{eq: 3point Eij}\tag{${\mathcal E}_{ij}$}
&\psi_j(x_{ij}-x_i)+\phi_{kj}(x_{ij}-x_{ik})=\psi_i(x_{ij}-x_j)+\phi_{ki}(x_{ij}-x_{jk}).
\end{align}
They can be characterized as 5-point 4-leg equations, see Figure \ref{Fig: 3-point case}\subref{fig Ei} and \subref{fig Eij}. The consistency of the system of 3D corner equations is defined literally as in Definition 
\ref{def: corner eqs consist}. 

3D corner equations for three-point 2-forms are elementary building blocks for discrete Laplace type equations on the regular triangular lattice, like all symplectic realizations of the discrete time relativistic Toda lattices. This is illustrated on Figure \ref{fig: from corner eqs to 7-point}.

%%%%%%%%%%%%%%%%%%%%%%%%%%%%%%%%%%%%%%%%%%%%%%%%%%%%%%%%%%%%%%%%%%
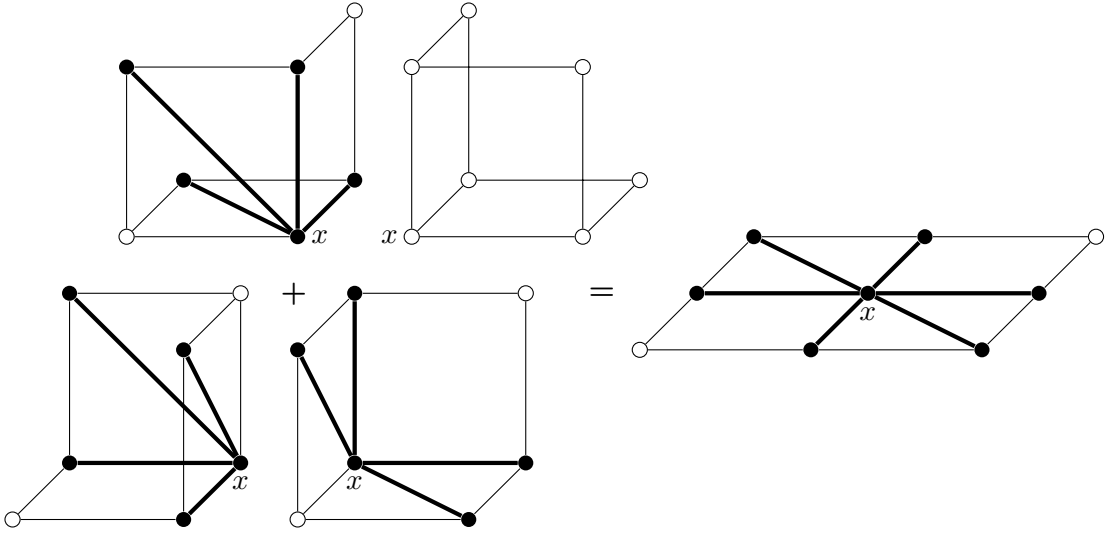
\begin{figure}[htbp]
   \centering
\begin{tikzpicture}[auto,scale=0.15,inner sep=2,>=stealth']
%
% left bottom cube
%
 \node (x 1) at (0,0) [circle,draw] {};
 \node (x1 1) at (15,0) [circle,fill] {};
 \node (x13 1) at (15,15) [circle,fill] {};
 \node (x2 1) at (5,5) [circle,fill] {};
 \node (x23 1) at (5,20) [circle,fill] {};
  \node (x123 1) at (20,20) [circle,draw] {};
  \node (x12 1) at (20,5) [circle,fill,label=-90:$x$] {};
 \draw [ultra thick] (x13 1) to (x12 1) to (x1 1); 
 \draw [ultra thick]  (x23 1) to (x12 1) to (x2 1);
 \draw [thin] (x 1) to (x2 1) to (x23 1) to (x12 1) to (x123 1);
  \draw [thin] (x 1) to (x1 1) to (x13 1) to (x123 1) to (x23 1);
%%
%% left top cube
%
 \node (x 2) at (10,25) [circle, draw]{};
 \node (x1 2) at (25,25) [circle, fill,label=0:$x$]{};
 \node (x3 2) at (10,40)  [circle, fill]{};
 \node (x13 2) at (25,40)[circle, fill]{};
 \node (x123 2) at (30,45)[circle, draw]{};
 \node (x2 2) at (15,30)[circle, fill]{};
 \node (x12 2) at (30,30)[circle, fill]{};
 \draw [ultra thick] (x3 2) to (x1 2) to (x13 2) ;
 \draw [ultra thick] (x2 2) to (x1 2) to (x12 2) ;
\draw [thin] (x 2) to (x2 2) to (x12 2) to (x1 2) to (x 2) to (x3 2) to (x13 2) to (x123 2) to (x12 2); 
%%
%% right bottom cube
%%
 \node (x 3) at (25,0) [circle,draw]{};
 \node (x1 3) at (40,0) [circle,fill]{};
 \node (x3 3) at (25,15) [circle,fill]{};
 \node (x23 3) at (30,20) [circle,fill]{};
 \node (x123 3) at (45,20) [circle,draw]{};
 \node (x2 3) at (30,5) [circle,fill,label=-90:$x$]{};
 \node (x12 3) at (45,5) [circle,fill]{};
\draw [ultra thick] (x1 3) to (x2 3) to (x12 3);
\draw [ultra thick] (x3 3) to (x2 3) to (x23 3);
\draw [thin] (x 3) to (x2 3) to (x12 3) to (x1 3) to (x 3) to (x3 3) to (x23 3) to (x123 3) to (x12 3); 
%%
%% right top cube
%%
 \node (x 4) at (35,25)  [circle,draw,label=180:$x$]{};
 \node (x1 4) at (50,25) [circle,draw]{};
 \node (x3 4) at (35,40) [circle,draw]{};
 \node (x13 4) at (50,40) [circle,draw]{};
 \node (x23 4)  at (40,45) [circle,draw]{};
 \node (x2 4) at (40,30) [circle,draw]{};
 \node (x12 4) at (55,30) [circle,draw]{};
\draw [thin]  (x2 4) to (x 4) to (x1 4);
\draw [thin] (x 4) to (x3 4) to (x13 4) to (x1 4) to (x12 4) to (x2 4) to (x23 4) to (x3 4);
\draw (23,20) node [right]{$\boldsymbol{+}$};
%
% 4x4 square
%
\node (x 5) at (55,15) [circle,draw]{};
\node (x1 5) at (70,15) [circle,fill]{};
\node (x11 5) at (85,15) [circle,fill]{};
\node (x2 5) at (60,20) [circle,fill]{};
\node (x12 5) at (75,20) [circle,fill,label=-90:$x$]{}; 
\node (x112 5) at (90,20) [circle,fill]{};
\node (x22 5) at (65,25) [circle,fill]{};
\node (x122 5) at (80,25) [circle,fill]{};
\node (x1122 5) at (95,25) [circle,draw]{};
\draw [thin]  (x 5) to (x1 5) to (x11 5) to (x112 5) to (x1122 5) to (x122 5) to (x22 5) to (x2 5) to (x 5);
\draw [ultra thick] (x12 5) to (x112 5);
\draw [ultra thick] (x12 5) to (x122 5);
\draw [ultra thick] (x12 5) to (x2 5);
\draw [ultra thick] (x12 5) to (x1 5);
\draw [ultra thick] (x12 5) to (x22 5);
\draw [ultra thick] (x12 5) to (x11 5);
\draw (50,20) node[right]{$\boldsymbol{=}$};
\end{tikzpicture}
   \caption{Sum of four 3D corner equations (matched at the vertex $x$, one of the equations being void) results in a planar seven-point Laplace type equation on the regular triangular lattice}
   \label{fig: from corner eqs to 7-point}
\end{figure}
%%%%%%%%%%%%%%%%%%%%%%%%%%%%%%%%%%%%%%%%%%%%%%%%%%%%%%%%%%%%%%%%%%

Now we turn to the case of one-parameter families of 3D corner equations, where we can obtain results generalizing to 3D the spectrality property of Theorem \ref{thm spectrality}. We fix the following framework. Suppose that one of the coordinate directions plays a special role (we denote this direction by ``0''). Assume that all other coordinate directions (denoted by $i$, $j$, etc.) correspond to certain instances of a parameter (denoted, respectively, by $\lambda$, $\mu$, etc.). Thus,
\begin{equation}\label{Lagr 2D param}
\Psi_i(\xi)=\Psi(\xi;\lambda), \quad \Phi_{ij}(\xi)=\Phi(\xi; \lambda,\mu), \quad \Phi_{i0}(\xi)=\Phi_0(\xi;\lambda),
\end{equation}
where $\Phi(\xi;\lambda,\mu)=-\Phi(-\xi;\mu,\lambda)$. Moreover, we will denote the shifts in the coordinate directions $i$, $j$ by tilde and by hat, respectively. The indices for the coordinate direction ``0'' will be denoted by $k$, and their shift will not be abbreviated. In this specific context, we can re-write expression \eqref{eq: 3point S} for $d\cL$ as follows:
\begin{equation}\label{eq: dL}
\begin{split}
d\cL=S^{ij0}& = \Psi(\wx_{k+1}-x_{k+1};\lambda)-\Psi(\whx_{k+1}-x_{k+1};\mu)-\Psi(\widehat{\wx}_k-\whx_k;\lambda)+\Psi(\widehat{\wx}_k-\wx_k;\mu)\\
       &\quad -\Psi_0(\wx_{k+1}-\wx_k) +\Psi_0(\whx_{k+1}-\whx_k)-\Phi(\whx_{k+1}-\wx_{k+1};\lambda,\mu)+\Phi(\whx_k-\wx_k;\lambda,\mu)\\
       &\quad +\Phi_{0}(\whx_{k+1}-\widehat{\wx}_k;\lambda)-\Phi_{0}(\wx_{k+1}-\widehat{\wx}_k;\mu)- \Phi_{0}(x_{k+1}-\wx_k;\lambda)+\Phi_{0}(x_{k+1}-\whx_k;\mu).\qquad
\end{split}
\end{equation}

\begin{theorem}\label{th: 2d conserv laws}
A three-point discrete 2-form $\cL$ with the discrete edge Lagrangians \eqref{Lagr 2D param} is closed on solutions of the system of 3D corner equations if and only if the latter system admits the conservation law
\begin{equation}\label{eq: 2d conserv law}
    \Delta_j P_{i0}=\Delta_0 P_{ij},
\end{equation}
with the densities
\begin{align}
    &P_{i0} = \frac{\partial \Psi(\wx_k-x_k;\lambda)}{\partial\lambda}-
    \frac{\partial\Phi_{0}(x_{k+1}-\wx_k;\lambda)}{\partial\lambda},
    \label{eq: 2d conserv law dens i0}\\
    &P_{ij} = \frac{\partial \Psi(\wx_k-x_k;\lambda)}{\partial\lambda}-
    \frac{\partial\Phi(\whx_k-\wx_k;\lambda,\mu)}{\partial\lambda}.
    \label{eq: 2d conserv law dens ij}
\end{align}
\end{theorem}
\begin{proof} According to Theorem \ref{Th: almost closed 2d}, quantity \eqref{eq: dL} is constant on solutions of the system of 3D corner equations: $d\cL=\ell(\lambda,\mu)$. This constant is obviously skew-symmetric: $\ell(\mu,\lambda)=-\ell(\lambda,\mu)$. Then $\ell(\lambda,\mu)=0$ is equivalent to $\partial \ell/\partial\lambda=0$. Differentiating equation \eqref{eq: dL} with respect to $\lambda$ and taking into account that the terms containing $\partial \wx_k/\partial\lambda$ etc., appearing due to the chain rule, vanish by virtue of the corresponding 3D corner equations, we arrive at
\begin{align*}
&\frac{\partial \Psi(\wx_{k+1}-x_{k+1};\lambda)}{\partial\lambda}
-\frac{\partial \Psi(\widehat{\wx}_k-\whx_k;\lambda)}{\partial\lambda}
-\frac{\partial\Phi_{0}(x_{k+1}-\wx_k;\lambda)}{\partial\lambda}
 +\frac{\partial\Phi_{0}(\whx_{k+1}-\widehat{\wx}_k;\lambda)}{\partial\lambda} \\
&\qquad -\frac{\partial\Phi(\whx_{k+1}-\wx_{k+1};\lambda,\mu)}{\partial\lambda}
 +\frac{\partial\Phi(\whx_k-\wx_k;\lambda,\mu)}{\partial\lambda}
 \, = \,  0.
\end{align*}
This is equivalent to formula \eqref{eq: 2d conserv law}.
\end{proof}

\paragraph{Bibliographical remarks.} Discrete three-point 2-forms as in \eqref{eq: 3point} were introduced in \cite{LN09} as an ingenious device to generalize the action \eqref{action dRTL 2d} from $\mathbb Z^2$ to an arbitrary quad-surface in a multi-dimensional lattice. For several equations from the ABS list, it was shown in \cite{LN09} that solutions of quad-equations deliver critical points for the action functional over an arbitrary quad-surface in $\mathbb Z^m$, and that the critical value of action is invariant under local flips of the quad-surface. This paper pioneered the pluri-Lagrangian theory. In \cite{BS10}, a conceptual proof  of these facts has been given for all quad-equations of the ABS list. 
 A decisive step, shifting the focus from quad-equations to 3D corner equations as main objects of interest, was made in \cite{BPS15}. Our presentation follows that paper.

\section{3D corner equations for relativistic Toda type systems}
\label{sect: from pluri to rel}

Like in the previous section, we now consider the situation where one of the coordinate directions (which we denote as the 0\textsuperscript{th} one) plays a distinguished role. We will use the index $k$ for this coordinate direction only. It will enumerate the sites of relativistic Toda chains. Accordingly, we will only consider surfaces in $\mathbb Z^m$ which contain, along with any point, the whole line through this point parallel to the 0\textsuperscript{th} coordinate axis. One can call such surfaces \emph{cylindrical}. The set of values of $x$ along such a line, $x=\{x_k: k\in\mathbb Z\}$, or, upon a finite-dimensional reduction, $x=\{x_k: 1\le k\le N\}$, is an element of the configuration space $\mathcal X$ of the relativistic Toda lattice. For other coordinate directions (denoted by $i$, $j$), we will use tilde, resp. hat to denote the corresponding shifts. The shift in the 0\textsuperscript{th} coordinate direction will not be abbreviated.

\begin{definition}\label{def: Fi}
Consider a pluri-Lagrangian system with a three-point 2-form \eqref{eq: 3point}. 
The maps $F_i: T^*\mathcal X\to T^*\mathcal X$, $(x,p)\mapsto(\wx,\wip)$ are defined as the symplectic maps with the generating functions
\begin{equation}\label{eq: dRTL1 Lagr}
\Lambda_i(x,\wx)=\sum_{n=1}^N \Psi_i(\wx_k-x_k)-\sum_{n=1}^N \Psi_0(x_{k+1}-x_k)-\sum_{n=1}^N\Phi_{i0}(x_{k+1}-\wx_k),
\end{equation}
thus equations of motion for $F_i$ read:
\begin{equation}
F_i:\left\{\begin{array}{ll}
p_k=-\dfrac{\partial \Lambda_{i}}{\partial x_k} & =\ \psi_i(\wx_k-x_k)+\phi_{i0}(x_k-\wx_{k-1})-\psi_0(x_{k+1}-x_k)+\psi_0(x_k-x_{k-1}), \vspace{3pt}\\
\wip_k= \quad\dfrac{\partial \Lambda_{i}}{\partial \wx_k} & =\ \psi_i(\wx_k-x_k)+\phi_{i0}(x_{k+1}-\wx_k),
\end{array}\right.
\end{equation}
with the corresponding Euler-Lagrange equations
\begin{align}
& \psi_i(\wx_k-x_k)-\psi_i(x_k-\undertilde{x}_k) \nonumber\\
& \qquad = \psi_0(x_{k+1}-x_k)-\psi_0(x_k-x_{k-1})+\phi_{i0}(\undertilde{x}_{k+1}-x_k)-\phi_{i0}(x_k-\wx_{k-1}).
\end{align}
\end{definition}

The map $F_i$ corresponds to the edges $(x,\wx)=(x,x_i)$ of the $i$\textsuperscript{th} coordinate direction, to which the strip supporting $\Lambda_{i}$ projects along the 0\textsuperscript{th} coordinate axis. See the identifications of variables on Figure~\ref{Fig: dRTL Lagrangian +}. 

\begin{theorem}\label{thm 3D for Fi commute Fj}
If the system of 3D corner equations corresponding to a three-point 2-form \eqref{eq: 3point} is consistent, then the maps $F_i$, $F_j$ commute.
\end{theorem}
%%%%%%%%%%%%%%%%%%%%%%%%%%%%%%%%%%%%%%%%%%%%%%%%%%%%%%%%%%%%%%%%%%
\begin{figure}[tp]
   \centering
   \subfloat[Domain of the equation \eqref{eq: rtl E}]{\label{fig:6a}
   \begin{tikzpicture}[scale=0.6,inner sep=2]
      \node (x) at (0,0) [circle,draw,label=-90:$x_{k-1}$] {};
      \node (x1) at (3,0) [circle,fill,label=-90:$x_{k}$] {};
      \node (x2) at (1,1) [circle,fill,label=135:$\widehat{x}_{k-1}$] {};
      \node (x3) at (0,3) [circle,fill,label=90:$\widetilde{x}_{k-1}$] {};
      \node (x11) at (6,0) [circle,draw,label=-90:$x_{k+1}$] {};
      \node (x12) at (4,1) [circle,fill,label=90:$\widehat{x}_{k}$] {};
      \node (x13) at (3,3) [circle,fill,label=90:$\widetilde{x}_{k}$] {};
      \node (x112) at (7,1) [circle,draw,label=90:$\widehat{x}_{k+1}$] {};
      \node (x113) at (6,3) [circle,draw,label=90:$\widetilde{x}_{k+1}$] {};
      \node (y) at (9,0) [circle,fill,label=-90:$x$] {};
      \node (y2) at (10,1) [circle,fill,label=90:$\widehat{x}$] {};
      \node (y3) at (9,3) [circle,fill,label=90:$\widetilde{x}$] {};
      \draw (x) to (x1);
      \draw (x) to (x2);
      \draw (x) to (x3);
      \draw [ultra thick] (x1) to (x2);
      \draw [ultra thick] (x1) to (x3);
      \draw (x1) to (x11);
      \draw [ultra thick] (x1) to (x12);
      \draw [ultra thick] (x1) to (x13);
      \draw (x2) to (x12);
      \draw (x3) to (x13);
      \draw (x11) to (x112);
      \draw (x11) to (x113);
      \draw (x12) to (x112);
      \draw (x13) to (x113);
      \draw [ultra thick] (y2) to (y) to (y3);
      \node (x2) at (1,1) [circle,fill,label={135,fill=white}:$\widehat{x}_{k-1}$] {};
   \end{tikzpicture}
   }\qquad
   \subfloat[Domain of the equation \eqref{eq: rtl E1}]{\label{fig:6b}
   \begin{tikzpicture}[scale=0.6,inner sep=2]
      \node (x) at (0,0) [circle,draw,label=-90:$x_{k-1}$] {};
      \node (x1) at (3,0) [circle,fill,label=-90:$x_{k}$] {};
      \node (x3) at (0,3) [circle,fill,label=-135:$\widetilde{x}_{k-1}$] {};
      \node (x11) at (6,0) [circle,fill,label=-90:$x_{k+1}$] {};
      \node (x13) at (3,3) [circle,fill,label=-135:$\widetilde{x}_{k}$] {};
      \node (x23) at (1,4) [circle,fill,label=90:$\htilde{x}_{k-1}$] {};
      \node (x113) at (6,3) [circle,fill,label=-45:$\widetilde{x}_{k+1}$] {};
      \node (x123) at (4,4) [circle,fill,label=90:$\htilde{x}_{k}$] {};
      \node (x1123) at (7,4) [circle,draw,label=90:$\htilde{x}_{k+1}$] {};
      \node (y) at (9,0) [circle,fill,label=-90:$x$] {};
      \node (y23) at (10,4) [circle,fill,label=90:$\htilde{x}$] {};
      \node (y3) at (9,3) [circle,fill,label=-45:$\widetilde{x}$] {};
      \draw (x) to (x1);
      \draw (x) to (x1);
      \draw (x) to (x3);
      \draw (x1) to (x11);
      \draw [ultra thick] (x1) to (x13);
      \draw [ultra thick] (x3) to (x13);
      \draw (x3) to (x23);
      \draw [ultra thick] (x11) to (x13);
      \draw (x11) to (x113);
      \draw [ultra thick] (x13) to (x23);
      \draw [ultra thick] (x13) to (x113);
      \draw [ultra thick] (x13) to (x123);
      \draw (x23) to (x123);
      \draw (x113) to (x1123);
      \draw (x123) to (x1123);
      \draw [ultra thick] (y23) to (y3) to (y);
   \end{tikzpicture}
   }\\
   \subfloat[Domain of the equation \eqref{eq: rtl E2}]{\label{fig:6c}
   \begin{tikzpicture}[scale=0.6,inner sep=2]
      \node (x) at (0,0) [circle,draw,label=-90:$x_{k-1}$] {};
      \node (x1) at (3,0) [circle,fill,label=-90:$x_{k}$] {};
      \node (x2) at (1,1) [circle,fill,label=135:$\widehat{x}_{k-1}$] {};
      \node (x11) at (6,0) [circle,fill,label=-90:$x_{k+1}$] {};
      \node (x12) at (4,1) [circle,fill,label=45:$\widehat{x}_{k}$] {};
      \node (x23) at (1,4) [circle,fill,label=90:$\htilde{x}_{k-1}$] {};
      \node (x112) at (7,1) [circle,fill,label=45:$\widehat{x}_{k+1}$] {};
      \node (x123) at (4,4) [circle,fill,label=90:$\htilde{x}_{k}$] {};
      \node (x1123) at (7,4) [circle,draw,label=90:$\htilde{x}_{k+1}$] {};
      \node (y) at (9,0) [circle,fill,label=-90:$x$] {};
      \node (y2) at (10,1) [circle,fill,label=45:$\widehat{x}$] {};
      \node (y23) at (10,4) [circle,fill,label=90:$\htilde{x}$] {};
      \draw (x) to (x1);
      \draw (x) to (x2);
      \draw (x1) to (x11);
      \draw [ultra thick] (x1) to (x12);
      \draw [ultra thick] (x2) to (x12);
      \draw (x2) to (x23);
      \draw [ultra thick] (x11) to (x12);
      \draw (x11) to (x112);
      \draw [ultra thick] (x12) to (x23);
      \draw [ultra thick] (x12) to (x112);
      \draw [ultra thick] (x12) to (x123);
      \draw (x23) to (x123);
      \draw (x112) to (x1123);
      \draw (x123) to (x1123);
      \draw [ultra thick] (y23) to (y2) to (y);
   \end{tikzpicture}
   }\qquad
      \subfloat[Domain of the equation \eqref{eq: rtl E12}]{\label{fig:6d}
   \begin{tikzpicture}[scale=0.6,inner sep=2]
      \node (x2) at (1,1) [circle,draw,label=-90:$\widehat{x}_{k-1}$] {};
      \node (x3) at (0,3) [circle,draw,label=-90:$\widetilde{x}_{k-1}$] {};
      \node (x12) at (4,1) [circle,fill,label=-90:$\widehat{x}_{k}$] {};
      \node (x13) at (3,3) [circle,fill,label=-90:$\widetilde{x}_{k}$] {};
      \node (x23) at (1,4) [circle,draw,label=90:$\htilde{x}_{k-1}$] {};
      \node (x112) at (7,1) [circle,fill,label=-90:$\widehat{x}_{k+1}$] {};
      \node (x113) at (6,3) [circle,fill,label=-45:$\widetilde{x}_{k+1}$] {};
      \node (x123) at (4,4) [circle,fill,label=90:$\htilde{x}_{k}$] {};
      \node (x1123) at (7,4) [circle,draw,label=90:$\htilde{x}_{k+1}$] {};
      \node (y23) at (10,4) [circle,fill,label=90:$\htilde{x}$] {};
      \node (y2) at (10,1) [circle,fill,label=-90:$\widehat{x}$] {};
      \node (y3) at (9,3) [circle,fill,label=-90:$\widetilde{x}$] {};
      \draw (x2) to (x12);
      \draw (x2) to (x23);
      \draw (x3) to (x13);
      \draw (x3) to (x23);
      \draw (x12) to (x112);
      \draw [ultra thick] (x12) to (x123);
      \draw (x13) to (x113);
      \draw [ultra thick] (x13) to (x123);
      \draw (x23) to (x123);
      \draw [ultra thick] (x112) to (x123);
      \draw (x112) to (x1123);
      \draw [ultra thick] (x113) to (x123);
      \draw (x113) to (x1123);
      \draw (x123) to (x1123);
      \draw [ultra thick] (y2) to (y23) to (y3);
      \node (x113) at (6,3) [circle,fill,label={-45,fill=white}:$\widetilde{x}_{k+1}$] {};
   \end{tikzpicture}
   }
   \caption{2D~corner equations for the system of $F_i$ and $F_j$ interpreted as 3D corner equations: equations $(E)$ and $(E_{ij})$ are 3D~corner equations of the corresponding three-point 2-form, while each of the equations $(E_i)$ and $(E_j)$ is a sum of two 3D~corner equations coming from two neighboring 3D~corners sharing a face.}
   \label{fig: 2d for Fi, Fj}
\end{figure}
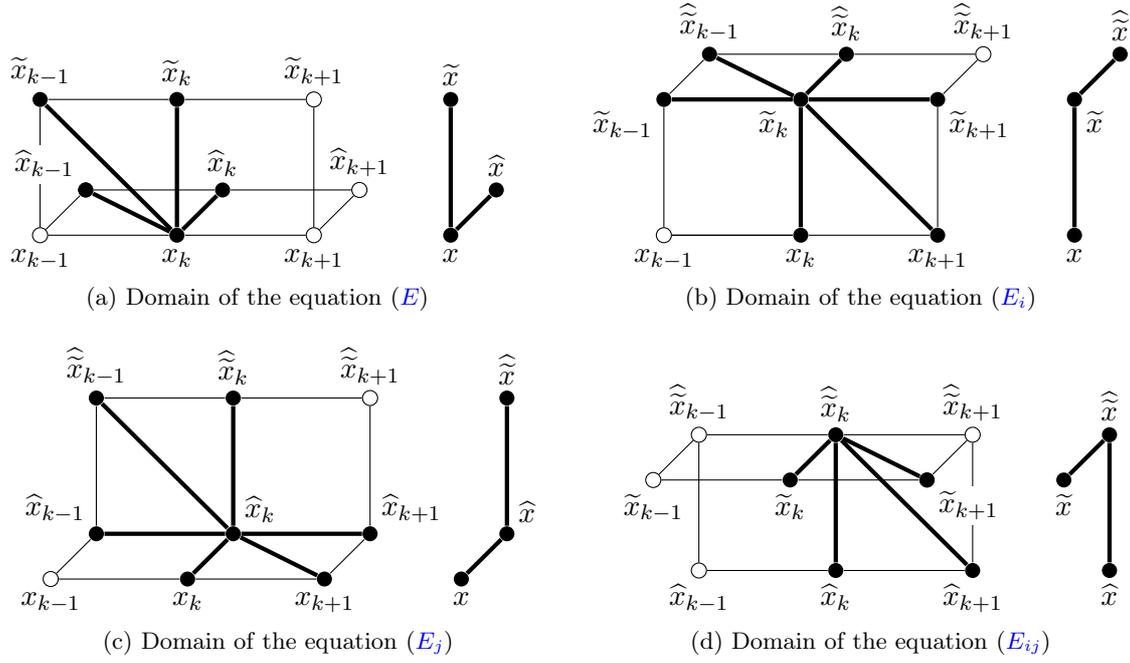
%%%%%%%%%%%%%%%%%%%%%%%%%%%%%%%%%%%%%%%%%%%%%%%%%%%%%%%%%%%%%%%%%%

\begin{proof} According to Theorem \ref{thm 1d commute}, commutativity of $F_i$ and $F_j$  is equivalent to the consistency of the system of 2D corner equations
for the corresponding one-dimensional pluri-Lagrangian system with the multi-time $\bbZ^2$. This system reads:
\begin{align}
\psi_i(\wx_k-x_k)+\phi_{i0}(x_k-\wx_{k-1})= & \ 
\psi_j(\whx_k-x_k)+\phi_{j0}(x_k-\whx_{k-1}),
\label{eq: rtl E}\tag{$E$} \\
\psi_i(\wx_k-x_k)+\phi_{i0}(x_{k+1}-\wx_k)= & \ 
\psi_j(\widehat{\wx}_k-\wx_k)+\phi_{j0}(\wx_k-\widehat{\wx}_{k-1}) \nonumber\\
&\qquad -\psi_0(\wx_{k+1}-\wx_k)+\psi_0(\wx_k-\wx_{k-1}),
\label{eq: rtl E1}\tag{$E_i$} \\
\psi_j(\whx_k-x_k)+\phi_{j0}(x_{k+1}-\whx_k)= & \ 
\psi_i(\widehat{\wx}_k-\whx_k)+\phi_{i0}(\whx_k-\widehat{\wx}_{k-1}) \nonumber\\
&\qquad-\psi_0(\whx_{k+1}-\whx_k)+\psi_0(\whx_k-\whx_{k-1}),
\label{eq: rtl E2}\tag{$E_j$}\\
\psi_i(\widehat{\wx}_k-\whx_k)+\phi_{i0}(\whx_{k+1}-\widehat{\wx}_k)= & \
\psi_j(\widehat{\wx}_k-\wx_k)+\phi_{j0}(\wx_{k+1}-\widehat{\wx}_k).
\label{eq: rtl E12}\tag{$E_{ij}$}
\end{align}
A visualization of the 2D~corner equations embedded in $\bbZ^{3}$ is given in Figure~\ref{fig: 2d for Fi, Fj}. Discrete curves in the multi-time plane $\bbZ^{2}$ (including the simplest such curves, the 2D corners themselves) are in a one-to-one correspondence with cylindrical surfaces in $\bbZ^3$, via the projection along the 0\textsuperscript{th} coordinate direction of $\bbZ^3$. 

Consistency of system \eqref{eq: rtl E}--\eqref{eq: rtl E12}  is proved with the help of the following statement.
\begin{theorem}\label{th: rtl+}
Suppose that the system of 3D corner equations corresponding to a three-point 2-form \eqref{eq: 3point} is consistent. Let the fields $x$, $\wx$, and $\whx$ satisfy 2D~corner equations \eqref{eq: rtl E}. Define the fields $\widehat{\wx}$ by any of the following four formulas, which are equivalent by virtue of \eqref{eq: rtl E}:
\begin{align}
\label{eq: rtl S1}\tag{$S1a$}
& \psi_j(\widehat{\wx}_k-\wx_k)+\phi_{ij}(\whx_k-\wx_k)=
\psi_0(\wx_{k+1}-\wx_k)+\phi_{i0}(x_{k+1}-\wx_k), \\
\label{eq: rtl S2}\tag{$S1b$}
& \psi_i(\widehat{\wx}_k-\whx_k)+\phi_{ji}(\wx_k-\whx_k)=
\psi_0(\whx_{k+1}-\whx_k)+\phi_{j0}(x_{k+1}-\whx_k), \\
\label{eq: rtl S3}\tag{$S2a$}
& \psi_i(\wx_{k+1}-x_{k+1})+\phi_{ji}(\wx_{k+1}-\whx_{k+1})=
\psi_0(\wx_{k+1}-\wx_k)+\phi_{j0}(\wx_{k+1}-\widehat{\wx}_k), \\
\label{eq: rtl S4}\tag{$S2b$}
& \psi_j(\whx_{k+1}-x_{k+1})+\phi_{ij}(\whx_{k+1}-\wx_{k+1})=
\psi_0(\whx_{k+1}-\whx_k)+\phi_{i0}(\whx_{k+1}-\widehat{\wx}_k),
\end{align}
called {\em superposition formulae} (note that each one of these formulas is local with respect to $\widehat{\wx}$). Then the 2D~corner equations \eqref{eq: rtl E1}, \eqref{eq: rtl E2}, and \eqref{eq: rtl E12} are satisfied, as well.
\end{theorem}
%%%%%%%%%%%%%%%%%%%%%%%%%%%%%%%%%%%%%%%%%%%%%%%%%%%%%%%%%%%%%%%%%%
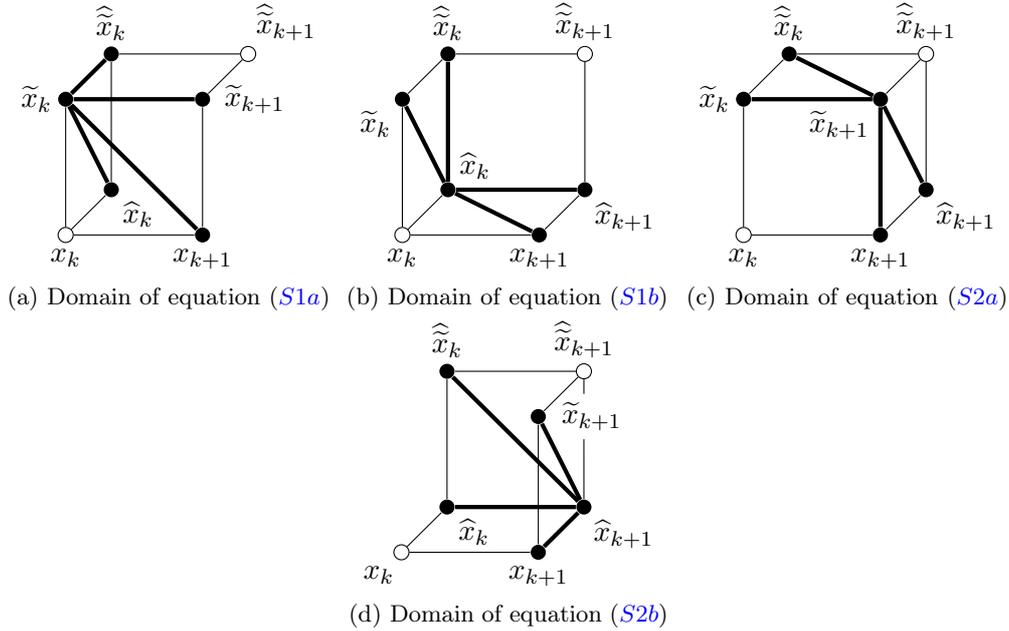
\begin{figure}[htb]
   \centering
   \subfloat[Domain of equation \eqref{eq: rtl S1}]{\label{fig: S1}
   \begin{tikzpicture}[scale=0.6,inner sep=2]
      \node (x) at (0,0) [circle,draw,label=-90:$x_{k}$] {};
      \node (x1) at (3,0) [circle,fill,label=-90:$x_{k+1}$] {};
      \node (x2) at (1,1) [circle,fill,label=-45:$\widehat{x}_{k}$] {};
      \node (x3) at (0,3) [circle,fill,label=180:$\widetilde{x}_{k}$] {};
%      \node (x12) at (4,1) [circle,fill,label=-60:$\widehat{x}_{k+1}$] {};
      \node (x13) at (3,3) [circle,fill,label=0:$\;\widetilde{x}_{k+1}$] {};
      \node (x23) at (1,4) [circle,fill,label=90:$\htilde{x}_{k}$] {};
      \node (x123) at (4,4) [circle,draw,label=75:$\htilde{x}_{k+1}$] {};
      \draw (x) to (x2);
      \draw (x) to (x3);
      \draw (x) to (x1);
      \draw (x1) to (x13);
      \draw (x2) to (x23);
      \draw (x23) to (x123);
      \draw (x13) to (x123);
      \draw [ultra thick] (x1) to (x3);
       \draw [ultra thick] (x2) to (x3);
      \draw [ultra thick] (x3) to (x13);
      \draw [ultra thick] (x3) to (x23);
   \end{tikzpicture}
   }\,
   \subfloat[Domain of equation \eqref{eq: rtl S2}]{\label{fig: S2}
   \begin{tikzpicture}[scale=0.6,inner sep=2]
      \node (x) at (0,0) [circle,draw,label=-90:$x_{k}$] {};
      \node (x1) at (3,0) [circle,fill,label=-90:$x_{k+1}$] {};
      \node (x2) at (1,1) [circle,fill,label=45:$\widehat{x}_{k}$] {};
      \node (x3) at (0,3) [circle,fill,label=-135:$\widetilde{x}_{k}$] {};
      \node (x12) at (4,1) [circle,fill,label=-60:$\widehat{x}_{k+1}$] {};
%      \node (x13) at (3,3) [circle,fill,label=-135:$\widetilde{x}_{k+1}$] {};
      \node (x23) at (1,4) [circle,fill,label=90:$\htilde{x}_{k}$] {};
      \node (x123) at (4,4) [circle,draw,label=90:$\htilde{x}_{k+1}$] {};
      \draw (x) to (x1);
      \draw (x) to (x2);
      \draw (x) to (x3);
      \draw (x3) to (x23);
      \draw (x1) to (x12);
      \draw (x12) to (x123);
      \draw (x23) to (x123);
      \draw [ultra thick] (x2) to (x12);
      \draw [ultra thick] (x2) to (x23);
      \draw [ultra thick] (x2) to (x3);
      \draw [ultra thick] (x2) to (x1);
     \end{tikzpicture}
   }\,
   \subfloat[Domain of equation \eqref{eq: rtl S3}]{\label{fig: S3}
   \begin{tikzpicture}[scale=0.6,inner sep=2]
       \node (x) at (0,0) [circle,draw,label=-90:$x_{k}$] {};
      \node (x1) at (3,0) [circle,fill,label=-90:$x_{k+1}$] {};
%      \node (x2) at (1,1) [circle,fill,label=-45:$\widehat{x}_{k}$] {};
      \node (x3) at (0,3) [circle,fill,label=180:$\widetilde{x}_{k}$] {};
      \node (x12) at (4,1) [circle,fill,label=-60:$\widehat{x}_{k+1}$] {};
      \node (x13) at (3,3) [circle,fill,label=-135:$\widetilde{x}_{k+1}$] {};
      \node (x23) at (1,4) [circle,fill,label=90:$\htilde{x}_{k}$] {};
      \node (x123) at (4,4) [circle,draw,label=90:$\htilde{x}_{k+1}$] {};
      \draw (x) to (x3);
      \draw (x) to (x1);
      \draw (x1) to (x12);
      \draw (x3) to (x23);
      \draw (x12) to (x123);
      \draw (x13) to (x123);
      \draw (x23) to (x123);
      \draw [ultra thick] (x1) to (x13);
       \draw [ultra thick] (x3) to (x13);
      \draw [ultra thick] (x23) to (x13);
      \draw [ultra thick] (x12) to (x13);
   \end{tikzpicture}
   }\,
       \subfloat[Domain of equation \eqref{eq: rtl S4}]{\label{fig: S4}
   \begin{tikzpicture}[scale=0.6,inner sep=2]
       \node (x) at (0,0) [circle,draw,label=-100:$x_{k}$] {};
      \node (x1) at (3,0) [circle,fill,label=-90:$x_{k+1}$] {};
      \node (x2) at (1,1) [circle,fill,label=-45:$\widehat{x}_{k}$] {};
%      \node (x3) at (0,3) [circle,fill,label=180:$\widetilde{x}_{k}$] {};
      \node (x12) at (4,1) [circle,fill,label=-60:$\widehat{x}_{k+1}$] {};
      \node (x13) at (3,3) [circle,fill,label=0:$\ \widetilde{x}_{k+1}$] {};
      \node (x23) at (1,4) [circle,fill,label=90:$\htilde{x}_{k}$] {};
      \node (x123) at (4,4) [circle,draw,label=90:$\htilde{x}_{k+1}$] {};
      \draw (x) to (x2);
      \draw (x) to (x1);
      \draw (x1) to (x13);
      \draw (x2) to (x23);
      \draw (x12) to (4,2.5); 
      \draw (4,3.5) to (x123);
      \draw (x13) to (x123);
      \draw (x23) to (x123);
      \draw [ultra thick] (x1) to (x12);
       \draw [ultra thick] (x2) to (x12);
      \draw [ultra thick] (x23) to (x12);
      \draw [ultra thick] (x12) to (x13);
   \end{tikzpicture}
   }
   \caption{Superposition formulas as 3D~corner equations of the corresponding three-point 2-form.}
   \label{fig: 2d for superposition formulas}
\end{figure}
%%%%%%%%%%%%%%%%%%%%%%%%%%%%%%%%%%%%%%%%%%%%%%%%%%%%%%%%%%%%%%%%%%
\begin{proof}
One easily checks that the two corner equations \eqref{eq: rtl E}$_{k\to k+1}$, \eqref{eq: rtl E12} and the four superposition formulae \eqref{eq: rtl S1}--\eqref{eq: rtl S4} build nothing but the system of six 3D~corner equations \eqref{eq: 3point Ei}, \eqref{eq: 3point Eij} within one elementary cube of $\bbZ^3$, see Figures \ref{fig: 2d for Fi, Fj}\subref{fig:6a},\subref{fig:6d} and \ref{fig: 2d for superposition formulas}\subref{fig: S1}--\subref{fig: S4}. Due to consistency of the latter system, as formulated in Theorem \ref{th: rtl+}, if equation \eqref{eq: rtl E} and one of equations \eqref{eq: rtl S1}--\eqref{eq: rtl S4} hold, then  equation \eqref{eq: rtl E12} and the remaining three of equations \eqref{eq: rtl S1}--\eqref{eq: rtl S4} are satisfied, as well. Furthermore, equation \eqref{eq: rtl E1} is the difference of \eqref{eq: rtl S1} and \eqref{eq: rtl S3}$_{k\to k-1}$, while equation \eqref{eq: rtl E2} is the difference of \eqref{eq: rtl S2} and \eqref{eq: rtl S4}$_{k\to k-1}$. This completes the proof.
\end{proof}
Theorem \ref{th: rtl+} provides a proof of Theorem \ref{thm 3D for Fi commute Fj} in the case of open-end boundary conditions, where the maps $F_i$ are well-defined. At the same time, it provides us with an exhaustive understanding of commutativity also in the case of periodic boundary conditions, where the maps $F_i$ are double-valued. In this case, each of the compositions $F_i\circ F_j$ and $F_j\circ F_i$ applied to a point $(x,p)$ produces four different branches for $(\widehat{\wx},\widehat{\wip})$. Commutativity is reflected in the following fact: each of the branches of $F_i\circ F_j$ coincides with one of the branches of $F_j\circ F_i$. Indeed, Theorem~\ref{th: rtl+} delivers four possible values for $(x,\wx,\whx,\widehat{\wx})$ satisfying all 2D~corner equations \eqref{eq: rtl E}--\eqref{eq: rtl E12}, namely one $\widehat{\wx}$ for each of the four possible combinations of $\left(x,\wx,\whx\right)$. Cf. Figure \ref{Fig: consistency 4-valued}.
\end{proof}

\paragraph{Example: B\"acklund transformations for the additive exponential relativistic Toda lattice.}
We consider system \eqref{dRTL+ l New} which is a discretization of (and a B\"acklund trabnsformation for) system \eqref{RTL+ l New}.
The corresponding maps $F_i:\bbR^{2N}\to\bbR^{2N}$ are given by  
\begin{equation}\label{eq: BT RTL l}
F_i:
\begin{cases}[2]
p_{k}=\dfrac{1}{\lambda}\big(\eto{\wx_{k}\nm x_{k}}-1\big)+\dfrac{(\lambda-\alpha)\eto{x_{k}\nm \wx_{k-1}}}{1-\lambda\alpha \eto{x_{k}\nm \wx_{k-1}}}-\alpha \eto{x_{k+1}\nm x_{k}}+\alpha \eto{x_{k}\nm x_{k-1}},\vspace{5pt}\\
\wip_{k}=\dfrac{1}{\lambda}\big(\eto{\wx_{k}\nm x_{k}}-1\big)+\dfrac{(\lambda-\alpha)\eto{x_{k+1}\nm \wx_{k}}}{1-\lambda\alpha \eto{x_{k+1}\nm \wx_{k}}}.
\end{cases}
\end{equation}
Thus,
\begin{equation}\label{BT RTL l psi}
\psi_i(\xi)=\psi(\xi;\lambda)=\frac{1}{\lambda}\big(\eto{\xi}-1\big), \quad \psi_0(\xi)=\alpha\eto{\xi}, \quad \phi_{i0}(\xi)=\phi_0(\xi;\lambda)=\dfrac{(\lambda-\alpha)\eto{\xi}}{1-\lambda\alpha \eto{\xi}}.
\end{equation}
One can show that, in order to obtain a consistent system of 3D corner equations, these leg functions have to be supplemented  by
\begin{equation}\label{BT RTL l phi}
\phi_{ij}(\xi)=\phi(\xi;\lambda,\mu)=\frac{\eto{\xi}-1}{\lambda\eto{\xi}-\mu}.
\end{equation}
The corresponding 2D~corner equations are given by:
\begin{align} 
\frac{1}{\lambda}\Big(\eto{\wx_{k}\nm x_{k}}-1\Big)+\dfrac{(\lambda-\alpha)\eto{x_{k}\nm \wx_{k-1}}}{1-\lambda\alpha \eto{x_{k}\nm \wx_{k-1}}}= & \
\frac{1}{\mu}\Big(\eto{\whx_{k}\nm x_{k}}-1\Big)+\dfrac{(\mu-\alpha)\eto{x_{k}\nm \whx_{k-1}}}{1-\mu\alpha \eto{x_{k}\nm \whx_{k-1}}},
\label{eq: BT RTL l E}\tag{$E$}  \\ 
\dfrac{1}{\lambda}\Big(\eto{\wx_{k}\nm x_{k}}-1\Big)+\dfrac{(\lambda-\alpha)\eto{x_{k+1}\nm \wx_{k}}}{1-\lambda\alpha \eto{x_{k+1}\nm \wx_{k}}}= & \
\dfrac{1}{\mu}\Big(\eto{\htilde{x}_{k}\nm \wx_{k}}-1\Big)+\dfrac{(\mu-\alpha)\eto{\wx_{k}\nm \htilde{x}_{k-1}}}{1-\mu\alpha \eto{\wx_{k}\nm \htilde{x}_{k-1}}}
\nonumber \\
 & \qquad +\alpha \eto{\wx_{k}\nm \wx_{k-1}}-\alpha \eto{\wx_{k+1}-\wx_{k}},
\label{eq: BT RTL l E1}\tag{$E_i$} \\ 
\dfrac{1}{\mu}\Big(\eto{\whx_{k}\nm x_{k}}-1\Big)+\dfrac{(\mu-\alpha)\eto{x_{k+1}\nm \whx_{k}}}{1-\mu\alpha \eto{x_{k+1}\nm \whx_{k}}}= & \ 
\dfrac{1}{\lambda}\Big(\eto{\htilde{x}_{k}\nm \whx_{k}}-1\Big)+\dfrac{(\lambda-\alpha)\eto{\whx_{k}\nm \htilde{x}_{k-1}}}{1-\lambda\alpha \eto{\whx_{k}\nm \htilde{x}_{k-1}}}
\nonumber  \\
 & \qquad +\alpha \eto{\whx_{k}\nm \whx_{k-1}}-\alpha \eto{\whx_{k+1}-\whx_{k}},
\label{eq: BT RTL l E2}\tag{$E_j$} \\  
\dfrac{1}{\lambda}\Big(\eto{\htilde{x}_{k}\nm \whx_{k}}-1\Big)+\dfrac{(\lambda-\alpha)\eto{\whx_{k+1}\nm \htilde{x}_{k}}}{1-\lambda\alpha \eto{\whx_{k+1}-\htilde{x}_{k}}}= & \
\dfrac{1}{\mu}\Big(\eto{\htilde{x}_{k}\nm \wx_{k}}-1\Big)+\dfrac{(\mu-\alpha)\eto{\wx_{k+1}\nm \htilde{x}_{k}}}{1-\mu\alpha \eto{\wx_{k+1}\nm \htilde{x}_{k}}}.
\label{eq: BT RTL l E12}\tag{$E_{ij}$}
\end{align}
while the superposition formulas are given by:
\begin{align}
&\dfrac{1}{\mu}\Big(\eto{\htilde{x}_{k}-\wx_{k}}-1\Big)+\frac{\eto{\whx_{k}}-\eto{\wx_{k}}}{\lambda \eto{\whx_{k}}-\mu \eto{\wx_{k}}}=
\alpha \eto{\wx_{k+1}\nm \wx_{k}}+\dfrac{(\lambda-\alpha)\eto{x_{k+1}\nm \wx_{k}}}{1-\lambda\alpha \eto{x_{k+1}\nm \wx_{k}}},
\label{eq: BT RTL l S1}\tag{$S1a$}   \\
&\frac{1}{\lambda}\Big(\eto{\htilde{x}_{k}\nm \whx_{k}}-1\Big)+\frac{\eto{\wx_{k}}-\eto{\whx_{k}}}{\mu \eto{\wx_{k}}-\lambda \eto{\whx_{k}}}=
\alpha \eto{\whx_{k+1}\nm \whx_{k}}+\frac{(\mu-\alpha)\eto{x_{k+1}\nm \whx_{k}}}{1-\mu\alpha \eto{x_{k+1}\nm \whx_{k}}},
\label{eq: BT RTL l S2}\tag{$S1b$}  \\
&\frac{1}{\lambda}\Big(\eto{\wx_{k+1}\nm {x}_{k+1}}-1\Big)+\frac{\eto{\wx_{k+1}}-\eto{\whx_{k+1}}}{\mu \eto{\wx_{k+1}}-\lambda \eto{\whx_{k+1}}}=
\alpha \eto{\wx_{k+1}\nm \wx_{k}}+\frac{(\mu-\alpha)\eto{\wx_{k+1}\nm \htilde{x}_{k}}}{1-\mu\alpha \eto{\wx_{k+1}\nm \htilde{x}_{k}}},
\label{eq: BT RTL l S3}\tag{$S2a$}\\ 
&\frac{1}{\mu}\Big(\eto{\whx_{k+1}\nm {x}_{k+1}}-1\Big)+\frac{\eto{\whx_{k+1}}-\eto{\wx_{k+1}}}{\lambda \eto{\whx_{k+1}}-\mu \eto{\wx_{k+1}}}=
\alpha \eto{\whx_{k+1}\nm \whx_{k}}+\frac{(\lambda-\alpha)\eto{\whx_{k+1}\nm \htilde{x}_{k}}}{1-\lambda\alpha \eto{\whx_{k+1}-\htilde{x}_{k}}}.
\label{eq: BT RTL l S4}\tag{$S2b$}
\end{align}
Of course, in order to apply our general results, one has to prove consistency of the system of 3D corner equations consisting of \eqref{eq: BT RTL l E}$_{k\to k+1}$, \eqref{eq: BT RTL l E12}, and \eqref{eq: BT RTL l S1}--\eqref{eq: BT RTL l S4}. For this, one shows by direct computations that any two of these six equations are equivalent by virtue of the following {\em octahedron equation}:
\begin{equation}\label{BT RTL l oct}
\frac{1}{\lambda}\, \eto{\wx_{k+1}\nm x_{k+1}}-\frac{1}{\mu}\, \eto{\whx_{k+1}\nm x_{k+1}}-\frac{1}{\lambda}\, \eto{\htilde{x}_k \nm \whx_k} +\frac{1}{\mu}\, \eto{\htilde{x}_k \nm \wx_k}
+\alpha\, \eto{\whx_{k+1}\nm \whx_k}-\alpha\, \eto{\wx_{k+1}\nm \wx_k}=0.
\end{equation}
Indeed, octahedron relation \eqref{BT RTL l oct} is an immediate consequence of \eqref{eq: BT RTL l E}$_{k\to k+1}$, \eqref{eq: BT RTL l S1} and \eqref{eq: BT RTL l S2} (or, alternatively, of \eqref{eq: BT RTL l E12}, \eqref{eq: BT RTL l S3} and \eqref{eq: BT RTL l S4}). On the other hand, eliminating from any of the corner equations one of the variables by means of the (multi-affine) octahedron equation, we obtain another corner equation. \qed
\begin{theorem}\label{th: RTL l closure}
The system of 3D corner equations consisting of \eqref{eq: BT RTL l E}$_{k\to k+1}$, \eqref{eq: BT RTL l E12}, and \eqref{eq: BT RTL l S1}--\eqref{eq: BT RTL l S4} admits the conservation law \eqref{eq: 2d conserv law}. Therefore, the discrete Lagrangian 2-form $\cL$ is closed on any solution of this system.
\end{theorem}
\begin{proof}
We apply Theorem  \ref{th: 2d conserv laws}. To do this, we compute from \eqref{BT RTL l psi}, \eqref{BT RTL l phi}:
\begin{alignat*}{3}
&\Psi(\xi;\lambda)=\frac{1}{\lambda}\big(\eto{\xi}-1-\xi\big)&&\;\Rightarrow\; 
\frac{\partial\Psi(\xi;\lambda)}{\partial\lambda}=-\frac{1}{\lambda^2}\big(\eto{\xi}-1-\xi\big)\\
& \Phi(\xi;\lambda,\mu)=\frac{1}{\mu}\xi+\frac{\mu-\lambda}{\lambda\mu}\log(\lambda\eto{\xi}-\mu)&&\; \Rightarrow\; 
\frac{\partial\Phi(\xi;\lambda,\mu)}{\partial\lambda}=-\frac{1}{\lambda\mu}+\frac{1}{\lambda}\cdot\frac{\eto{\xi}-1}{\lambda\eto{\xi}-\mu}-\frac{1}{\lambda^2}\log(\lambda\eto{\xi}-\mu),\\
&\Phi_0(\xi;\lambda)=-\frac{\lambda-\alpha}{\lambda\alpha}\log(1-\lambda\alpha\eto{\xi})&&\;\Rightarrow\; \frac{\partial\Phi_0(\xi;\lambda)}{\partial\lambda}=\frac{1}{\lambda}\cdot \frac{(\lambda-\alpha)\eto{\xi}}{1-\lambda\alpha\eto{\xi}}-\frac{1}{\lambda^2}\log(1-\lambda\alpha\eto{\xi}).
\end{alignat*}
As a consequence, we can write conservation law \eqref{eq: 2d conserv law} as 
\begin{equation}\label{eq: BT RTL l conserv law}
    \Delta_j (R_{i0}-\tfrac{1}{\lambda}S_{i0})=\Delta_0 (R_{ij}-\tfrac{1}{\lambda}S_{ij}),
\end{equation}
where
$$
R_{i0}=\frac{1}{\lambda}\big(\eto{\wx_k\nm x_k}-1\big)+\frac{(\lambda-\alpha)\eto{x_{k+1}\nm \wx_k}}{1-\lambda\alpha\eto{x_{k+1}\nm \wx_k}},
\quad
R_{ij}=\frac{1}{\lambda}\big(\eto{\wx_k\nm x_k}-1\big)+\frac{\eto{\whx_k}-\eto{\wx_k}}{\lambda\eto{\whx_k}-\mu\eto{\wx_k}},
$$
and
$$
S_{i0}=(\wx_k-x_k)+\log\big(1-\lambda\alpha\eto{x_{k+1}\nm\wx_k}\big),
\quad
S_{ij}=- x_k+\log\big(\lambda\eto{\whx_k}-\mu\eto{\wx_k}\big).
$$
We show that the two conservation laws  $\Delta_j R_{i0}=\Delta_0 R_{ij}$ and $\Delta_j S_{i0}=\Delta_0 S_{ij}$ are satisfied separately. Indeed, $\Delta_j R_{i0}=\Delta_0 R_{ij}$ is an immediate consequence of \eqref{eq: BT RTL l E12}, \eqref{eq: BT RTL l S1} and  \eqref{eq: BT RTL l S3}, while (the exponentiated form of) $\Delta_j S_{i0}=\Delta_0 S_{ij}$ reads
$$
\eto{\htilde{x}_k\nm \wx_k\nm \whx_k\np x_k} \cdot \frac{1-\lambda\alpha\eto{\whx_{k+1}\nm \htilde{x}_k}}{1-\lambda\alpha\eto{x_{k+1}\nm \wx_k}}=
\eto{-x_{k+1}\np x_k} \cdot \frac{\lambda\eto{\whx_{k+1}}-\mu\eto{\wx_{k+1}}}{\lambda\eto{\whx_k}-\mu\eto{\wx_k}},
$$
and is equivalent to \eqref{BT RTL l oct} (as a straightforward clearing of denominators shows).
\end{proof}

As a corollary, $\sum_{k=1}^N S_{i0}$ is a common conserved quantity for all $F_j$. Its exponentiated form is
\begin{equation}
P(x,\wx;\lambda)=\prod_{k=1}^N \gamma_k,\quad \mathrm{where} \quad \gamma_k=\eto{\wx_k\nm x_k}\Big(1-\lambda\alpha\eto{x_{k+1}\nm \wx_k}\Big).
\end{equation}
Like in the non-relativistic case, one can find a nice expression of this quantity through $(x,p)$.

\begin{theorem}\label{th: BT RTL l zcr}
Set
\begin{equation*}\label{eq: BT RTL l zcr L}
    L_k(x,p;\lambda)=\begin{pmatrix} 1+\lambda p_k-\lambda\alpha\eto{x_k\nm x_{k-1}} & -\lambda(\lambda-\alpha) \eto{x_k\nm x_{k-1}} \vspace{2mm}\\ 
                                     1 & 0 \end{pmatrix},
\end{equation*}
and
\begin{equation*}\label{eq: BT RTL l zcr T}
    T_N(x,p;\lambda)=L_N(x,p;\lambda) \cdots L_2(x,p;\lambda)L_1(x,p;\lambda).
\end{equation*}
Then in the periodic case the quantity $P(x,\wx;\lambda)$ is an eigenvalue of $T_N(x,p;\lambda)$, while in the open-end case it is equal to $\tr T_N(x,p;\lambda)$.
\end{theorem}
\begin{proof} The first equation in \eqref{eq: BT RTL l} can be put as
\[
\eto{\wx_k\nm x_k}\Big(1-\lambda\alpha\eto{x_{k+1}\nm \wx_k}\Big)=\Big(1+\lambda p_k-\lambda\alpha\eto{x_k\nm x_{k-1}}\Big)-
\frac{\lambda(\lambda-\alpha) \eto{x_k\nm x_{k-1}}}{\eto{\wx_{k-1}\nm x_{k-1}}\Big(1-\lambda\alpha\eto{x_k\nm\wx_{k-1}}\Big)}.
\]
This is equivalent to saying that
\[
L_k(x,p;\lambda)\begin{pmatrix} \gamma_{k-1} \\ 1 \end{pmatrix} \sim
\begin{pmatrix} \gamma_k \\ 1 \end{pmatrix}.
\]
From this point, the proof is literally the same as for Theorem \ref{th: BT Toda zcr}.
\end{proof}

\paragraph{Bibliographical remarks.} Our presentation here follows \cite{BPS15}. Similar results are obtained there for all relativistic Toda lattices listed in Section \ref{Sect Newtonian rel Toda}.

\section*{Acknowledgments}
This research was supported by the DFG Collaborative Research Center TRR 109 ``Discretization in Geometry and Dynamics''.

{\small
\bibliographystyle{amsalpha}
\providecommand{\bysame}{\leavevmode\hbox to3em{\hrulefill}\thinspace}
\providecommand{\MR}{\relax\ifhmode\unskip\space\fi MR }
\providecommand{\MRhref}[2]{%
  \href{http://www.ams.org/mathscinet-getitem?mr=#1}{#2}
}
\providecommand{\href}[2]{#2}

}

\end{document}